\documentclass[10pt,reqno]{article}
\usepackage{fullpage}

\usepackage[unicode=true]{hyperref}
\usepackage{amsmath}
\usepackage{cite}
\usepackage{amsfonts}
\usepackage{amssymb}
\usepackage{amsthm}
\usepackage{authblk}

\usepackage[pdftex]{color,graphicx}
\usepackage{mypack} 
\usepackage{adjustbox}
\usepackage{soul}
\usepackage{comment}
\usepackage{bbold}
\usepackage{tikz}
\usetikzlibrary{decorations.pathreplacing,angles,quotes}
\usepackage{bbold}

\usepackage{diagbox}
\usepackage{enumerate}

\usetikzlibrary{matrix,arrows,decorations.pathmorphing}
\usepackage{tikz-cd}
\usetikzlibrary{arrows} 

\usetikzlibrary{decorations.markings}
 \usepackage{caption}
\usepackage{subcaption}
 
 \usepackage{pdfpages}

\usepackage{blkarray}

\usepackage{centernot}
\usepackage{mathtools}
\usepackage{stmaryrd}

  \usepackage{arydshln}

\definecolor{bluegray}{rgb}{0.4, 0.6, 0.8}
\definecolor{turquoise}{rgb}{0.2, 0.7, 0.6}
\definecolor{hy-green}{rgb}{0.1, 0.5, 0.1}

\usepackage{soul}  
 \newcommand{\suchthat}{\;\ifnum\currentgrouptype=16 \middle\fi|\;}

\title{Mermin polytopes in quantum computation and foundations}

\author{Cihan Okay\footnote{cihan.okay@bilkent.edu.tr}}
\author{Ho Yiu Chung\footnote{hoyiu.chung@bilkent.edu.tr,
}}
\author{Selman Ipek\footnote{selman.ipek@bilkent.edu.tr}}
\affil{Department of Mathematics, Bilkent University, Ankara, Turkey}

\begin{document}
  \maketitle  
  
\begin{abstract}
Mermin square scenario provides a simple proof for state-independent contextuality.  
In this paper, we study polytopes $\MP_\beta$ obtained from the Mermin scenario, parametrized by a function $\beta$ on the set of contexts. 
Up to combinatorial isomorphism,
there are two types of polytopes $\MP_0$ and $\MP_1$ depending on the parity of $\beta$. Our main result is the classification of the vertices of these two polytopes. In addition, we describe the graph associated with the polytopes. 
All the vertices of $\MP_0$ turn out to be deterministic. This result provides a new topological proof of a celebrated result of Fine characterizing noncontextual distributions on the CHSH scenario. 
$\MP_1$ can be seen as a nonlocal toy version of $\Lambda$-polytopes, a class of polytopes introduced for the simulation of universal quantum computation. In the $2$-qubit case, we provide a decomposition of the $\Lambda$-polytope using $\MP_1$, whose vertices are classified, and the nonsignaling polytope of the $(2,3,2)$ Bell scenario, whose vertices are well-known. 
\end{abstract}
  
 \tableofcontents

\section{Introduction}\label{sec:intro}

Central to many of the paradoxes arising in quantum theory is that the act of measurement cannot be understood as merely revealing the pre-existing values of some hidden variables.\footnote{A classic counterexample to this viewpoint is the well-known de Broglie Bohm pilot wave theory \cite{bohm1952suggested}. For more modern approaches seeking to bypass these claims, see e.g., \cite{schmid2020unscrambling,caticha2022entropic}.} Instead, as shown by the `no-go' theorems of Bell \cite{bell1966problem}, and Kochen-Specker (KS) \cite{kochen1975problem}, the outcomes of quantum measurements depend crucially on what else they are being measured with, a phenomenon known as contextuality. (For a recent review, see e.g., \cite{budroni2021quantum}.) A particularly accessible illustration of this quantum mechanical feature using just two spin-$1/2$ particles was given some years ago by Mermin \cite{mermin1993hidden}, an example which is now commonly called Mermin's square.
This scenario, as illustrated in Fig.~(\ref{fig:mermin-scenario}), consists of $9$ measurements $M$ and $6$ contexts $\cC$ given by the rows and the columns of the square grid. Together with the function $\beta$ which assigns a value in $\ZZ_2=\set{0,1}$ to each context this scenario specifies a binary linear system \cite{cleve2014characterization}.
It is known that this binary system $(M,\cC,\beta)$ has a classical solution if and only if 
\begin{equation}\label{eq:beta-coho}
[\beta] = \sum_{C\in \cC} \beta(C) = 0 \mod 2.
\end{equation}
 \begin{figure}[h!] 
  \centering
  \includegraphics[width=.4\linewidth]{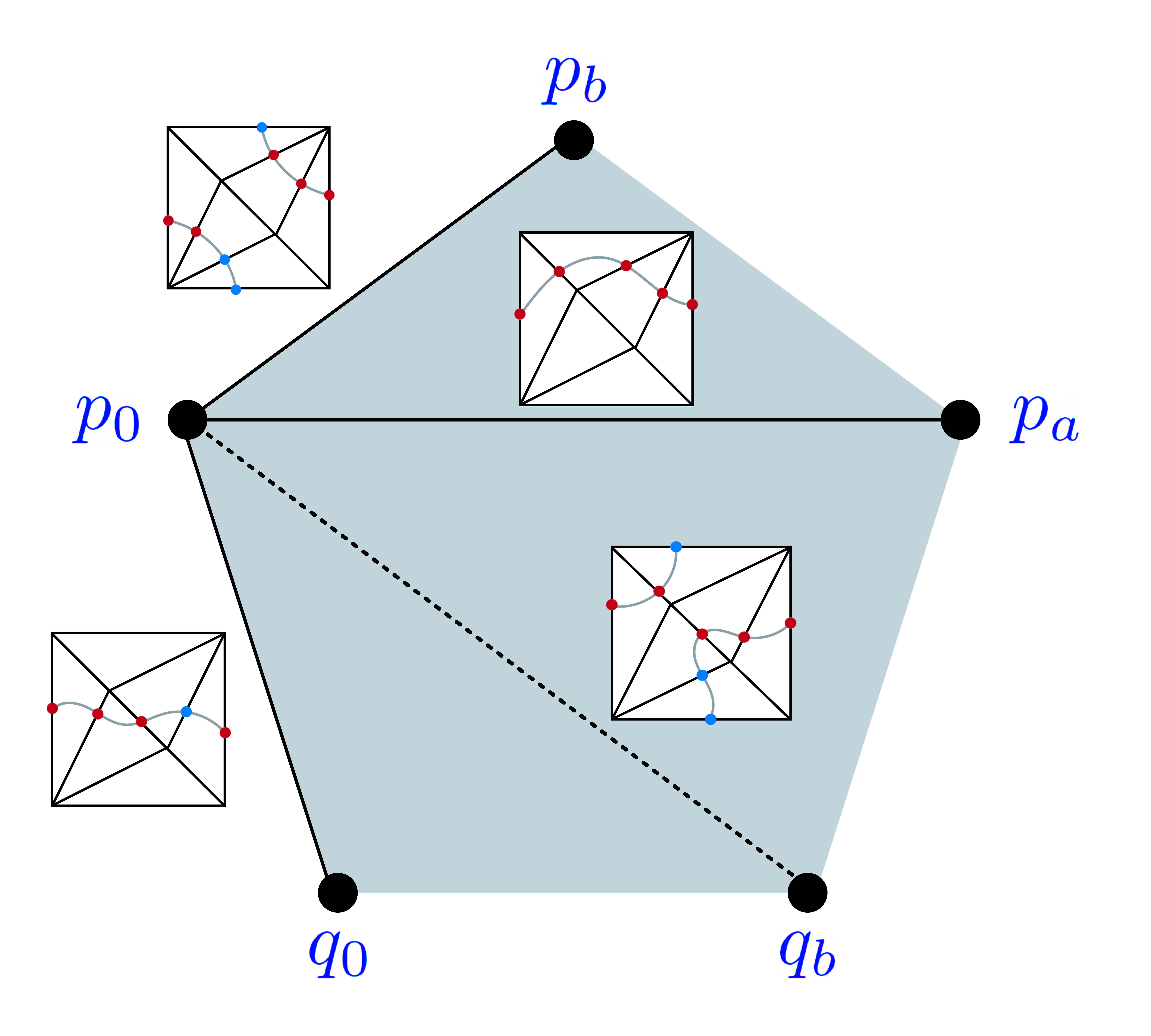}
\caption{Local structure $\MP_1$ at the type $1$ vertex $\VFS$ and type $2$ vertex $\VFe$. The former has only type $2$ neighbors, a single orbit under the action of the stabilizer of the vertex. The latter has both type $1$ (single orbit) and type $2$ neighbors (breaks into two orbits with representatives $\Vnn$ and $\Vtt$). Edges in the polytope are represented by loops on the Mermin torus. $\Vte$ can be connected to $\VFe$ by a path corresponding to a loop but is not a neighbor. 
}
\label{fig:local-MP1}
\end{figure}
However, even in the case of $[\beta]=1$ there is a quantum solution, e.g., over $2$-qubits as given in Fig.~(\ref{fig:mermin-scenario-os}).
The quantity $[\beta]$ is, in fact, cohomological, as first observed in \cite{Coho}. 
The 
cohomological perspective is based on reorganizing the scenario into a  space. Then the Mermin scenario is represented as a torus; see Fig.~(\ref{fig:mermin-scenario-beta1}). 
In this representation, measurements label the edges of the triangles, and $\beta$ assigns a value in $\ZZ_2$ to each triangle.  
Choosing a quantum state induces a nonsignaling distribution on the Mermin scenario with support on each context $C$ consisting of the set $O_\beta(C)$ of  outcome assignments $s:C\to \ZZ_2$ that satisfy $\sum_{m\in C} s(m)=\beta(C)$.
 \begin{figure}[h!]
\centering
\begin{subfigure}{.49\textwidth}
  \centering
  \includegraphics[width=.6\linewidth]{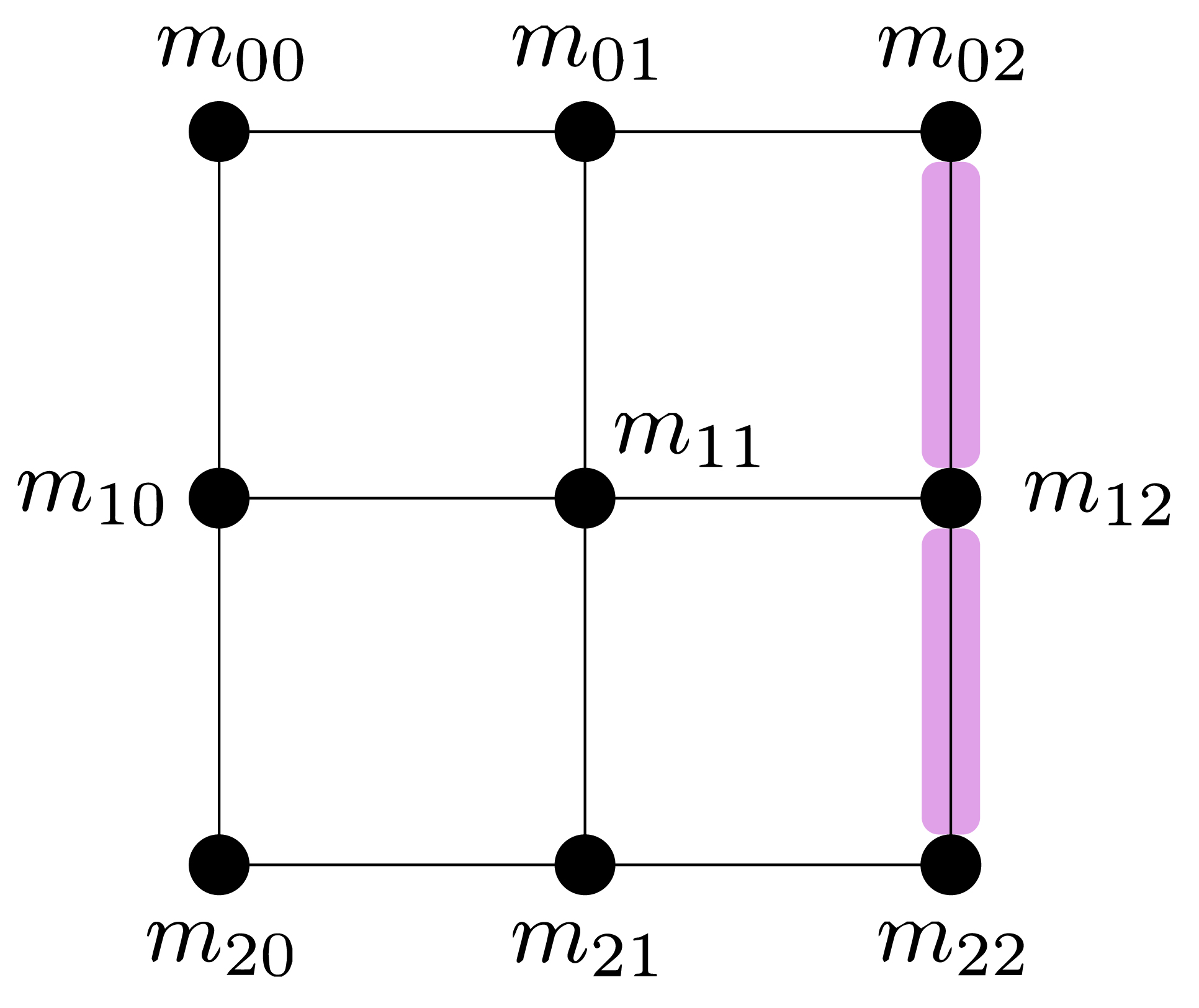}
  \caption{}
  \label{fig:mermin-scenario}
\end{subfigure}%
\begin{subfigure}{.49\textwidth}
  \centering
  \includegraphics[width=.6\linewidth]{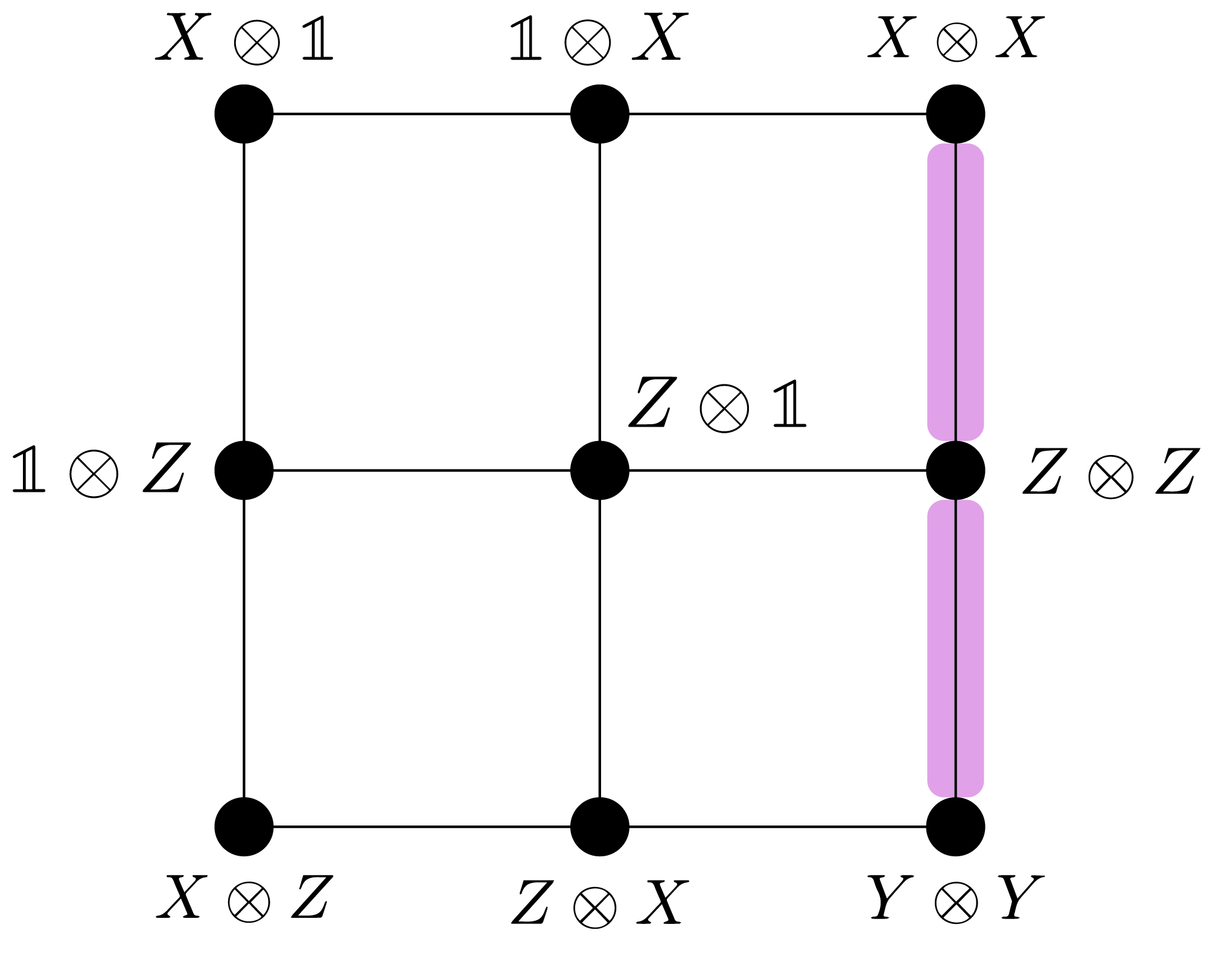}
  \caption{}
  \label{fig:mermin-scenario-os}
\end{subfigure}
\caption{Mermin scenario with $\beta=1$ for the context indicated in red color.  
}
\label{fig:mermin-scenario-and-os}
\end{figure}
Let $\NS_\cC$ denote the nonsignaling polytope for the Mermin scenario. 
We introduce a subpolytope, called the {\it Mermin polytope}, 
\begin{equation}\label{eq:NSbeta}
\MP_\beta \subset \NS_{\cC}
\end{equation}
that consists of nonsignaling distributions, that is, tuples of $p=(p_C)_{C\in \cC}$ probability distributions compatible under marginalization,  such that the support of each $p_C$ is contained in $O_\beta(C)$. 
We show that  the combinatorial isomorphism type of the  polytope $\MP_\beta$ is determined by 
$[\beta]$. 
As canonical representatives for $[\beta]=0$ and $1$ we take the choices of $\beta$'s given in Fig.~(\ref{fig:mermin-scenario-beta0}) and Fig.~(\ref{fig:mermin-scenario-beta3}); respectively. 
The resulting Mermin polytopes will be denoted by $\MP_0$ and $\MP_1$.
One of our main technical contributions is the classification of the vertices of these two polytopes.

\Thm{\label{thm:VertexClassification}
Let $\MP_\beta$ denote the Mermin polytope.
\begin{enumerate}
\item All the vertices of $\MP_0$ are deterministic distributions corresponding to the functions 
$$s:\set{m_{00},m_{01},m_{10},m_{11}}\to \ZZ_2.$$ 
There are $16$ vertices. 

\item For $\MP_1$ the vertices are given by  pairs $(\Omega,s)$ where $\Omega\subset M$ is a maximal closed noncontextual (cnc) set and $s:\Omega\to \ZZ_2$ is an outcome assignment. 
There are two types of vertices: 
\begin{itemize}
\item Type $1$: When $\Omega$ is of type $1$. 
There are $48$ vertices of this type.
\item Type $2$: When $\Omega$ is of type $2$.  
There are $72$ vertices of this type.
\end{itemize}   
\end{enumerate} 
}

Our vertex classification result relies on 
the symmetries of the Mermin polytopes. 
We identify a subgroup $G_\beta$ of the combinatorial automorphisms of $\MP_\beta$. 
We show that $G_0$ acts transitively on the vertices of $\MP_0$. This means that for any pair of vertices,
there is a symmetry of the polytope that moves one  
to the other. For $\MP_1$ the symmetry group $G_1$ acts transitively within each type of vertices. 
We also study the stabilizer group of the vertices, that is, symmetry elements that fix a given vertex, and the action of this group on the neighbor vertices to obtain a description of the graph associated to the polytopes.
In  the graph of $\MP_1$ the main structural elements are the loops on the Mermin torus that give the edges of the graph connecting a pair of neighbor vertices; see Fig.~(\ref{fig:local-MP1}).

\Thm{\label{thm:GraphClassification}
Let $\MP_\beta$ denote the Mermin polytope.
\begin{enumerate}
\item The graph of $\MP_0$ is the complete graph $K_{16}$.

\item The graph of $\MP_1$ consists of $120$ vertices and the local structure at the type $1$ and $2$ vertices is depicted in Fig.~(\ref{fig:local-MP1}). 
\end{enumerate}
}

\begin{figure}[h!]
\centering
\begin{subfigure}{.33\textwidth}
  \centering
  \includegraphics[width=.6\linewidth]{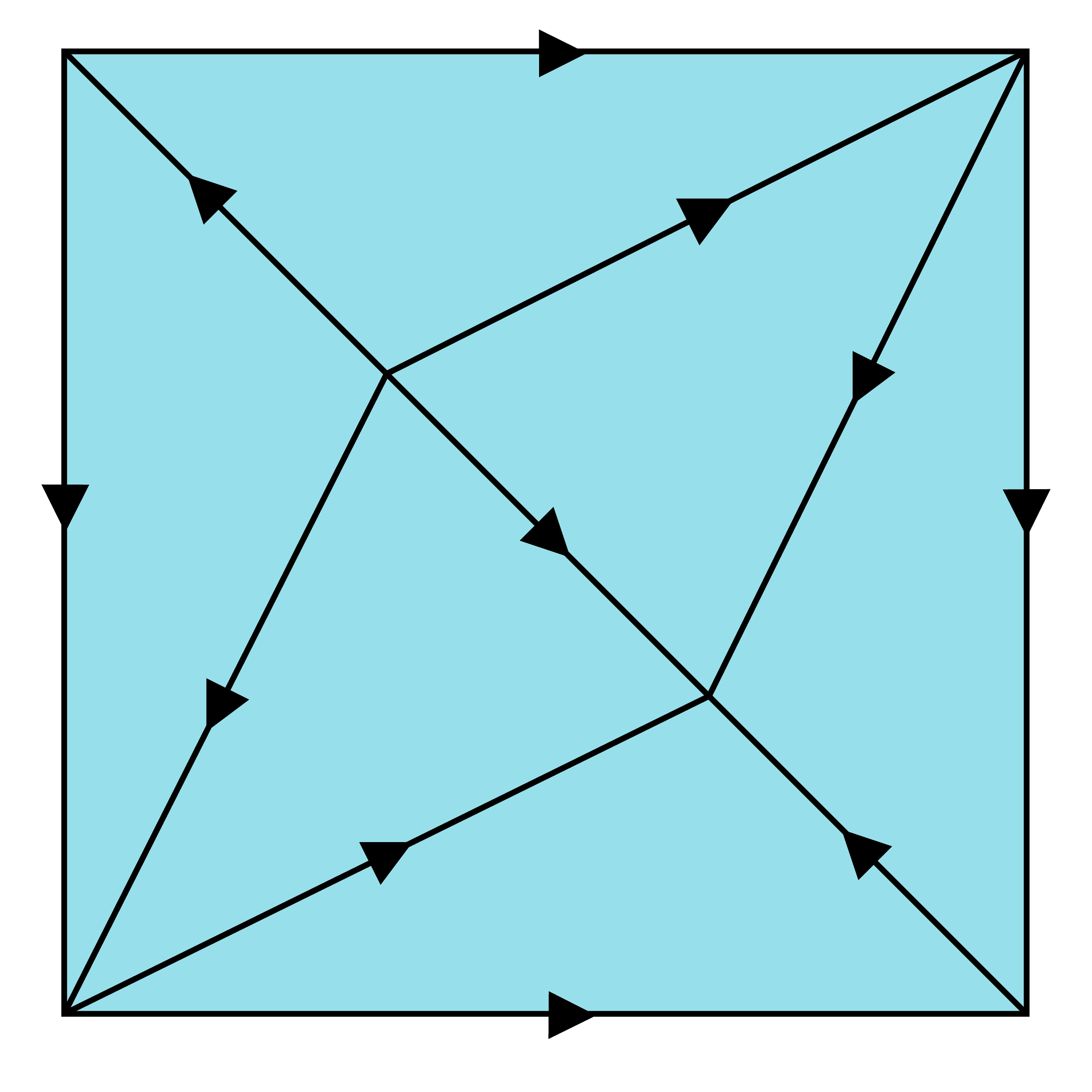}
  \caption{}
  \label{fig:mermin-scenario-beta0}
\end{subfigure}%
\begin{subfigure}{.33\textwidth}
  \centering
  \includegraphics[width=.6\linewidth]{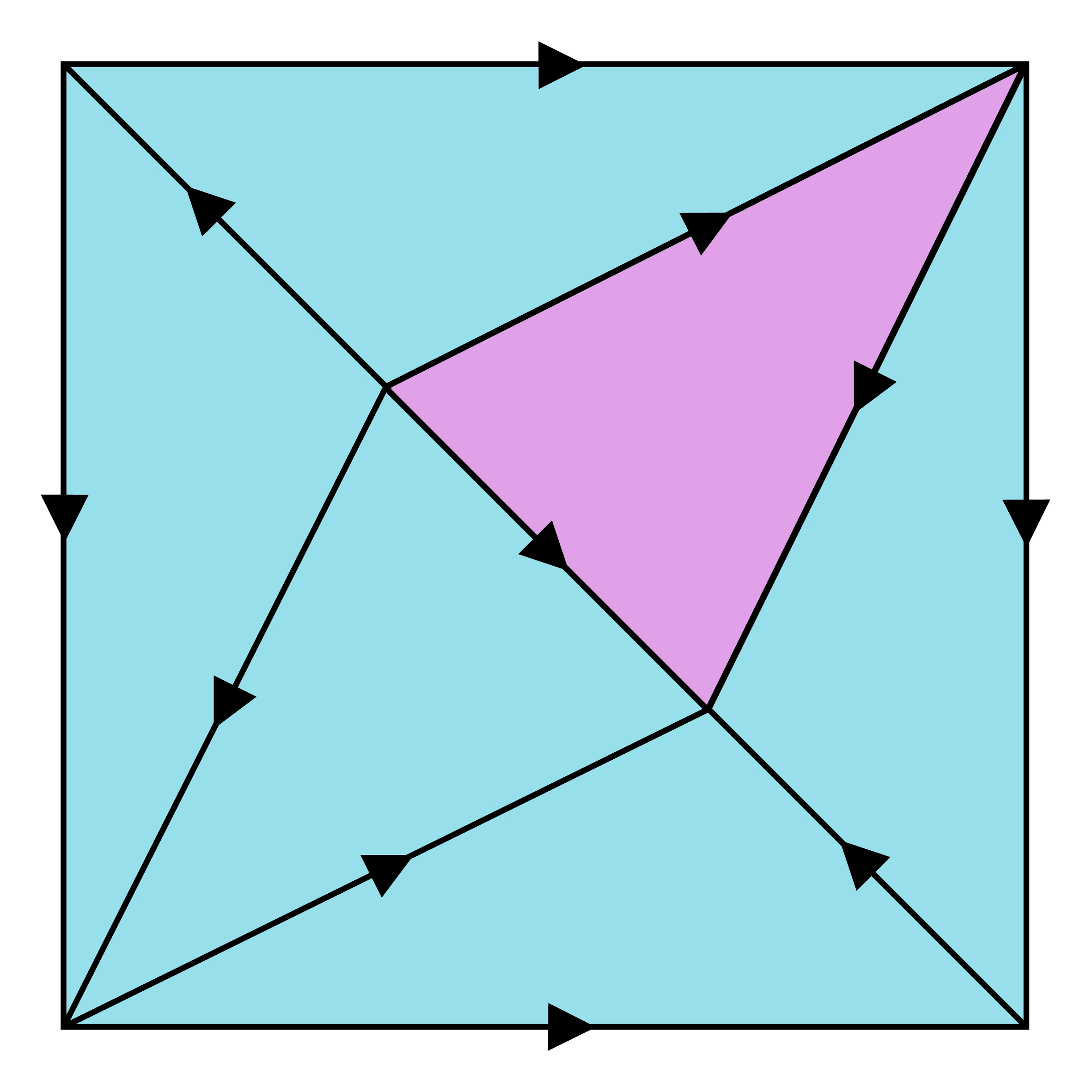}
  \caption{}
  \label{fig:mermin-scenario-beta1}
\end{subfigure}
\begin{subfigure}{.33\textwidth}
  \centering
  \includegraphics[width=.6\linewidth]{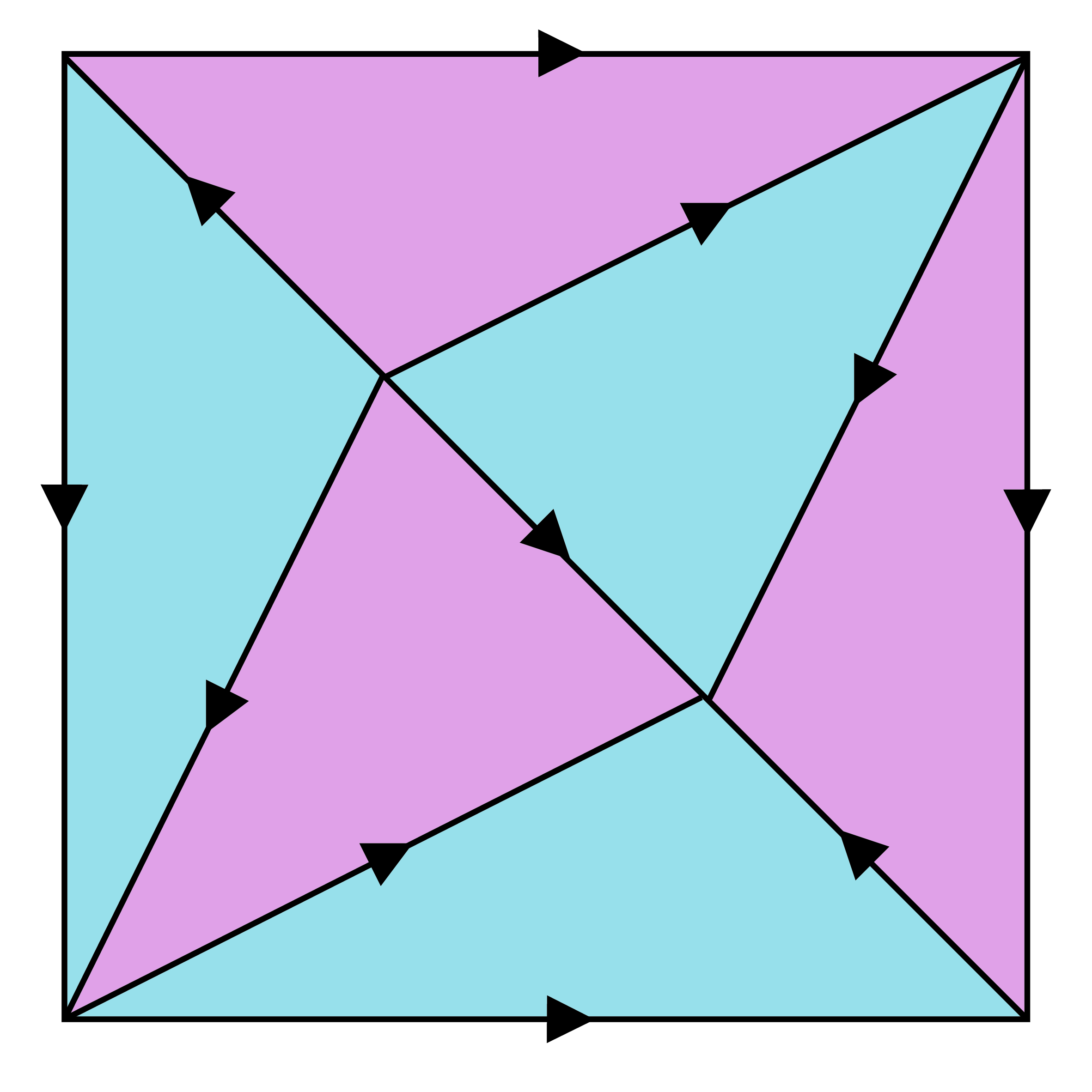}
  \caption{}
  \label{fig:mermin-scenario-beta3}
\end{subfigure}
\caption{
Mermin scenario represented as a torus: Top and bottom (left and right) edges are identified. In this representation $\beta$ assigns $\set{0,1}$ to each triangle. Red color indicates that $\beta=1$; otherwise $\beta=0$.
}
\label{fig:mermin-scenario-beta}
\end{figure}

Let us put our result into context: The polytope $\MP_0$ can be best studied within the framework of simplicial distributions introduced in \cite{okay2022simplicial}. In this framework, nonsignaling distributions can be interpreted as distributions on spaces,  as in Fig.~(\ref{fig:mermin-ns}). The nonsignaling conditions are encoded at the faces of the triangles. The Mermin scenario 
can be regarded as an extension of the well-known Clauser, Horne, Shimony, Holt (CHSH) \cite{clauser1969proposed} scenario, a Bell scenario consisting of two parties and two measurements $x_i,y_j$, where $i,j\in \ZZ_2$, for each party with binary outcomes; see Fig.~(\ref{fig:mermin-ns-measurements}). 
A fundamental result for the CHSH scenario is Fine's theorem \cite{fine1982hidden}. This theorem says that a distribution on the CHSH scenario is noncontextual if and only if the CHSH inequalities are satisfied. Our vertex classification for $\MP_0$ can be turned into a new topological proof of Fine's theorem. This proof diverges from Fine's original argument; see \cite[Thm. 4.12]{okay2022simplicial} for an alternative topological proof closer to Fine's original argument. 
We present our topological proof of Fine's theorem 
in Section \ref{sec:fine-thm}.

The other polytope, $\MP_1$, can be seen as a toy model
of a 
more complicated polytope introduced in \cite{zurel2020hidden} for classical simulation of universal quantum computation. 
For $n$-qubits, the polytope $\Lambda_n$ used in this classical simulation is defined as the polar dual of the $n$-qubit stabilizer polytope. 
These polytopes are only fully understood in the case of a single qubit: $\Lambda_1$  is a $3$-dimensional cube containing the Bloch sphere. 
The combinatorial structure of $\Lambda_n$ for $n\geq 2$ is yet to be understood. 
This mathematical problem is the main obstacle to quantifying the complexity of the $\Lambda$-simulation algorithm, a fundamental question in the study of quantum computational advantage.
The next case, $\Lambda_2$, is only understood numerically
(e.g., using  Polymake \cite{gawrilow2000polymake}).
A geometric
understanding of $\Lambda_2$ will bring insight into 
the structure of $\Lambda$-polytopes with   higher number of qubits. 
Tensoring a vertex of $\Lambda_2$ with an $(n-2)$-qubit stabilizer state produces a vertex in $\Lambda_n$ 
\cite[Theorem 2]{okay2021extremal}. 
Some of the vertices of $\Lambda_2$ are similar to the vertices of $\MP_1$. These vertices are also described by   cnc sets \cite{raussendorf2020phase}.
In fact, the Mermin polytope $\MP_1$ can be seen as a nonlocal version of $\Lambda_2$. 
The local part is captured by the 
nonsignaling polytope $\NS_{232}$ of the two party Bell scenario, consisting of two measurements with binary outcomes per parties.
Our decomposition result provides a description of $\Lambda_2$ in terms of two well-understood polytopes: $\NS_{232}$ whose vertices are described in \cite{jones2005interconversion} and $\MP_1$ described in Theorem \ref{thm:VertexClassification}.
 
Our main contributions in this paper can be summarized as follows:
\begin{itemize}
\item We define families of Mermin polytopes parametrized by a function $\beta$
and classify the corresponding polytopes by 
the
cohomology class $[\beta]$ (Proposition~\ref{pro:MPbeta-cohomology}).

\item The symmetry groups $G_\beta$ of each equivalence class of Mermin polytopes are described and we demonstrate that they are isomorphic (Proposition~\ref{pro:G0-G1}).

\item A complete characterization of the vertices for both classes of Mermin polytopes is given (Theorem~\ref{thm:VertexClassification}).

\item $G_0$ acts transitively on the vertices of $\MP_0$ (Lemma \ref{lem:g0-transitive}) and $G_1$ acts transitively on the vertices of $\MP_1$ of a fixed type (Lemma \ref{lem:stabilizer-MP1}). The latter result also describes the stabilizers of each type of vertices.

\item Graphs of both Mermin polytopes are described (Theorem \ref{thm:Graph-MP0} and Theorem \ref{thm-MP1-graph}).

\item We exploit the relationship between the Mermin and CHSH scenarios to provide a 
new topological proof of Fine's theorem \cite{fine1982hidden,fine1982joint} (Theorem~\ref{thm:fine}). An important step is  the vertex classification for  $\MP_0$, which implies that any distribution on the Mermin torus is noncontextual (Corollary \ref{cor:Mermin-noncontextual}).

\item The $\Lambda_{2}$ polytope is decomposed into local and nonlocal polytopes. The former is a well-known nonsignaling polytope $NS_{232}$ \cite{jones2005interconversion}, while the nonlocal part is precisely the Mermin polytope 
(Theorem~\ref{thm:Lambda2NS}).
\end{itemize}

The rest of the paper is organized as follows. In Section~\ref{sec:mermin-polytopes} we formalize the Mermin scenario and the notion of Mermin polytopes. In Section~\ref{sec:mermin-vertices} we characterize the vertices of the Mermin polytopes.
In Section \ref{sec:Graph} we describe the graphs of the polytopes. In Section~\ref{sec:applications} we apply the vertex characterization to 
problems in
quantum foundations and quantum computation. More involved proofs for Propositions~\ref{pro:MPbeta-cohomology} and \ref{pro:G0-G1} can be found in Appendices~\ref{sec:ProofPro} and \ref{sec:proof-g0-g1}, respectively. Appendix \ref{sec:Stabilizer} contains the description of the stabilizer groups of the vertices of $\MP_1$.


\section{Mermin polytopes} \label{sec:mermin-polytopes}

Mermin polytopes mentioned in this paper are certain subpolytopes of nonsignaling polytopes associated to the Mermin scenario. In this section we introduce these polytopes formally and show that up to combinatorial isomorphism of polytopes there are two types denoted by $\MP_0$ and $\MP_1$. Our main result is a classification theorem for  the vertices of these polytopes.


\subsection{Definition}

A {\it measurement scenario}, or more briefly a {\it scenario}, consists of the following data:
\begin{itemize}
\item a set $M$ of measurements,
\item a collection $\cC$ of subsets $C\subset M$, called contexts, that cover the whole set of measurements, i.e. 
$$M=\cup_{C\in \cC} C,$$
\item a set of outcomes, which through the paper is fixed as $\ZZ_2=\set{0,1}$.
\end{itemize} 
Since the outcome set is fixed we will write $(M,\cC)$ to denote a scenario.
For a set $U$ we will write $\ZZ_2^U$ 
for the set $\set{s:U\to \ZZ_2}$ of functions on a context $C\in \cC$.
The {\it nonsignaling polytope} on this scenario, denoted by $\NS_\cC$, consists of collections $(p_C)_{C\in \cC}$ of probability distributions, each given by a function $p_C:\ZZ_2^C\to \RR_{\geq 0}$ where $\sum_{m\in C}p(m)=1$, satisfying the nonsignaling condition given by
$$
p_C|_{C\cap C'} = p_{C'}|_{C'\cap C}\;\;\; \forall C,C'\in \cC.
$$
The {\it restriction} $p_C|_{C\cap C'}$ corresponds to marginalization of the distribution to the intersection. A distribution $p$ is called {\it noncontextual} if there exists a distribution $d:\ZZ_2^M\to \RR_{\geq 0}$ such that $p_C=d|_C$ for all $C\in \cC$. Otherwise, $p$ is called {\it contextual}. For more details see \cite{abramsky2011sheaf}. We will write $\NS_\cC$ for the polytope of nonsignaling distributions on the scenario $(M,\cC)$.

We are interested in polytopes associated to binary linear systems \cite{cleve2014characterization}. 
A {\it binary linear system} consists of  a scenario $(M,\cC)$ together with a function $\beta:\cC\to \ZZ_2$. For each $C$ we will write
$$
O_\beta(C) = \set{s:C\to \ZZ_2:\, \sum_{m\in C} s(m)=\beta(C)}\,\subset \ZZ_2^C.
$$
A function in this set will be referred to as an {\it outcome assignment} on the context $C$.
We introduce a subpolytope 
\begin{equation}\label{eq:NSbeta}
\NS_{\cC,\beta} \subset \NS_{\cC}
\end{equation}
that consists of nonsignaling distributions $p=(p_C)_{C\in \cC}$ such that
$$
\supp(p_C) \subset O_\beta(C)\;\;\; \forall C\in \cC
$$
where $\supp(p_C)$ stands for the support of $p_C$, i.e.,  the set of functions $s:C\to \ZZ_2$ 
such that $p_C(s)> 0$.

\Def{\label{def:MerminPolytope}
{\rm
The {\it Mermin scenario} consists of 
\begin{itemize}
\item the measurement set $M=\set{m_{ij}:\, i,j\in \ZZ_3}$, and
\item the cover $\cC$ given by two types of contexts:
\begin{itemize}
\item Horizontal: $\cC^\hor=\set{C_i^\hor:\,i\in \ZZ_3}$ where $C_i^\hor= \set{m_{ij}:\,j\in \ZZ_3}$,
\item Vertical: $\cC^\ver=\set{C_j^\ver:j\in \ZZ_3}$ where $C_j^\ver=\set{m_{ij}:\,i\in \ZZ_3}$.
\end{itemize}
\end{itemize}
The {\it Mermin polytope} for a function $\beta:\cC\to \ZZ_2$ is defined to be $\MP_\beta=\NS_{\cC,\beta}$.
Analogously we can consider quasiprobability distributions on the Mermin scenario with   restricted support. We will write $\MP_\beta^\RR$ for this polytope. 
}
}

In this paper we will study the Mermin polytope associated to the Mermin scenario  $(M,\cC)$.

\subsection{Topological representation}\label{sec:top-rep}

In \cite{mermin1993hidden} it was shown that the Mermin scenario can be represented by a torus with a certain triangulation. In this representation contexts are represented by triangles. We will follow the more recent approach developed in \cite{okay2022simplicial} to represent nonsignaling distributions in a topological way. Given a context $C=\set{x,y,z}$ in $\cC$ we represent the distribution $p_C$ as in Fig. (\ref{fig:single-triangle}).
For a measurement $x$ we write $p_x^0$ for the probability of measuring outcome $0$. Similarly given a pair $x,y$ of measurements $p_{xy}^{ab}$ denotes the probability for the outcome assignment $(x,y)\mapsto (a,b)$.
\begin{figure}[h!]
\centering
\begin{subfigure}{.33\textwidth}
  \centering
  \includegraphics[width=.6\linewidth]{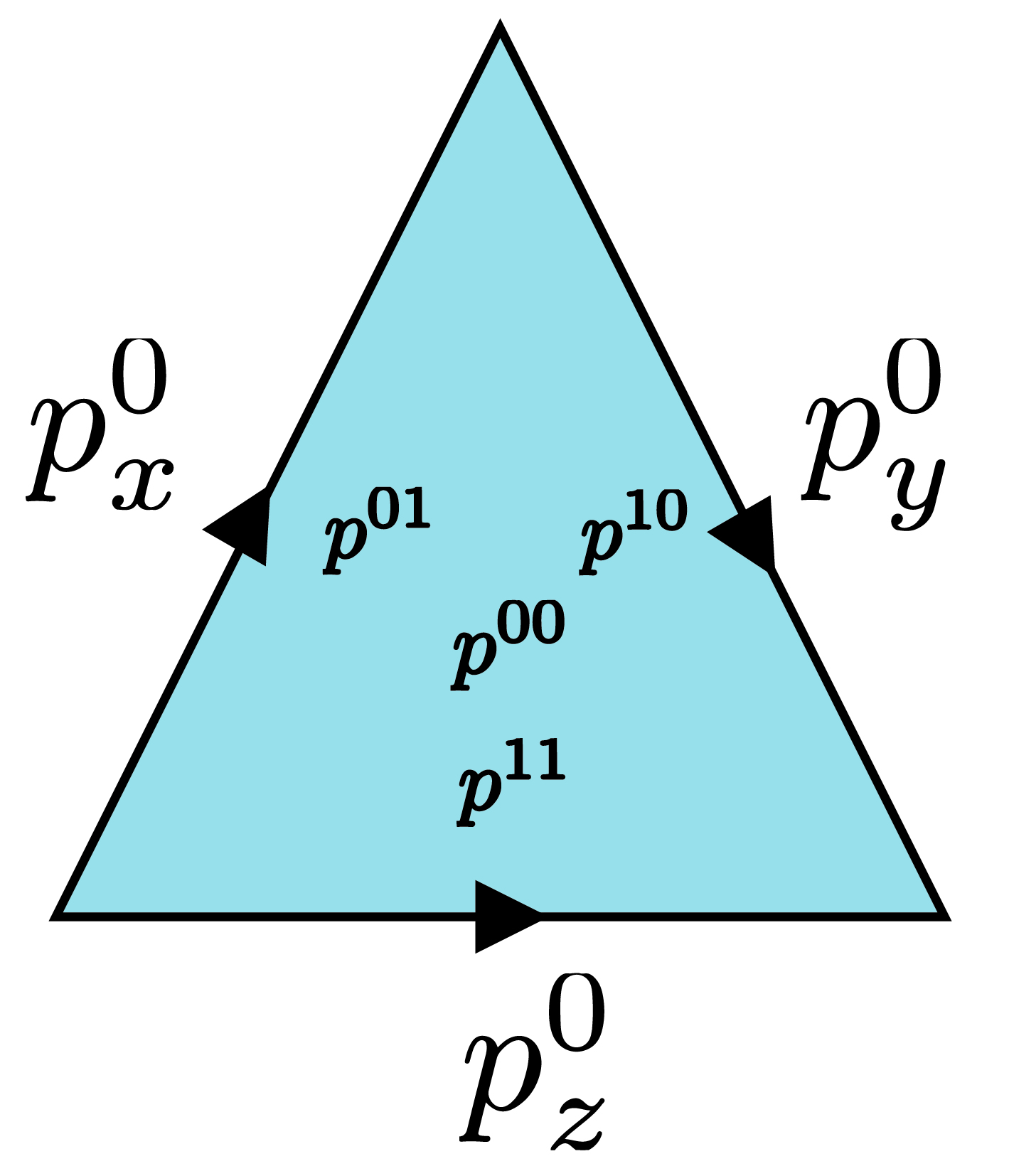}
  \caption{}
  \label{fig:triangle-beta0}
\end{subfigure}%
\begin{subfigure}{.33\textwidth}
  \centering
  \includegraphics[width=.6\linewidth]{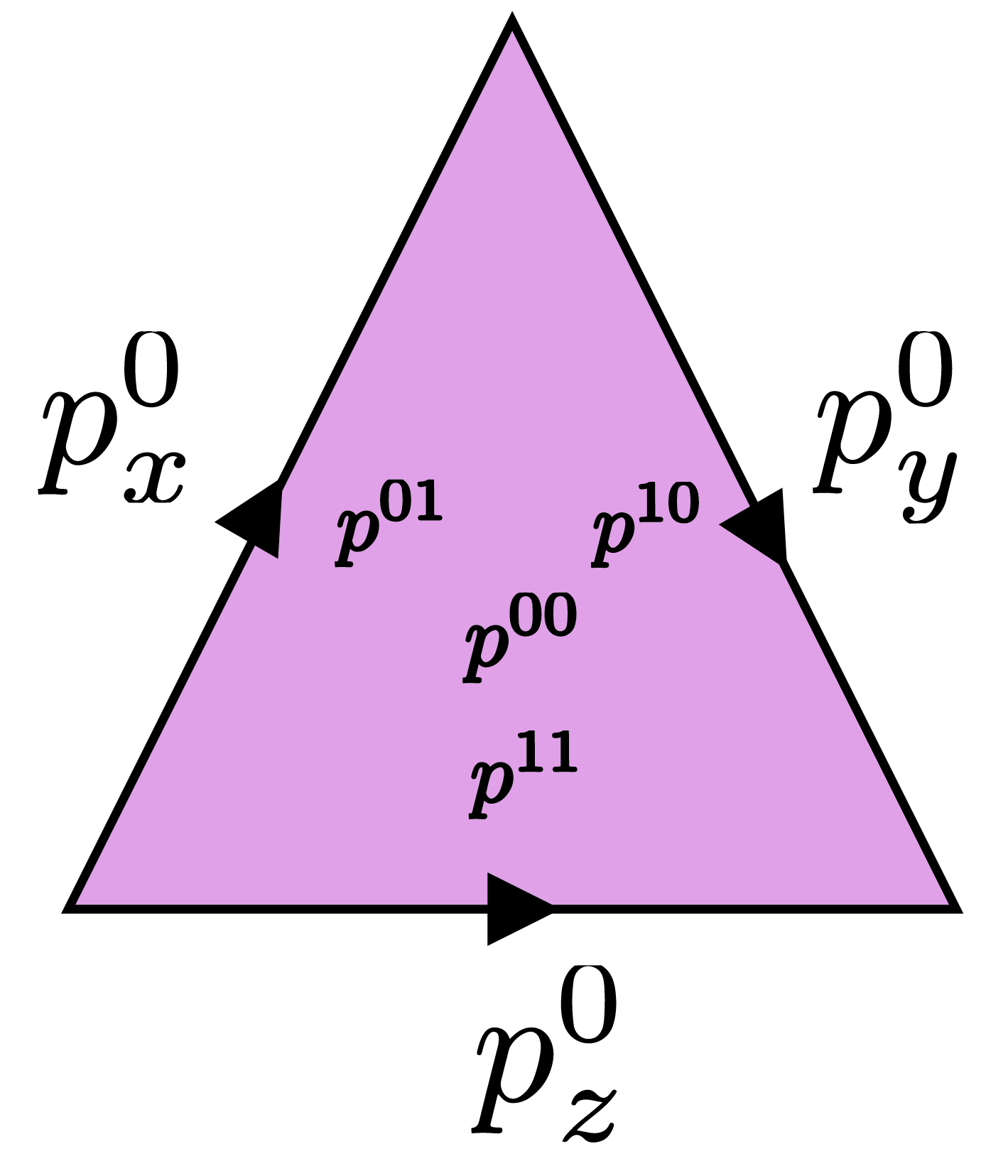}
  \caption{}
  \label{fig:triangle-beta1}
\end{subfigure}
\caption{(a) Triangle with $\beta=0$. (b) Triangle with $\beta=1$. The marginal at $z$ is given by Eq.~(\ref{eq:z-beta0}), or Eq.~(\ref{eq:z-beta1}); respectively.    
}
\label{fig:single-triangle}
\end{figure}
Given a triangle with a probability distribution as in Fig.~(\ref{fig:single-triangle}) the probabilities at the $x,y$ edges are given by
$$
\begin{aligned}
p_x^0 &= p^{01}+p^{00} \\
p_y^0 &= p^{10}+p^{00}.
\end{aligned}
$$
Fig.~(\ref{fig:triangle-beta0}) represents the case where $\beta=0$. In this case  
\begin{equation}\label{eq:z-beta0}
p_{z}^0 = p^{00}+p^{11},
\end{equation}
whereas if $\beta=1$ as in Fig.~(\ref{fig:triangle-beta1}) then 
\begin{equation}\label{eq:z-beta1}
p_{z}^0 = p^{01}+p^{10}.
\end{equation}
Therefore, in effect $z$ is the XOR measurement $x\oplus y$ in the first case, and the NOT of the XOR measurement $\overline{x\oplus y}$ in the second. 

In Fig.~(\ref{fig:mermin-scenario-beta}) Mermin scenario with various choices of $\beta$'s are represented on a torus.
In this framework, $\beta$   assigns $0$ or $1$ to each triangle, hence can be interpreted as a cochain from algebraic topology.
The value given by the sum in Eq.~(\ref{eq:beta-coho})
has a special meaning in this context known as the cohomology class of $\beta$.  In this paper we don't assume familiarity with cochains, or with other topological notions such as cohomology; see \cite{Coho} for more on the cohomological perspective.

\Pro{\label{pro:MPbeta-cohomology}
Given two functions $\beta,\beta':\cC\to \ZZ_2$ the Mermin polytope
$\MP_\beta$ is combinatorially isomorphic to $\MP_{\beta'}$ if and only if $[\beta]=[\beta']$.
}

Proof of this result is given in  Appendix \ref{sec:ProofPro}.
As a consequence there are two types of Mermin polytopes, up to combinatorial isomorphism, 
corresponding to the cases $[\beta]=0$ and $1$.

\subsection{The even case: $\MP_0$}\label{sec:mp0-case}

Let $\beta_0:\cC\to \ZZ_2$ denote the function defined by
\begin{equation}\label{eq:beta0}
\beta_0(C)=0,\;\;\; \forall C\in \cC.
\end{equation}
We will simply write $\MP_0$ to denote the
  Mermin polytope $\MP_{\beta_0}$. Note that this notation is justified by the observation that the isomorphism type of $\MP_\beta$ only depends on $[\beta]$ as proved in Proposition \ref{pro:MPbeta-cohomology}.
Our goal in this section is to relate this polytope to 
a famous bipartite Bell scenario, usually referred to as the CHSH scenario.

\begin{figure}[h!]
\centering
\begin{subfigure}{.49\textwidth}
  \centering
  \includegraphics[width=.6\linewidth]{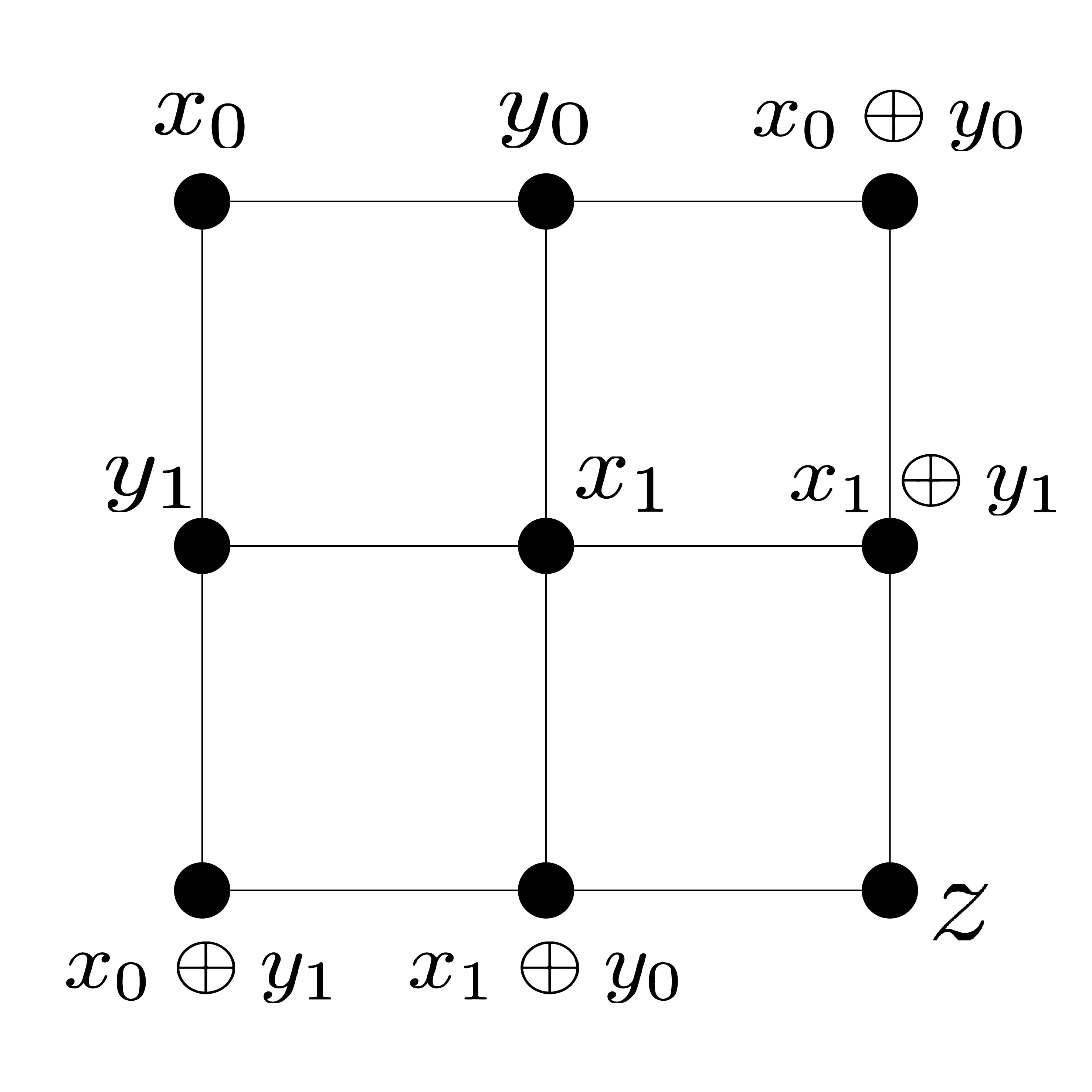}
  \caption{}
  \label{fig:mermin-ns-measurements}
\end{subfigure}%
\begin{subfigure}{.49\textwidth}
  \centering
  \includegraphics[width=.6\linewidth]{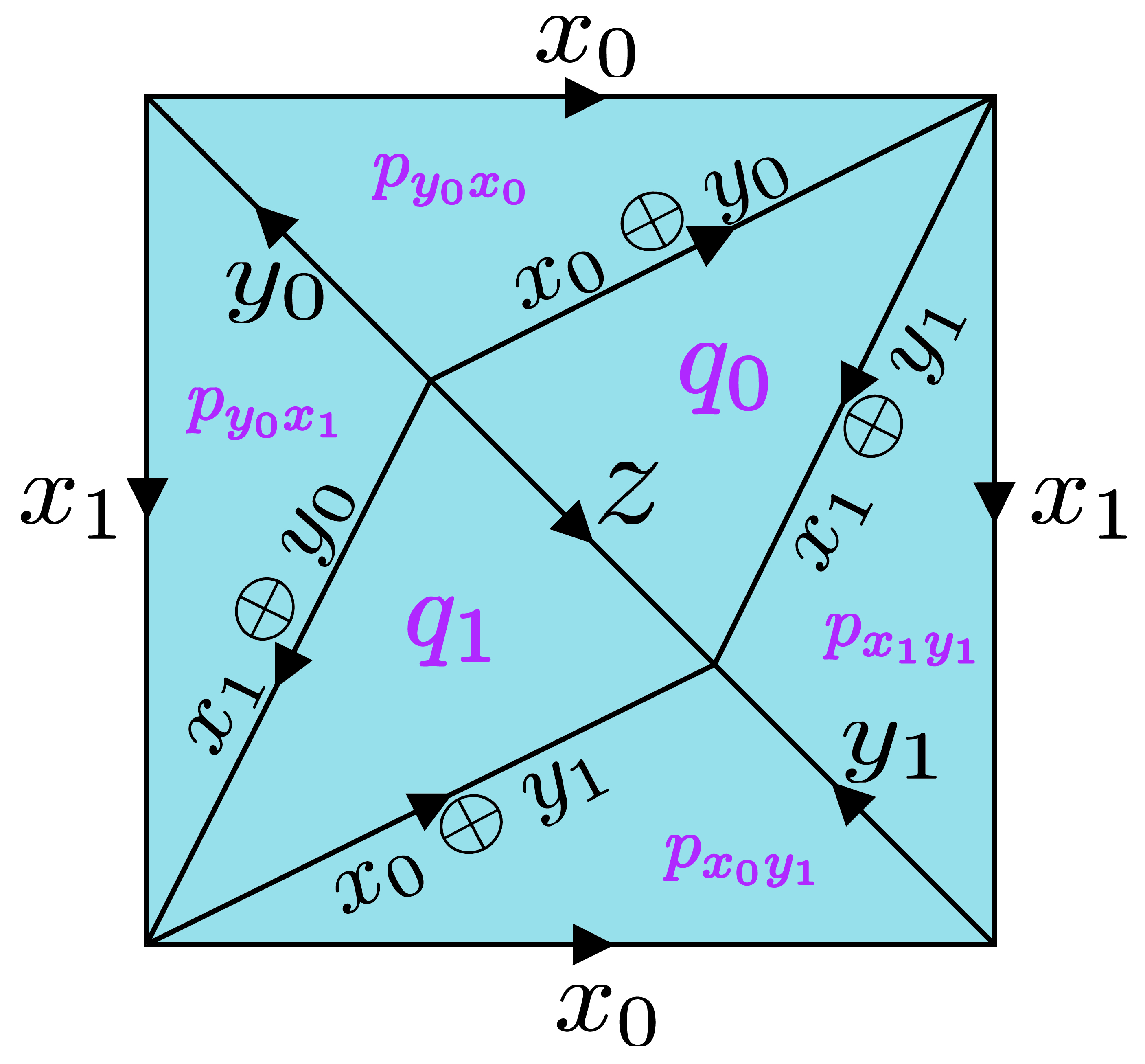}
  \caption{}
  \label{fig:mermin-ns}
\end{subfigure}
\caption{(a) Mermin scenario in the conventional representation. Vertices correspond to measurement labels. (b) Mermin scenario in the topological representation. Measurements label the edges.
}
\label{fig:ns}
\end{figure}

 The CHSH scenario is a particular type of Bell scenario for $2$ parties, $2$ measurements per party and $2$ outcomes per measurement. More precisely, this scenario consists of
\begin{itemize}
\item the measurement set $\set{x_i,y_j:\,i,j\in \ZZ_2}$ where $x_i$'s are for Alice and $y_j$'s are for Bob, and
\item  the contexts $\set{x_i,y_j}$ where $i,j\in \ZZ_2$.
\end{itemize}
Mermin scenario can be obtained from the CHSH scenario by adding two additional contexts $\set{x_0\oplus y_0,x_1\oplus y_1, z}$ and $\set{x_0\oplus y_1,x_1\oplus y_0,z}$, where $z=x_0\oplus y_0 \oplus x_1\oplus y_1$, consisting of the XOR's of the measurements of Alice and Bob; see Fig.~(\ref{fig:mermin-ns-measurements}). See Fig.~(\ref{fig:mermin-ns}) for a topological representation. For the convenience of the reader we list the nonsignaling conditions 
\begin{equation}\label{eq:Mermin-ns}
\begin{aligned}
p_{x_0}^0 &= p_{y_0x_0}^{10} + p_{y_0x_0}^{00} = p_{x_0y_1}^{01} + p_{x_0y_1}^{00}\\
p_{y_0}^0 &= p_{y_0x_0}^{01} + p_{y_0x_0}^{00} = p_{y_0x_1}^{01} + p_{y_0x_1}^{00}\\
p_{x_1}^0 &= p_{y_0x_1}^{10} + p_{y_0x_1}^{00} = p_{x_1y_1}^{01} + p_{x_1y_1}^{00}
 \\
p_{y_1}^0 &= p_{x_0y_1}^{10} + p_{x_0y_1}^{00} = p_{x_1y_1}^{10} + p_{x_1y_1}^{00}
 \\
p_{x_0\oplus y_0}^0  &= p_{y_0x_0}^{11} + p_{y_0x_0}^{00} = q_{0}^{01} + q_{0}^{00} \\
p_{x_1\oplus y_0}^0  &= p_{y_0x_1}^{11} + p_{y_0x_1}^{00} = q_{1}^{01} + q_{1}^{00} \\
p_{x_0\oplus y_1}^0  &= p_{x_0y_1}^{11} + p_{x_0y_1}^{00} = q_{1}^{10} + q_{1}^{00} \\
p_{x_1\oplus y_1}^0  &= p_{x_1y_1}^{11} + p_{x_1y_1}^{00} = q_{0}^{10} + q_{0}^{00} \\
p_{z}^0  &= q_{0}^{11} + q_{0}^{00} = q_{1}^{11} + q_{1}^{00}.
\end{aligned}
\end{equation}

\Pro{\label{pro:noncontextual-torus}
A distribution $p$ on the CHSH scenario is noncontextual if and only if it extends to a distribution on the Mermin scenario.
}
\Rem{{\rm
This result first appeared in \cite{okay2022simplicial}. Its proof relies on Fine's theorem characterizing noncontextual distributions using the CHSH inequalities. We will provide a proof of this result independent of Fine's theorem (see Proposition \ref{pro:simpcont-4.7}) by describing all the vertices of $\MP_0$. Then this observation will be used to provide a new topological proof of Fine's theorem.
}}

\begin{figure}[h!]
\centering
\includegraphics[width = .3\linewidth]{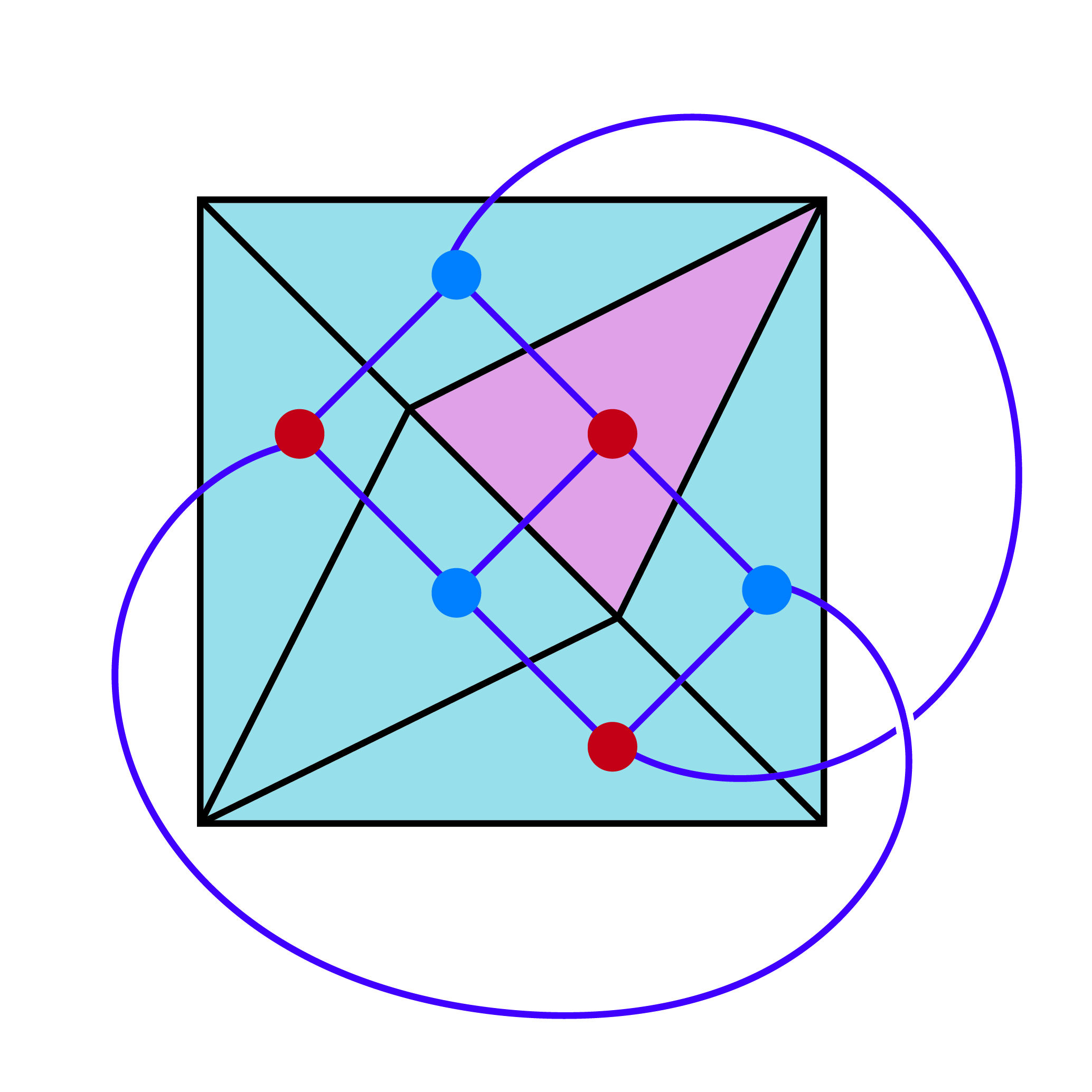}
\caption{Graph of $K_{3,3}$.
}
\label{fig:mermin-scenario-k33}
\end{figure}

Next we discuss the symmetries of $\MP_0$. 
For a polytope $P$ let $\Aut(P)$ denote the {\it group of combinatorial automorphisms} of the polytope. 
We begin by describing certain elements of this symmetry group. 
First we consider a graph obtained from the Mermin scenario. The vertices of this graph are given by the contexts, i.e., $\cC=\cC^\hor \sqcup \cC^\ver$, and the edges are given by the set $M$ of measurements. 
The resulting graph is the {\it bipartite complete graph} $K_{3,3}$; see Fig~(\ref{fig:mermin-scenario-k33}).   
The automorphism group $\Aut(K_{3,3})$ of this graph is generated by the following operations \cite{sreekumar2021automorphism}: 
\begin{enumerate}[(1)]
\item Permutation of the vertices in $\cC^\hor$ while keeping $\cC^\ver$ fixed.  
\item Permutation of the vertices in $\cC^\ver$ while keeping $\cC^\hor$ fixed.
\item The permutation exchanging 
$$
\begin{aligned}
C_1^\ver \leftrightarrow C_0^\hor\\
C_2^\ver \leftrightarrow C_2^\hor\\
C_0^\ver \leftrightarrow C_1^\hor
\end{aligned}
$$
\end{enumerate}
Denoting the symmetric group on $n$ letters by $\Sigma_n$ the symmetry group can be expressed as a semidirect product
$$
\Aut(K_{3,3}) = ( \Sigma_3 \times \Sigma_3  ) \rtimes \ZZ_2.
$$
Each factor represents a type of symmetry given in (1), (2) and (3); respectively. Geometrically the symmetry operation (3) corresponds to a reflection about the diagonal  in the torus; see Fig.~(\ref{fig:mermin-ns}).

\begin{figure}[h!]
\centering
\begin{subfigure}{.49\textwidth}
  \centering
  \includegraphics[width=.9\linewidth]{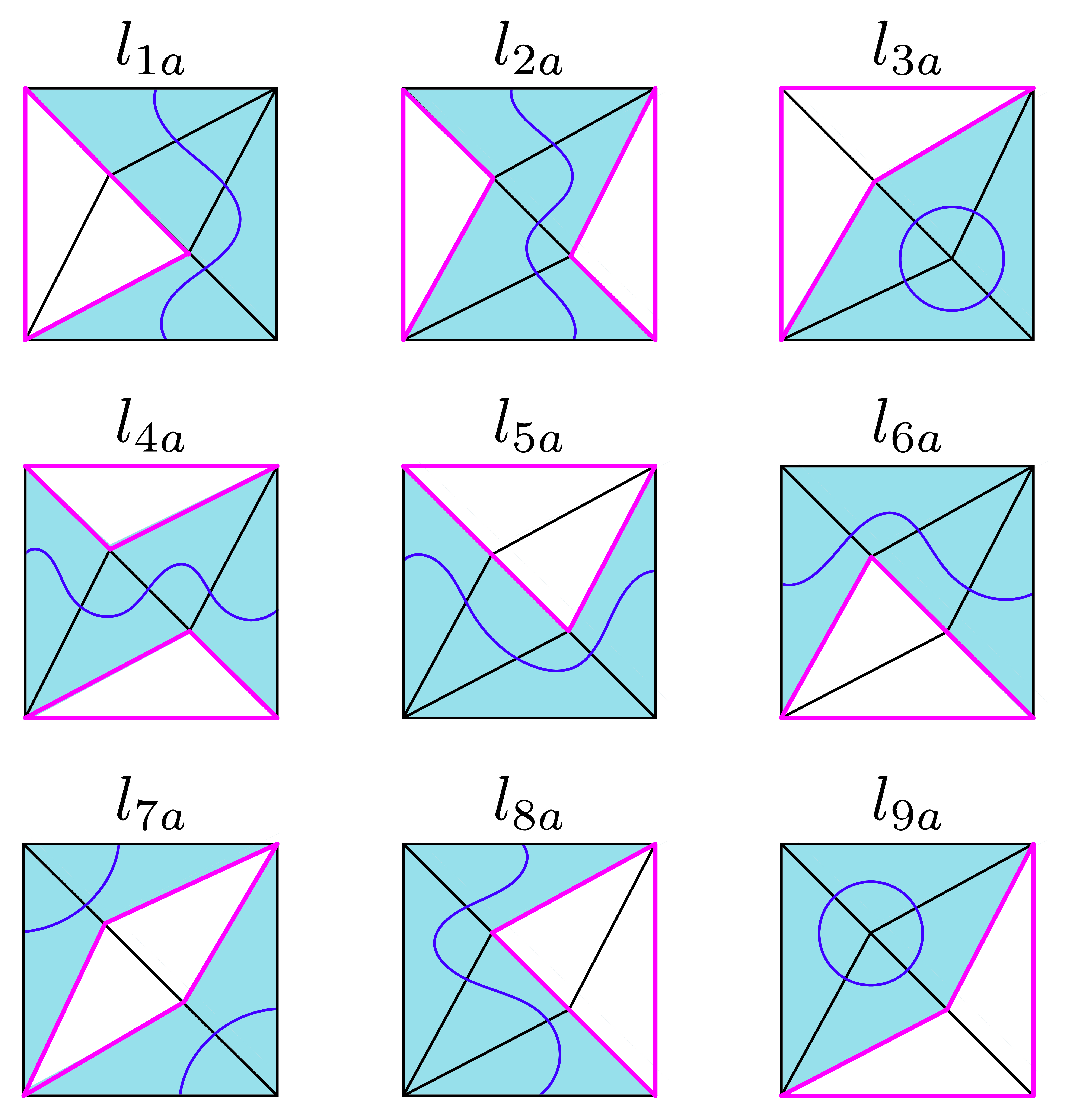}
  \caption{}
  \label{fig:mermin-torus-loops-a}
\end{subfigure}%
\begin{subfigure}{.43\textwidth}
  \centering
  \includegraphics[width=.64\linewidth]{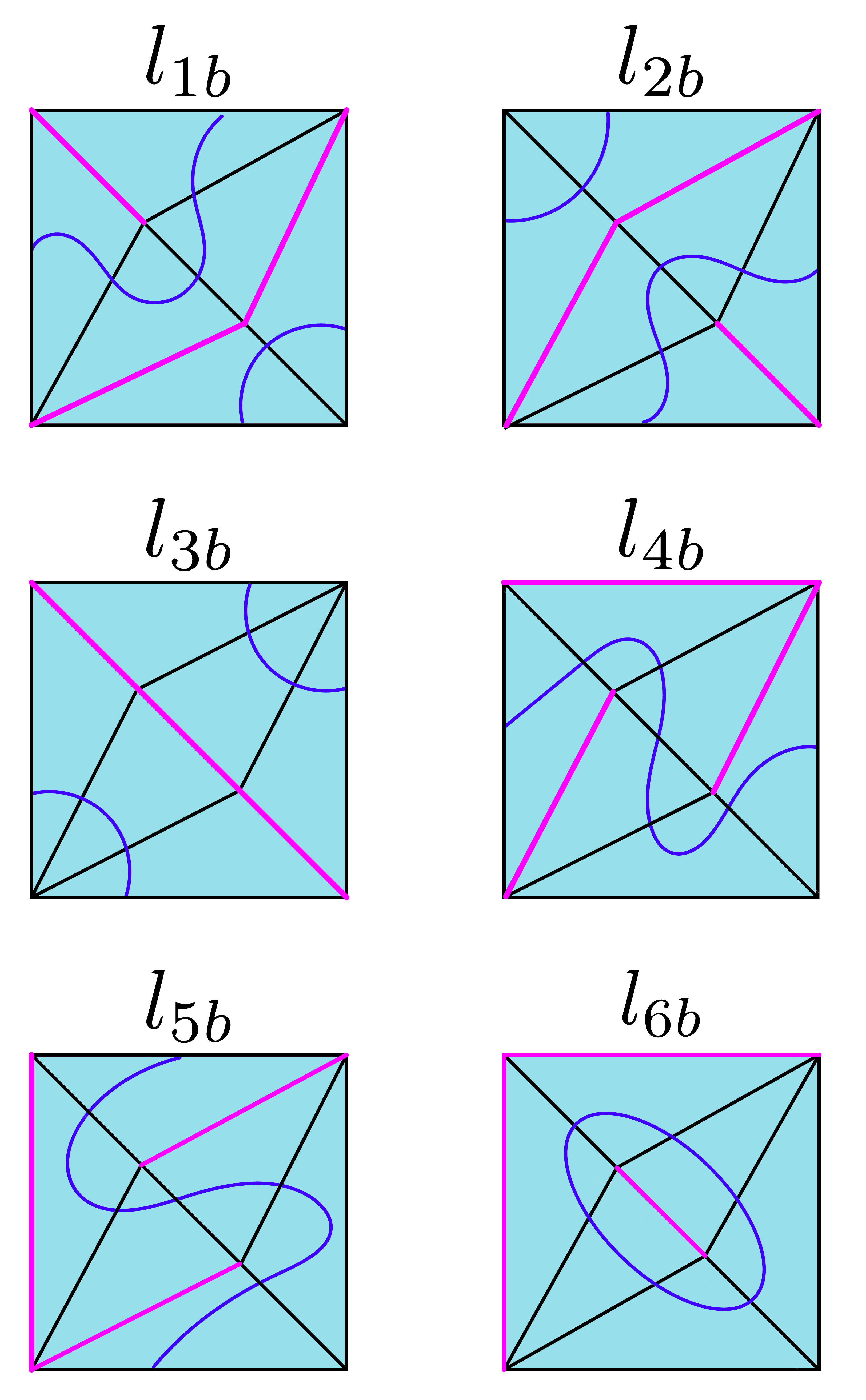}
  \caption{}
  \label{fig:mermin-torus-loops-b}
\end{subfigure}
\caption{There are two types of loops denoted by $l_{1a},l_{2b},\cdots,l_{9a}$ and $l_{1b},l_{2b},\cdots,l_{6b}$.   
}
\label{fig:mermin-torus-loops}
\end{figure}

Another kind of symmetry of $\MP_0$ comes from flipping the outcomes of the measurements in $M$. Each such symmetry operation can be represented by a loop in the graph $K_{3,3}$. Let $\ell(K_{3,3})$ denote the set of loops on the graph; see Fig.~(\ref{fig:mermin-torus-loops}) for the complete list of loops. For each such loop $l$ there is a group element $g_l$ that acts on $\MP_0$ by flipping the outcomes of the measurements that live on the loop. 
Let $G_\ell$ denote the subgroup of $\Aut(\MP_0)$ generated by the elements $g_l$ for $l\in \ell(K_{3,3})$.

\Lem{\label{lem:decomposition-loop}
$G_\ell$ is isomorphic to $\ZZ_2^4$ with the canonical generators given by the loops $l_{x_0}=l_{2a}$ (flipping $x_0$), $l_{x_1}=l_{4a}$ (flipping $x_1$), $l_{y_0}=l_{9a}$ (flipping  $y_0$) and $l_{y_1}=l_{3a}$ (flipping $y_1$).
}
\Proof{
Proof of this result follows from directly verifying that $g_l$ for each loop $l$ in Fig.~(\ref{fig:mermin-torus-loops}) can be decomposed into a product of these canonical generators. For example, $g_{l_{1a}} = g_{l_{x_0}} g_{l_{y_1}} $ and $g_{l_{1b}} = g_{l_{x_0}} g_{l_{y_1}} g_{l_{x_1}}$.
Similarly for the remaining loops.
}

We will write $G_0$ for the subgroup of  $\Aut(\MP_0)$ generated by the two subgroups $\Aut(K_{3,3})$ and $G_\ell$. This group can be expressed as an extension
\begin{equation}\label{eq:G0}
0\to G_\ell \to G_0 \to \Aut(K_{3,3}) \to 1,
\end{equation}
that is $G_\ell$ is a normal subgroup of $G_0$ and the quotient group $G_0/G_\ell$ is given by $\Aut(K_{3,3})$.



\subsection{The odd case: $\MP_1$} \label{sec:MP1}
Let $\beta_1:\cC\to \ZZ_2$ denote the function defined by
\begin{equation}\label{eq:beta1}
\beta_1(C)=0,\;\;\; \forall C\in \cC^\hor\;\; \text{ and }
\;\;
\beta_1(C)=1,\;\;\; \forall C\in \cC^\ver.
\end{equation}
We will write $\MP_1$ for the 
the Mermin polytope $\MP_{\beta_1}$.
(This notation is justified by Proposition \ref{pro:MPbeta-cohomology}.)
Our goal in this section is to provide a quantum mechanical description of  $\MP_1$. Using this description we will study the symmetries of the polytope.

\begin{figure}[h!] 
  \centering
  \includegraphics[width=.3\linewidth]{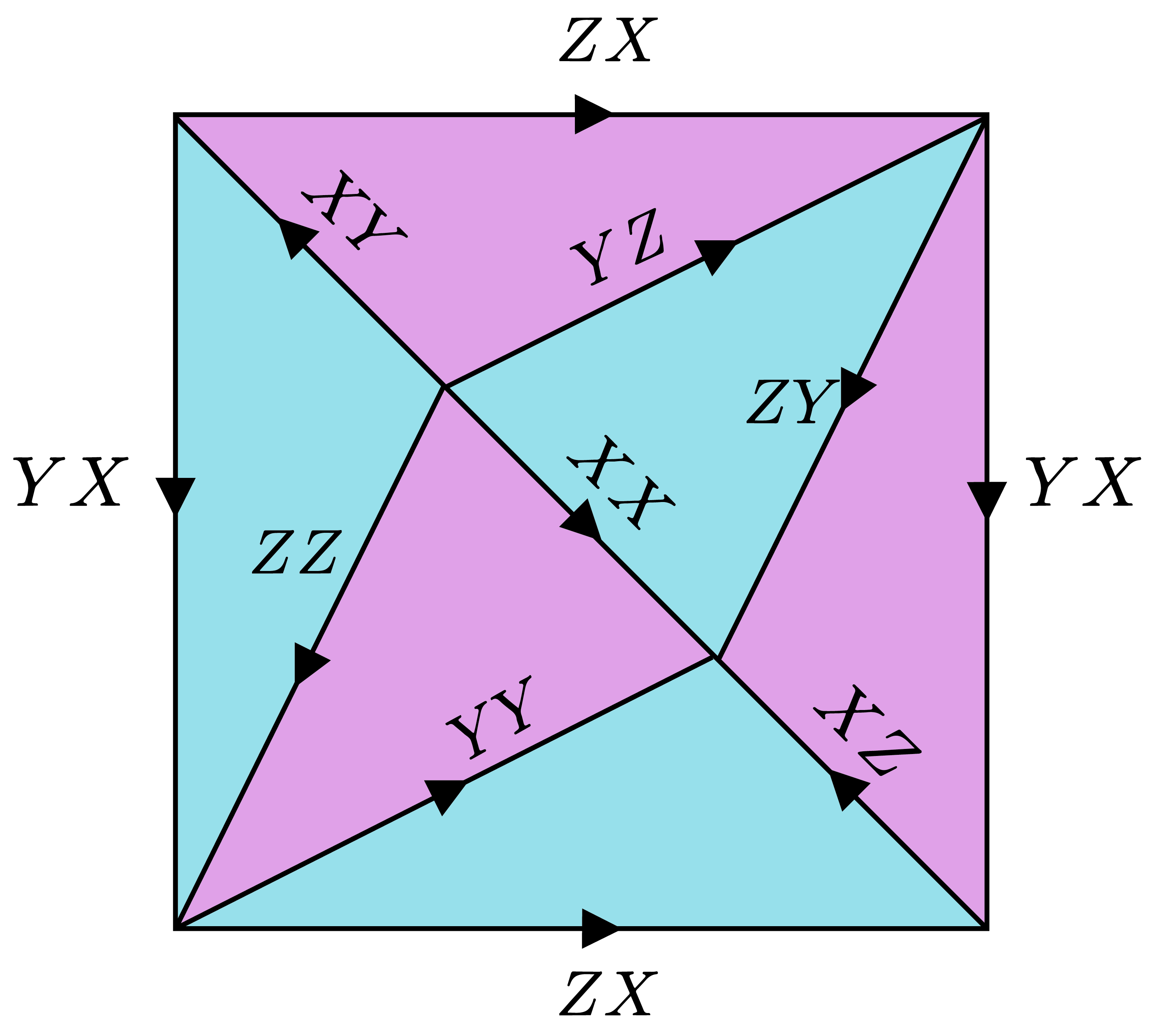}
\caption{ 
Mermin square whose edges are labeled by nonlocal Pauli operators. The tensor product is omitted from the notation.
}
\label{fig:mermin-nonlocal}
\end{figure}

The connection to quantum theory is via the notion of binary linear systems.
A {\it quantum solution} to the Mermin square binary linear system $(M,\cC,\beta)$   consists of unitary operators $A_{ij}\in U(\CC^d)$ where $i,j\in \ZZ_3$ such that
\begin{itemize}
\item $A_{ij}^2=\one$ for all $i,j\in \ZZ_3$,
\item $\set{A_{ij}:\,i\in \ZZ_3}$ and $\set{A_{ij}:\,j\in \ZZ_3}$ consist of pairwise commuting unitaries,
\item $A_{i0}A_{i1}A_{i2}=(-1)^{\beta(C_i^\hor)}$ and  $A_{0j}A_{1j}A_{2j}=(-1)^{\beta(C_j^\ver)}$ for all $i,j\in \ZZ_3$.
\end{itemize} 
A quantum solution over $U(\CC)$ is called a {\it classical solution}. It is known that a classical solution exists if and only if $[\beta]=0$. This can   either be proved directly by an argument similar to Mermin's proof of contextuality \cite{mermin1993hidden}, or by   cohomological arguments \cite{Coho}. Nonexistence of a classical solution is an indication of quantum contextuality in the sense of Kochen--Specker. As in the case of Mermin's proof, quantum solutions can come from Pauli operators.
The {\it $n$-qubit Pauli operators}\footnote{We only consider the ones whose eigenvalues are $\pm 1$.} are given by
$$
A= A_1\otimes A_2\otimes \cdots \otimes A_n,\;\;\; A_i\in \set{ \one, X,  Y,  Z},
$$
where $\one,X,Y,Z$ are the $2\times 2$ Pauli matrices. The {\it Pauli group}, denoted by $P_n$, consists of operators of the form $i^{\alpha} A$ where $\alpha \in \ZZ_4$.

We partition the set of $2$-qubit Pauli operators into local and nonlocal parts:  
\begin{itemize}
\item Local $2$-qubit Pauli operators
$$
\set{X\otimes \one, Y\otimes \one,Z\otimes \one, \one\otimes X, \one\otimes Y,\one\otimes Z}.
$$
\item Nonlocal $2$-qubit Pauli operators
$$
\set{X\otimes X, X\otimes Y,X\otimes Z,Y\otimes X,Y\otimes Y,Y\otimes Z,Z\otimes X,Z\otimes Y,Z\otimes Z}.
$$
\end{itemize}
For $\beta_1$ defined as in Eq.~(\ref{eq:beta1})
 nonlocal   Pauli operators constitute a quantum solution; see Fig.~(\ref{fig:mermin-nonlocal}). 
For a pair $A,B$ of distinct and commuting Pauli operators let $\Pi_{AB}^{ab}$ denote the projector onto the simultaneous eigenspace corresponding to the eigenvalues $(-1)^a$ and $(-1)^b$ of $A$ and $B$; respectively. More concretely, we have $\Pi_{AB}^{ab} = (\one + (-1)^a A+(-1)^b B + (-1)^{a+b+\beta}AB)/4$.
These projectors constitute the set $\sS_2$ of $2$-qubit stabilizer states. 
There is a corresponding local vs nonlocal decomposition:
\begin{equation}\label{eq:stabilizer-loc-nloc}
\sS_2 = \sS_2^\loc \sqcup \sS_2^\nloc
\end{equation}
where 
\begin{itemize}
\item  $\sS_2^\loc$   consists of projectors $\Pi_{AB}^{ab}$ where $A,B$ are local Pauli operators. 
\item  $\sS_2^\nloc$  consists of projectors $\Pi_{AB}^{ab}$ where $A,B$ are nonlocal Pauli operators. 
\end{itemize}


\Lem{\label{lem:M1-quantum}
$\MP_1$ can be identified with the set of Hermitian operators $\rho\in \Herm((\CC^2)^{\otimes n})$ of trace $1$ such that $\Tr(\rho A)= 0$  for all   local Pauli's $A$ and $\Tr(\rho\Pi_{AB}^{ab})\geq 0$ for all pairwise commuting nonlocal Pauli operators $A,B$ and $a,b\in \ZZ_2$.
} 
\Proof{
Let $Q$ denote the set of operators $\rho$ described in the statement. The Born rule gives a map $p:Q\to \MP_1$ sending $\rho \mapsto p_\rho$. If we know $p_\rho^{ab}$ for all $a,b\in \ZZ_2$ then we can compute the expectation $\Span{A}_\rho$ for any nonlocal Pauli. Since by assumption $\Span{B}_\rho=0$ for every local Pauli $B$ this way we can determine $\rho$. In other words, we can define a map $e:\MP_1 \to Q$ by sending a distribution $d$ to the operator  
$$
\rho_d = \frac{1}{4} \left(  \one + \sum_{A} \alpha_A A \right)
$$  
where $A$ runs over nonlocal Pauli's and $\alpha_A$ is the expectation obtained from $d$. Then $e$ is the inverse of $p$. Therefore $p$ is a bijection.
}

Lemma \ref{lem:M1-quantum} provides a quantum mechanical description of $\MP_1$. In particular, some of the symmetries of $\MP_1$ come from quantum mechanics, that is, by conjugation with a Clifford unitary. 
The {\it $n$-qubit Clifford group} $\Cl_n$ is the quotient of 
the normalizer of $P_n$, the  group of unitaries $U\in U((\CC^{2})^{\otimes n})$ such that $U A U^\dagger \in P_n$ for all $A\in P_n$,
by the central subgroup $\set{e^{i\theta} \one:\, 0\leq \theta< 2\pi}$.  
Acting by the elements of $\Cl_1$ on each qubit 
preserves the set of nonlocal Pauli's. 
By Lemma \ref{lem:M1-quantum} this group acts on  the polytope $\MP_1$. 
Additionally, the $\SWAP$ gate that permutes the parties is also a symmetry of the polytope. Let us define the following subgroup of $\Cl_2$:
\begin{equation}\label{eq:G1}
G_1=\Span{\Cl_1\times \Cl_1,\SWAP}.
\end{equation}
As we observed this is also a subgroup of $\Aut(\MP_1)$. Next we will express $G_1$ as an extension similar to the one for $G_0$ given in Eq.~(\ref{eq:G0}). First recall that $\Cl_1$ has two parts: the Pauli part isomorphic to $\ZZ^2_2$ generated by conjugation with $X$ and $Z$, and the symplectic part $\Sp_2(\ZZ_2)$. 
The latter group is isomorphic to $\Sigma_3$ since in the single qubit case the symplectic action is determined by the permutation of the subgroups $\Span{-\one, X}$, $\Span{-\one, Y}$, $\Span{-\one, Z}$.  
We can express this decomposition as an exact sequence
\begin{equation}\label{eq:G1-exact}
0\to \ZZ_2^4 \to G_1 \to  \Aut(K_{3,3}) \to 1.
\end{equation}
The quotient is given by $\Aut(K_{3,3})$ since $G_1$ (up to signs) permutes the set of contexts which in return induces as action on the graph $K_{3,3}$. By comparing sizes we conclude that the quotient group is the whole automorphism group of the graph.

\Pro{\label{pro:G0-G1}
There is an isomorphism of groups $\phi:G_1\to G_0$.
}

\noindent The proof can be found in Appendix~\ref{sec:proof-g0-g1}.

\section{Vertices of the Mermin polytopes}\label{sec:mermin-vertices}

The description of Mermin polytopes is most naturally given in terms of the intersection of a finite number of half-spaces, or $H$-representation. However, by the Minkowski-Weyl theorem \cite{ziegler2012lectures} there is an equivalent representation of a polytope in terms of the convex hull of a finite number of vertices, called the $V$-representation. The problem of switching from the $H$ to the $V$-description is called the vertex enumeration problem; see e.g., \cite{chvatal1983linear}. Here we do precisely this and enumerate the vertices of $MP_{\beta}$, using the rich structure of these polytopes to aid in this task.
 
\subsection{Closed noncontextual subsets}

We recall some definitions  from \cite{raussendorf2020phase}.

\Def{\label{def:Omega}
A subset $\Omega \subset M$ is called {\it closed}  if $M\cap C\neq \emptyset$ implies $C\subset M$.
An {\it outcome assignment} on a closed subset $\Omega$ is a function $s:\Omega \to \ZZ_2$ such that
$$
s(c) = s(a) + s(b) +\beta(a,b)
$$
for all $C=\set{a,b,c}\subset \Omega$.
A closed subset $\Omega$ is called {\it noncontextual} if it admits an outcome assignment. We call subsets $\Omega$ obeying both of these properties closed noncontextual, or cnc sets. And as noted in \cite{raussendorf2020phase}, it suffices to consider just the maximal cnc sets, which we do from here forth.
}


\begin{figure}[h!]
\centering
\begin{subfigure}{\textwidth}
  \centering
  \includegraphics[width=.4\linewidth]{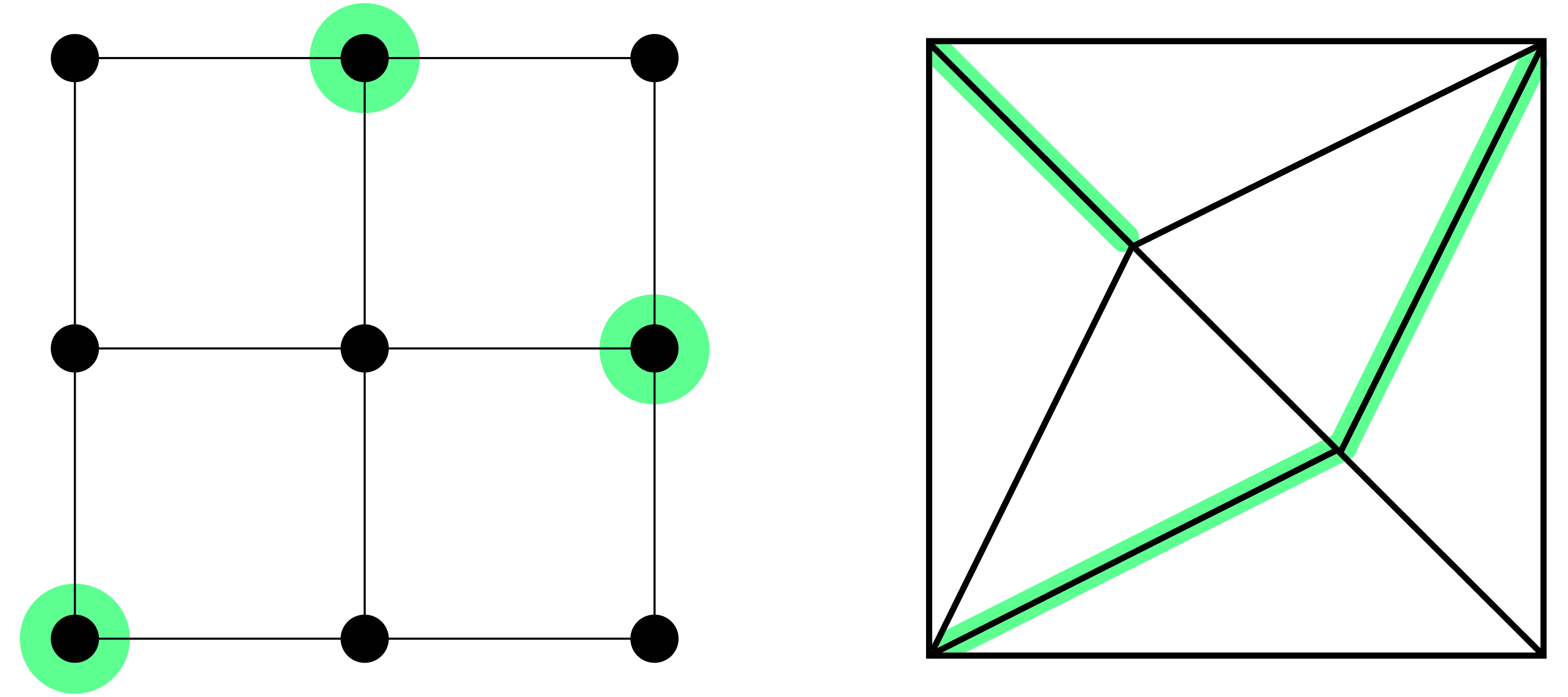}
  \caption{Type 1 cnc set: All edges correspond to anti-commuting operators, thus never in the same triangle.}
  \label{fig:cnc-type1}
\end{subfigure}\\%
\begin{subfigure}{\textwidth}
  \centering
  \includegraphics[width=.4\linewidth]{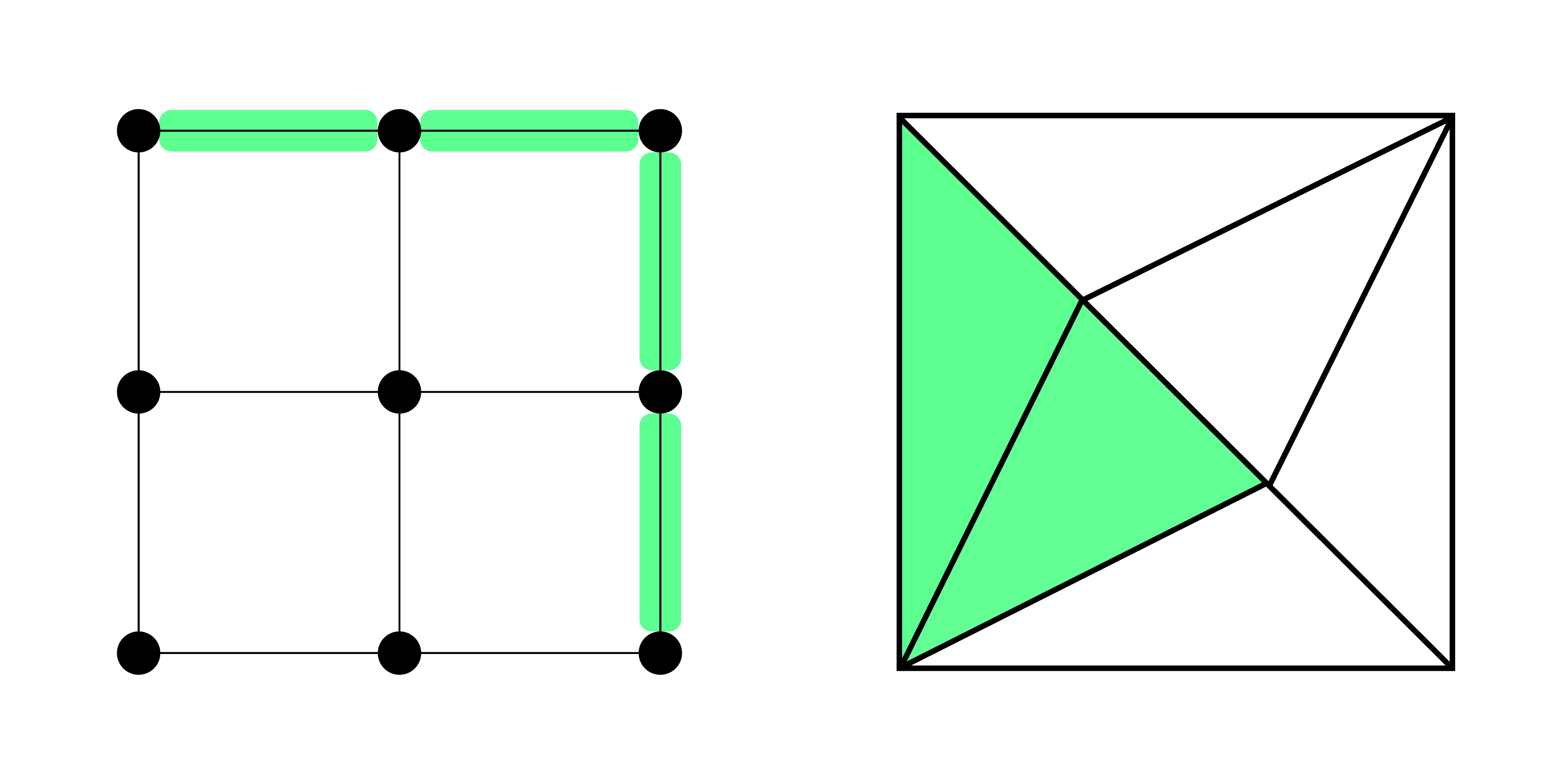}
  \caption{Type 2 cnc set: Two adjacent triangles.  }
  \label{fig:cnc-type2}
\end{subfigure}
\caption{
}
\label{fig:cnc}
\end{figure}

As we observed in the description of $\MP_1$ (see Section \ref{sec:MP1}) contexts of the Mermin scenario can be realized by the commutation relation of nonlocal Pauli operators. Next we make this connection more precise.
Recall that the measurements $m_{ij}\in M$ are labeled by pairs $(i,j)\in \ZZ_3^2$. 
We consider a map $\iota:\ZZ_3\to \ZZ_2^2-\set{0}$ defined as follows:
$$
\iota(0) = (0,1),\;\;\; \iota(1)=(1,1),\;\;\; \iota(2)=(1,0).
$$
The corresponding Pauli operators are $T_{(0,0)}=\one$, $T_{(0,1)}=X$, $T_{(1,1)}=Y$ and $T_{(1,0)}=Z$. 
On the other hand, $2$-qubit Pauli operators can be labeled by $(v,w)\in E=\ZZ_2^2\times \ZZ_2^2$, which corresponds to $T_v\otimes T_w$.
This way we obtain an embedding
\begin{equation}\label{eq:embedding-M-E}
M\to E,\;\;\; m_{ij}\mapsto (\iota(i),\iota(j)). 
\end{equation}
Throughout we will use this identification. Given $(v,w)$ and $(v',w')$ there is a symplectic form
$$
[(v,w),(v',w')] = v\cdot w' + v'\cdot w \mod 2.
$$
We say that such a pair {\it commutes} if $[(v,w),(v',w')]=0$; otherwise we say that they {\it anticommute}.
A subspace $I\subset E$ is called {\it isotropic} if each pair of elements in this subspace commute.
Observe that contexts in $\cC$ are precisely the maximal isotropic subspaces of $M$.

\Lem{\label{lem:Omega-structure}
The structure of maximal closed noncontextual subsets of $M$ with respect to $\beta_1$ is given as follows:
\begin{itemize}
\item {Type $1$: the subset $\Omega$ consists of three distinct pairwise anticommuting elements; i.e., none lie within the same context. We have $6$ such sets $\Omega$ and  an outcome assignment $s:\Omega \to \ZZ_2$ is  a function; see Fig.~(\ref{fig:cnc-type1}).}
\item  Type $2$: the subset $\Omega$ is a union of {two distinct contexts with a single measurement $m \in M$ lying on their (nonempty) intersection, and hence consists of $5$ elements total. There are $9$ such subsets $\Omega$, one for each $m\in M$. Additionally, there are $3$ elements that generate the set and an outcome assignment $s:\Omega \to \ZZ_2$ is determined by a function on these $3$ generators; see Fig.~(\ref{fig:cnc-type2}).}
\end{itemize}
}
\Proof{
{This result follows from Raussendorf, et al. \cite{raussendorf2020phase}, specifically Lemma 3. The case of the Mermin square is treated in more detail in Example 2. Quoting their results, there are $48 = 2^{3}\times 6$ Type 1 cnc sets and $72 = 2^{3} \times 9$ Type 2 cnc sets.}
}


\begin{figure}[h!]
\centering
\begin{subfigure}{.49\textwidth}
  \centering
  \includegraphics[width=.6\linewidth]{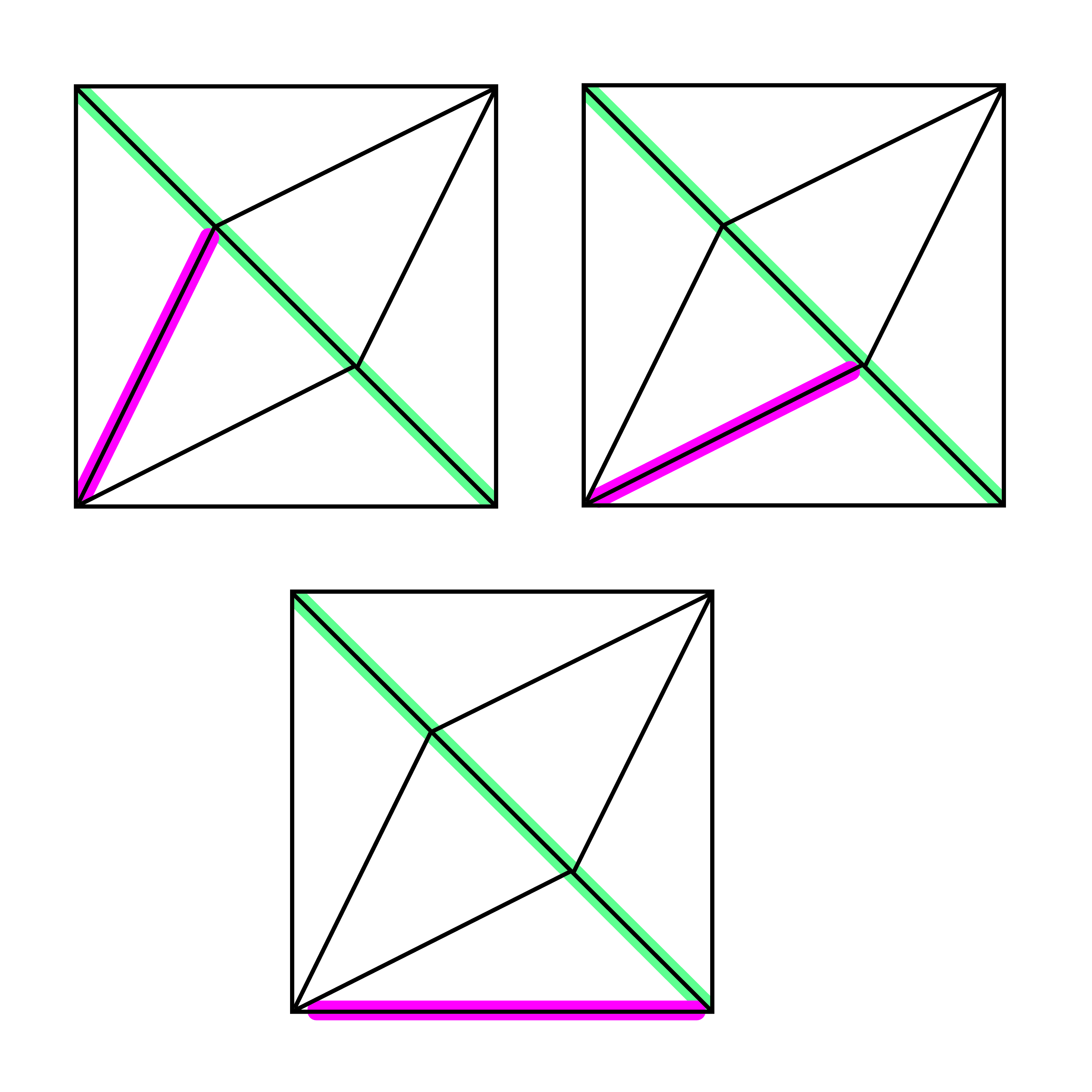}
  \caption{}
  \label{fig:cnc1-transitive}
\end{subfigure}%
\begin{subfigure}{.49\textwidth}
  \centering
  \includegraphics[width=.6\linewidth]{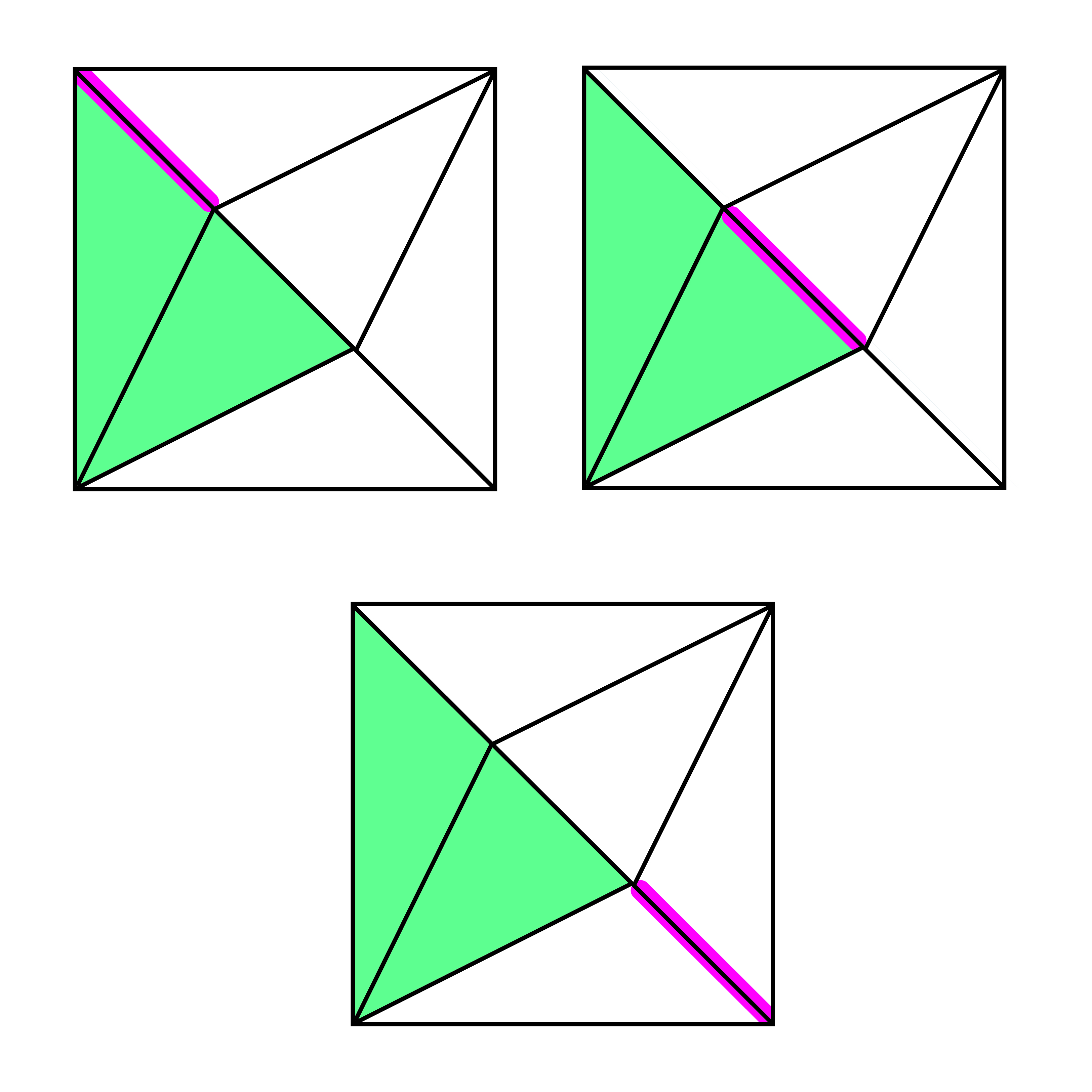}
  \caption{}
  \label{fig:cnc2-transitive}
\end{subfigure}
\caption{Type $1$ and $2$ cnc sets indicated by green color and the nonlocal Pauli's that flip only one outcome indicated by red color. 
}
\label{fig:cnc-transitive}
\end{figure}

\Lem{\label{lem:g1-transitive-vert}
$G_1$ acts transitively on the set
$$
\set{(\Omega,s):\, \Omega \text{ is a maximal cnc set of type }k\text{, and }s:\Omega\to \ZZ_2\text{ is an outcome assignment}}
$$
where 
$k$ is either $1$, or $2$.
}
\Proof{
We begin with type $1$ cnc sets. First let us ignore the outcome assignment. A type $1$ cnc set is specified by three pairwise anticommuting nonlocal Pauli operators. There are $6$ of these sets, three are of the form $\set{A\otimes X, A\otimes Y,A\otimes Z}$ and the remaining three are of the form $\set{X\otimes A, Y\otimes A,Z\otimes A}$, where $A=X,Y,Z$. To move from one such cnc set to another one one can use a local Clifford unitary and the SWAP gate if needed. Now including the outcome assignments, for a fixed cnc set $\Omega$ we can move from one outcome assignment to the other by flipping the signs of the outcomes   by conjugating with a nonlocal Pauli operator that commutes with two of them, but anticommutes with the remaining one; see Fig.~(\ref{fig:cnc1-transitive}).

In the case of type $2$ cnc sets we can move between the $\Omega$ sets since the $\Aut(K_{3,3})$ quotient of $G_1$ acts transitively on the edges of the torus (or the dual $K_{3,3}$ graph) \cite[Sec. 3.2]{godsil2001algebraic}. To move between the outcome assignments on a fixed type $2$ cnc set $\Omega$ we can conjugate with a nonlocal Pauli; see Fig.~(\ref{fig:cnc2-transitive}).  
}

\Cor{\label{cor:action-maximal}
$\Aut(K_{3,3})$ acts transitively on the set of maximal cnc sets of a fixed type.
}
\Proof{
As we observed in the proof of Lemma \ref{lem:g1-transitive-vert}  the $G_1$ action factors through the action of the quotient group $\Aut(K_{3,3})$ when the outcome assignments are ignores. 
}

\subsection{Vertex classification}
 
In the rest of this section we will prove the following result. Recall the embedding $M\subset E$ given in Eq.~(\ref{eq:embedding-M-E}).
 
\Thm{\label{thm:VertexClassification}
There is a bijection between the set of vertices of $\MP_\beta$ and the set of functions $s:\Omega\to \ZZ_2$ satisfying the following properties: 
\begin{enumerate}[(i)]
\item For $\MP_0$ the subset $\Omega=M$ and the functions are group homomorphisms $s:M\to \ZZ_2$. Each such function is given by specifying 
$s(v_i,w_j)\in \ZZ_2$ for $v_i,w_j\in \set{(0,1),(1,1)}$. 
In particular, there are $16$ vertices. 

\item For $\MP_1$ the subset $\Omega$ is a maximal closed noncontextual subset (one of the two types in Lemma \ref{lem:Omega-structure}) and $s:\Omega \to \ZZ_2$ is an outcome assignment.
There are two types of vertices: 
\begin{itemize}
\item Type $1$: When $\Omega$ is of type $1$. In particular, there are $48$ vertices of this type.
\item Type $2$: When $\Omega$ is of type $2$. In particular, there are $72$ vertices of this type.
\end{itemize}   
\end{enumerate}
Given a function $s:\Omega\to \ZZ_2$ as in (i) or (ii)  the corresponding vertex $p\in \MP_\beta$ is uniquely determined by 
{
\begin{eqnarray}\label{eq:prob-from-marginal}
p_{m}^{0}=%
\left\lbrace
\begin{array}{ll}
\frac{1+ (-1)^{s(m)}}{2}                   & \forall m\in \Omega\\
0                   & {\rm\text{otherwise}}
\end{array}
\right.~.%
\end{eqnarray}
}
}

We begin with some recollections from polytope theory; see \cite{chvatal1983linear, ziegler2012lectures}. Let $P(A,b)=\set{x:\, Ax \geq b}$ denote a polytope where $A\in \RR^{n\times m}$ and $b\in \RR^{n}$.  
Assume that $P(A,b)\subset \RR^m$ is full dimensional. 
Let us establish some terminology. If an inequality is satisfied with equality then we call that inequality \emph{tight}. For a point $p \in P(A,b)$ we refer to the \emph{active set} at $p$ as a subset $\mathcal{Z}_{p}\subset \{1,\cdots,n\}$ which indexes the set of tight inequalities at $p$. 
A 
point $p$
is a vertex of $P(A,b)$ if and only if there exists a subset of tight inequalities $Z \subseteq \mathcal{Z}_{p}$ with $|Z|=m$ such that
$$
v = A[Z]^{-1} b,
$$ 
where $A[Z]$ is the matrix obtained from $A$ by removing all the rows whose index is not in $Z$. Note that $|Z|\leq |\mathcal{Z}_{p}|$.

Let us apply these observations to Mermin polytopes $\MP_\beta$. We can express $\MP_\beta$ in the form $P(A,b)$. Let us
write $x\circ y$ to mean the XOR measurement $x\oplus y$ if $\beta(x,y)=0$, or the NOT of the XOR measurement $\overline{x\oplus y}$ if $\beta(x,y)=1$.

\Pro{\label{pro:P-MP}
The Mermin polytope $\MP_\beta$ has a description in the form $P(A,b)$ where $b=-\one_{9\times 1}$, $A \in \RR^{24\times 9}$ is a matrix whose rows are labeled by the set
$$
S=\set{(C,ab):\, C\in \cC,\; a,b\in \ZZ_2},
$$
columns labeled by $M$ (once both sets are ordered) and for $C=\set{x,y,x\circ y}$ its entries are given by
\begin{eqnarray}
A_{(C,ab),m} =
\left\lbrace
\begin{array}{ll}
(-1)^{a}                    & m=x\\
(-1)^{b}                    & m=y\\
(-1)^{a+b+\beta(C)}      & m=x\circ y\\
0                    & \text{otherwise.}
\end{array}
\right.%
\label{eq:row-A-matrix}
\end{eqnarray}
}

{\noindent Before we proceed to proving Proposition \ref{pro:P-MP}, let us prove the following useful lemma:
\Lem{\label{lem:dist-by-edge}
The distribution $p_{C}^{ab}$ in a single triangle (see Fig.~(\ref{fig:single-triangle})) with edges labeled $ \{x,y,z\}$ and outcomes $a,b,c \in \ZZ_2$, respectively, with $c = a+b+\beta_{C}$, is uniquely determined by the marginals along the edges $\left \{p_{\{x\}}^{a},p_{\{y\}}^{b},p_{\{z\}}^{c}\right \}$ according to
\begin{eqnarray}
p_{C}^{ab} = \frac{1}{2}\left(p_{\{x\}}^{a} + p_{\{y\}}^{b} - p_{\{z\}}^{c+1}   \right )~.%
\label{eq:dist-by-edge}
\end{eqnarray}
}

\Proof{%
First note that we have:
\begin{eqnarray}
p_{\{x\}}^{0} &=& p_{C}^{00}+p_{C}^{01}\label{eq:marginal-px}\\
p_{\{y\}}^{0} &=& p_{C}^{00}+p_{C}^{10}\label{eq:marginal-py}\\
p_{\{z\}}^{\beta_{C}} &=& p_{C}^{00}+p_{C}^{11}\label{eq:marginal-pz}
\end{eqnarray}
from which the normalization condition $\sum_{ab}p_{C}^{ab} = 1$ becomes:
\begin{eqnarray}
p_{C}^{00}+ p_{C}^{01}+ p_{C}^{10}+ p_{C}^{11} &=& p_{\{x\}}^{0} + p_{\{y\}}^{0} + p_{\{z\}}^{\beta_{C}} - 2p_{C}^{00} = 1~.\notag
\end{eqnarray}
From this we can obtain:
\begin{eqnarray}
p_{C}^{00} &=& \frac{1}{2} \left ( p_{\{x\}}^{0} + p_{\{y\}}^{0} + p_{\{z\}}^{\beta_{C}} -1  \right )\notag\\
&=& \frac{1}{2} \left ( p_{\{x\}}^{0} + p_{\{y\}}^{0} - p_{\{z\}}^{\beta_{C}+1}  \right )~,%
\label{eq:p00-edge}
\end{eqnarray}
where in the second line we used $p_{\{z\}}^{c+1} = 1 -  p_{\{z\}}^{c}$. Equation (\ref{eq:dist-by-edge}) then follows by inserting Eq.~(\ref{eq:p00-edge}) into Eqns.~(\ref{eq:marginal-px})-(\ref{eq:marginal-pz}) and solving for the remaining $p_{C}^{ab}$.
}

\Rem{%
{\rm Notice that if two contexts $C$ and $C'$ intersect on an edge $m = C\cap C^{\prime}$ then distributions $p_{C}^{ab}$, $p_{C^\prime}^{ab}$ represented as in Eq.~(\ref{eq:dist-by-edge}) will automatically satisfy the nonsignaling conditions if one and the same marginal $p_{\{m\}}^{a}$ is used in both.}
}

\noindent The $24$ probabilities $p_{C}^{ab}$ can therefore be uniquely expressed by the marginal probabilities $p_{\{m\}}^{a}$, where $m \in M$. In particular, any $p_{C}^{ab}$ can be expressed by just the $0$-outcome marginals $ p_{\{m\}}^{0}$ since their complement is given by $p_{\{m\}}^{1} = 1 - p_{\{m\}}^{0}$. These nine marginal probabilities therefore serve as a system of coordinates for $\MP_{\beta}$, which can be embedded in $\RR^{9}$.\\

\noindent We now introduce a new set of coordinates in terms of the \emph{expectation values} of the measurement outcomes, denoted $\bar{m}$. The two are related by an affine transformation. 
The expectation value of a measurement $m\in M$ is then given by $$\bar{m} := \sum_{a} (-1)^{a}p_{\{m\}}^{a} = p_{\{m\}}^{0} - p_{\{m\}}^{1}.$$ Using that $p^{1}_{\{m\}} = 1 - p_{\{m\}}^{0}$ and solving for $p_{\{m\}}^{0}$ we obtain the desired relationship
\begin{eqnarray}
p_{\{m\}}^{0} = \frac{1}{2}\left ( 1+ \bar{m}\right ).\label{eq:expectation-marginal}
\end{eqnarray}\\


\begin{proof}[Proof of Proposition~\ref{pro:P-MP}]
{Note that $\MP_{\beta}$ is defined as the intersection of the half-space inequalities $p_{C}^{ab} \geq 0$ intersected by the affine subspace generated by the nonsignaling conditions and normalization. By Lemma~\ref{lem:dist-by-edge} this is equivalent to requiring the nonnegativity of Eq.~(\ref{eq:dist-by-edge}), for every $C\in\mathcal{C}$ and $a,b\in\ZZ_{2}$. Plugging in Eq.~(\ref{eq:expectation-marginal}) for $p_{\{m\}}^{0}$ in terms of $\bar{m}$ gives us the expression
\begin{eqnarray}\label{eq:pCab}
p_{C}^{ab} = \frac{1}{4}\left (1 + (-1)^{a}\bar{x} + (-1)^{b}\bar{y} + (-1)^{a+b+\beta}\bar{z}  \right ).
\end{eqnarray}
Requiring nonnegativity and rearranging yields
\begin{eqnarray}
A_{(C,ab)}^{T} \, x \geq -1
\end{eqnarray}
where $A_{(C,ab)}$ is given as in Eq.~(\ref{eq:row-A-matrix}) and $x \in \RR^{9}$ has components $\bar{m}$, where $m\in M$. The $24$ inequalities defining the polytope  can now be compactly expressed as $A \, x \geq b $, which concludes the proof.
}
\end{proof}
}




Let $Z \subset S$ be a subset of indices such that $|Z| = 9$. For each $C\in \cC$ let us write $n(C)=|Z\cap \set{(C,ab):\, a,b\in \ZZ_2}|$. The numbers $n(C)$ satisfy the following properties:
\begin{itemize}
\item $\sum_{C\in \cC} n(C)=9$ since $|Z|=9$.
\item $0\leq n(C)\leq 3$ since $\sum_{a,b\in \ZZ_2}p_C^{ab}=1$ for each context.
\end{itemize} 
Our case classification will be in terms of the following numbers:
$$
n_k = |\set{C\in \cC:\, n(C)=k} | 
$$
Table (\ref{tab:zero-cases}) displays all the cases that can occur. These cases will be denoted by $(n_3,n_2)$.\\
 
 \begin{table}[h!] 
\centering
{\renewcommand{\arraystretch}{1.5}%
\begin{tabular}{|c  c  c |}
\hline
$n_3$ & $n_2$ & $n_1 $ \\
\hline
$3$ & $0$ & $0$ \\
\hdashline
$2$ & $1$ & $1$ \\
$2$ & $0$ & $3$ \\
\hdashline
$1$ & $3$ & $0$ \\
$1$ & $2$ & $2$ \\
$1$ & $1$ & $4$ \\
\hdashline
$0$ & $3$ & $3$ \\
\hline%
 \end{tabular}}%
\caption{Each row displays the number of contexts with the indicated number of zeros. These are the triples $(n_3,n_2,n_1)$ satisfying $0\leq n_i\leq 3$ and $3n_3+2n_2+n_1=9$.  
}\label{tab:zero-cases}
\end{table}

\noindent A triangle  representing a context $C$ is called a {\it deterministic triangle} if $p_C$ is a deterministic distribution.
An edge labeled by a measurement $x\in C$ is called a {\it deterministic edge} if $p_C|_{\set{x}}$ is a deterministic distribution.

\Lem{\label{lem:two-edges}
A triangle $C$ with two deterministic edges is deterministic.
$$
\centering
  \includegraphics[width=.4\linewidth]{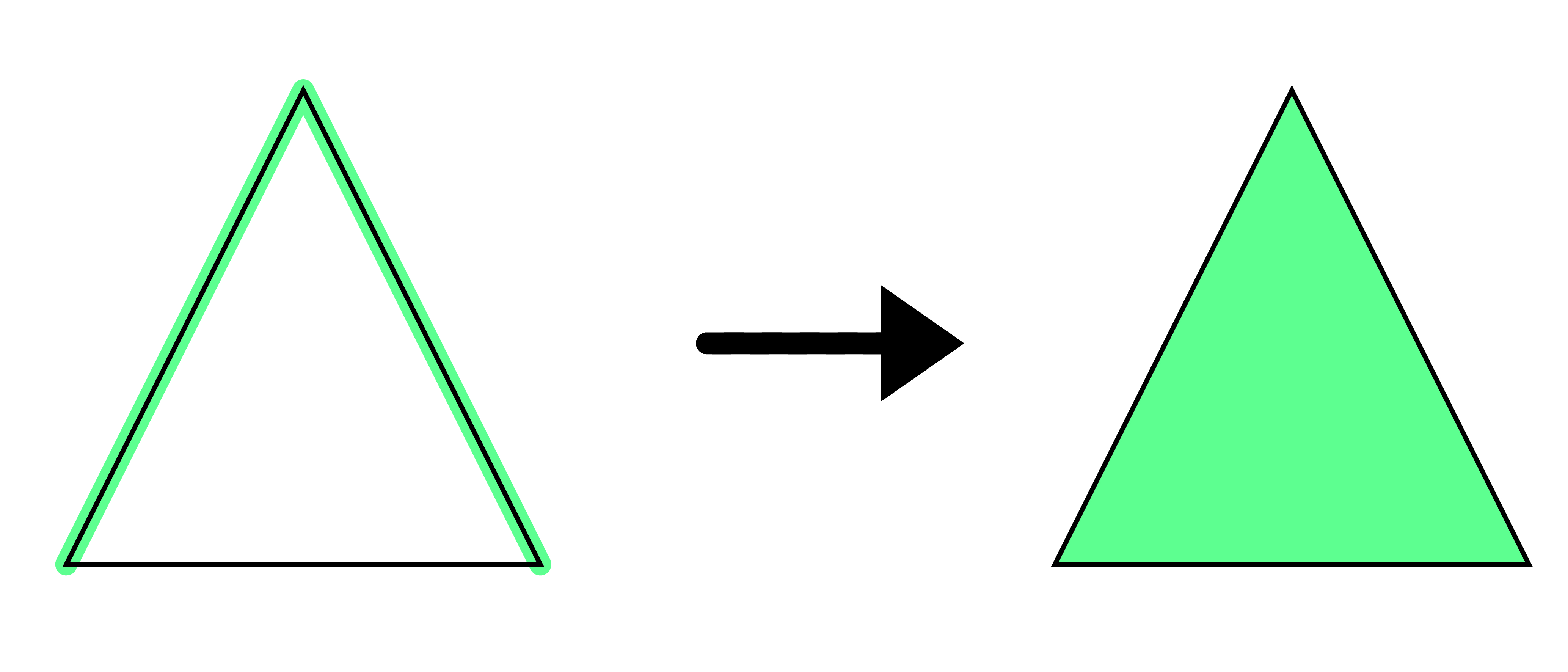}
$$
}  
\Proof{We can assume $\beta=0$ on the triangle, the case $\beta=1$ is treated similarly.
Let $p_C=\set{p_C^{ab}}_{a,b}$ be a distribution on the triangle where $C=\set{x,y,x\oplus y}$. 
Assume that $p|_{x}=\delta^a_x$ and $p|_{y}=\delta^b$ for some $a,b\in \ZZ_2$. This implies
$$
p^{\bar a0}+p^{\bar a1}=p^{0\bar b}+p^{1\bar b} =0,
$$
where $\bar{a} = a +1$. In every case three of the four probabilities are zero giving us a deterministic distribution. Other cases where $x,x\oplus y$ and $y,x\oplus y$ are deterministic are also treated similarly
}

\Lem{\label{lem:det-triangle-rank}
Let $p$ be a distribution on a single triangle $C$.
Let $Z_{C} = \{A_{(C,ab)}:\,a,b\in \ZZ_{2}\}$ be the set of tight inequalities.
Then $rank(A[Z_{C}]) = |Z_C|$.  
}
\Proof{
We can assume $\beta=0$, the case $\beta=1$ is similar.
Assume that $Z_C$ is nonempty, otherwise the rank is zero.
Let us write $C=\set{x,y,x\oplus y}$.
We will use the symmetry group $G_C$ generated by flipping the outcomes of $x,y$, which is isomorphic to $\ZZ_2^2$.
By Lemma \ref{lem:two-edges} there are two cases (up to $G_C$ action):
\begin{itemize}
\item Single deterministic edge: $p|_{\set{x}}=\delta^0$. In this case
\begin{eqnarray}
\rank(A) =
\rank
\begin{bmatrix}
1 & -1 & -1 \\
1 & 1  & 1
\end{bmatrix} =2,
\end{eqnarray}
where a row corresponds to outcomes of $(x,y,x\oplus y)$ in this order.
\item Deterministic triangle: $p=\delta^{00}$. In this case
\begin{eqnarray}
\rank(A) =
\rank
\begin{bmatrix}
1 & -1 & -1 \\
-1 & 1  & -1 \\
-1  & -1 & 1
\end{bmatrix} =3.
\end{eqnarray}
\end{itemize}
}

Next, we consider distributions on the {\it diamond} $D$, which is obtained by gluing two triangles, $C$ and $C'$, along a common edge. We will denote the common edge by $z$. See Fig.~(\ref{fig:rank-diamond}).

\Lem{\label{lem:rank-two-triangles}
Let $p$ be a distribution on the diamond $D$.  
Let $Z_{C}$, $Z_{C'}$ be the set of tight inequalities in triangles $C$, $C'$; respectively. Define $Z_{C,C'} = Z_{C}\cup Z_{C'}$.
Then
$$
\rank(A[Z_{C,C'}]) = 
\left\lbrace
\arraycolsep=3.5pt\def\arraystretch{1.5}
\begin{array}{ll}
|Z_{C,C'}|-1 &  z\text{ is deterministic,}\\
|Z_{C,C'}|  & \text{otherwise.} 
\end{array}
\right.
$$
}
\Proof{
We assume $\beta=0$ on both triangles, the case $\beta=1$ on one of the triangles is treated similarly.
Assume that $Z_{C,C'}$ is nonempty, otherwise the rank is zero.
We will also use symmetries to reduce the number of cases.
Let us write $C=\set{x,z,x\oplus z}$ and $C'=\set{x',z,x' \oplus z}$ for the contexts.
Let $G_D$ denote the symmetry group generated by flipping the outcomes of $x,z,x'$, which is isomorphic to $\ZZ_2^3$.
First, let us assume that $z$ is not deterministic. 
The cases where either $Z_C$ or $Z_{C'}$ are empty can be deduced from Lemma \ref{lem:det-triangle-rank}.
By Lemma \ref{lem:two-edges} the remaining cases are as follows (up to $G_D$ action):
\begin{itemize}
\item $|Z_C|=|Z_{C'}|=1$ with $p_C^{00}=p_{C'}^{00}=0$. In this case
\begin{eqnarray}
\rank(A[Z_{C,C'}]) =
\rank
\begin{bmatrix}
1 & 1 & 1 & 0 & 0 \\
0 & 0 & 1 & 1 & 1
\end{bmatrix} =2,
\end{eqnarray}
where a row corresponds to outcomes of $(x, x\oplus z, z, x',x'\oplus z)$ in this order.
\item $|Z_C|=2$, $|Z_{C'}|=1$ with $p_C|_{\set{x}}=\delta^0$ and $p_{C'}^{00}=0$. In this case
\begin{eqnarray}
\rank(A[Z_{C,C'}]) =
\rank
\begin{bmatrix}
1 & 1 & 1 & 0 & 0 \\
1 & -1 & -1 & 0 & 0 \\
0 & 0 & 1 & 1 & 1 \\
\end{bmatrix} =3.
\end{eqnarray}
The case $|Z_C|=1$, $|Z_{C'}|=2$ is similar. 
\item $|Z_C|=2$, $|Z_{C'}|=2$ with $p_C|_{\set{x}}=p_{C'}|_{\set{x'}}=\delta^0$. In this case
\begin{eqnarray}
\rank(A[Z_{C,C'}]) =
\rank
\begin{bmatrix}
1 & 1 & 1 & 0 & 0 \\
1 & -1 & -1 & 0 & 0 \\
0 & 0 & 1 & 1 & 1 \\
0 & 0 & -1 & 1 & -1
\end{bmatrix} =4. 
 \end{eqnarray}
\noindent
Next, we consider the case where $z$ is deterministic. Again up to the action of $G_D$  we have the following cases:
\item $|Z_C|=2$, $|Z_{C'}|=2$ with $p_C|_{\set{z}}=p_{C'}|_{\set{z}}=\delta^0$. In this case
\begin{eqnarray}
\rank(A[Z_{C,C'}]) =
\rank
\begin{bmatrix}
1 & 1 & 1 & 0 & 0 \\
-1 & -1 & 1 & 0 & 0 \\
0 & 0 & 1 & 1 & 1 \\
0 & 0 & 1 & -1 & -1
\end{bmatrix} =3.%
\label{eq:rank-one-det}
 \end{eqnarray}
\item $|Z_C|=3$, $|Z_{C'}|=2$ with $p_C=\delta^{00}$ and
$p_{C'}|_{\set{z}}=\delta^0$. In this case
\begin{eqnarray}
\rank(A[Z_{C,C'}]) =
\rank
\begin{bmatrix}
1 & -1 & -1 & 0 & 0 \\
-1 & 1 & -1 & 0 & 0 \\
-1 & -1 & 1 & 0 & 0 \\
0 & 0 & 1 & 1 & 1 \\
0 & 0 & 1 & -1 & -1
\end{bmatrix} =4. 
 \end{eqnarray}
The case $|Z_C|=2$, $|Z_{C'}|=3$ is similar. 
\item $|Z_C|=3$, $|Z_{C'}|=3$ with $p_C=p_{C'}=\delta^{00}$. In this case
\begin{eqnarray}
\rank(A[Z_{C,C'}]) =
\rank
\begin{bmatrix}
1 & -1 & -1 & 0 & 0 \\
-1 & 1 & -1 & 0 & 0 \\
-1 & -1 & 1 & 0 & 0 \\
0 & 0 & -1 & 1 & -1 \\
0 & 0 & 1 & -1 & -1 \\
0 & 0 & -1 & -1 & 1 
\end{bmatrix} =5. 
 \end{eqnarray}
\end{itemize}
}

\begin{figure}[h!]
\centering
\begin{subfigure}{.49\textwidth}
  \centering
  \includegraphics[width=.3\linewidth]{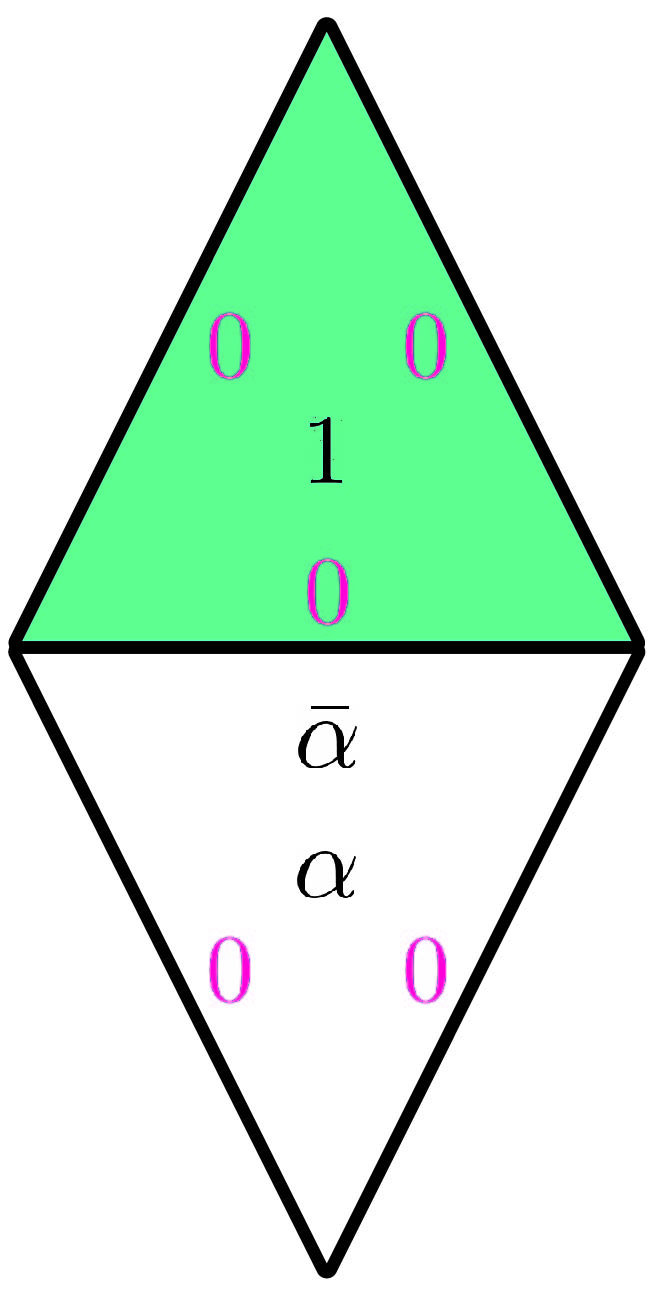}
  \caption{}
  \label{fig:rank-diamond}
\end{subfigure}%
\begin{subfigure}{.49\textwidth}
  \centering
  \includegraphics[width=.55\linewidth]{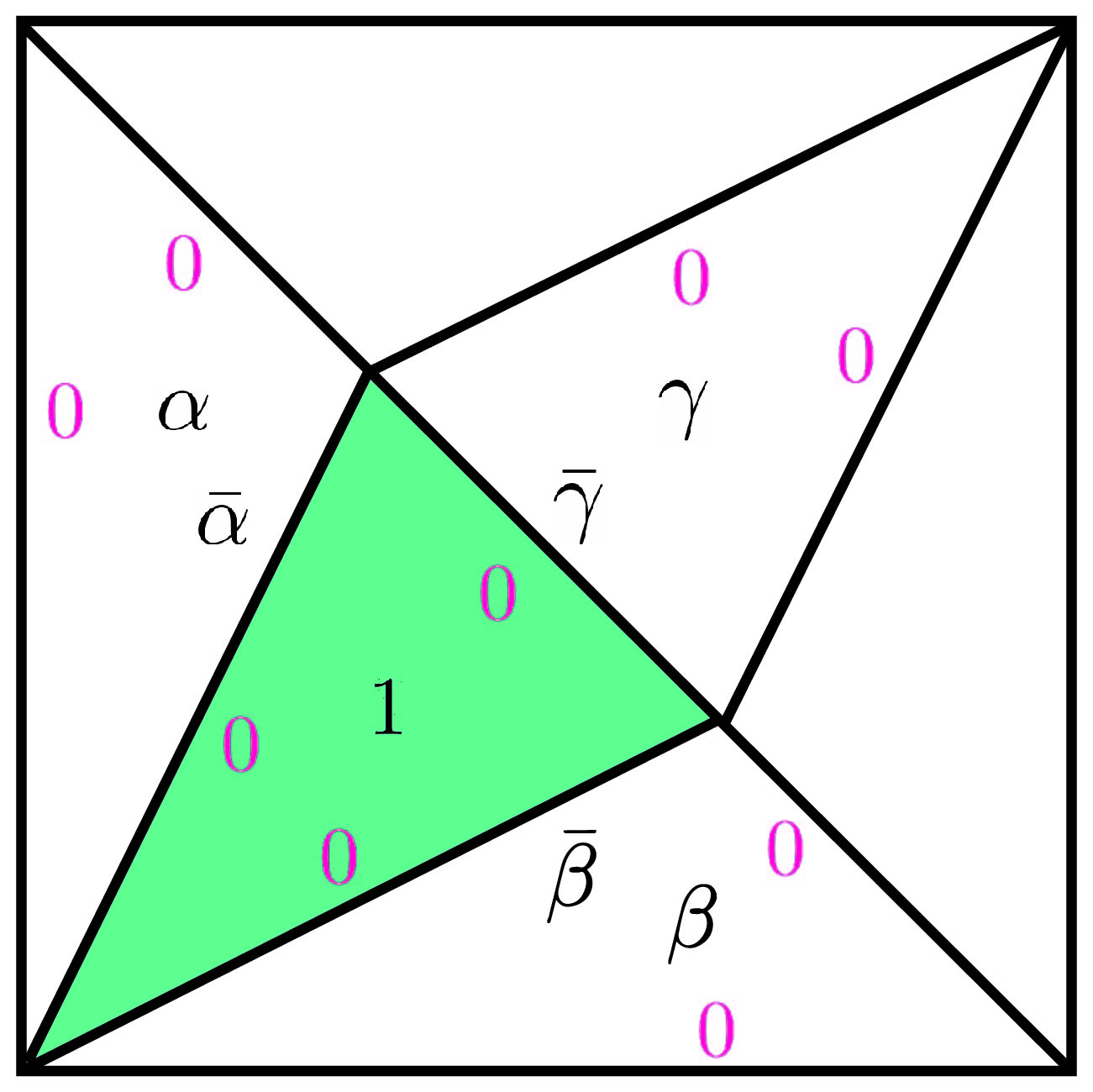}
  \caption{}
  \label{fig:rank-one-triangle}
\end{subfigure}
\caption{(a) Two triangles glued along a common edge. (b) A distribution on the torus with one deterministic triangle.
}
\label{fig:rank-stuff}
\end{figure}

\Lem{\label{lem:successive-rank-lem}
Let $p$ be a distribution in $\MP_\beta$ and  $Z$ denote the set of tight inequalities.  Assume that there exists a deterministic triangle $C$.
Then $rank(A[Z]) \geq 6$.
}
\Proof{%
Considering the action of $G_\beta$ will simplify the discussion. Our argument does not depend on $\beta$, so we assume $\beta=0$. Up to the action of the symmetry group we can assume that $p_C$ has the form given in 
Fig.~(\ref{fig:rank-one-triangle}). Let us write $C=\set{x,y,z}$ and $C_t=\set{t,t',t\oplus t'}$ where $t=x,y,z$ for the adjacent triangles. Then
\begin{eqnarray}
\rank(A[Z]) =
\rank
\begin{bmatrix}
1 & -1 & -1 & 0 & 0 & 0 & 0 & 0 & 0\\
-1 & 1 & -1 & 0 & 0 & 0 & 0 & 0 & 0\\
-1 & -1 & 1 & 0 & 0 & 0 & 0 & 0 & 0\\
-1 & 0 & 0 & 1 & -1 & 0 & 0 & 0 & 0\\
-1 & 0 & 0 & -1 & 1 & 0 & 0 & 0 & 0\\
0 & -1 & 0 & 0 & 0 & 1 & -1 & 0 & 0\\
0 & -1 & 0 & 0 & 0 & -1 & 1 & 0 & 0\\
0 & 0 & -1 & 0 & 0 & 0 & 0 & 1 & -1\\
0 & 0 & -1 & 0 & 0 & 0 & 0 & -1 & 1\\
\end{bmatrix} =6,
 \end{eqnarray} 
where a row corresponds to outcomes of $(x,y,z,x',x\oplus x',y',y\oplus y',z',z\oplus z')$. 
Lemma \ref{lem:det-triangle-rank} and Lemma \ref{lem:rank-two-triangles} can be used to compute the rank. We are looking at a region obtained by gluing three diamonds along a triangle. The rank of the matrix above is the sum of the ranks of each diamond minus two times the rank of the deterministic triangle.
 
}

\Rem{\label{rem:two-adjacent}
{\rm 
Lemma~\ref{lem:successive-rank-lem} implies that
to obtain a vertex we must fix three additional (linearly independent) zeros. This forces at least one additional edge to be deterministic, which by Lemma~\ref{lem:two-edges} implies that we must have at least two adjacent deterministic triangles.
}
}

We will use the number of deterministic triangles and the number of deterministic edges (that lie out side the boundary of the triangles) to organize the cases. 
In Fig.~(\ref{fig:MerminCases}) we see a diagram that illustrates all the possibilities. The base cases consist of three deterministic edges. Successive application of Lemma \ref{lem:two-edges} together with Lemma~\ref{lem:successive-rank-lem} and Remark~\ref{rem:two-adjacent} reduces the diagram to three main cases:
\begin{enumerate}
\item[(C1)] All triangles are deterministic.
\item[(C2)] Two adjacent deterministic triangles.
\item[(C3)] Three anticommuting deterministic edges.
\end{enumerate}
Note that if there are more than three deterministic edges again Lemma  \ref{lem:two-edges} can be used to reduce to (C1). The case where there is only one deterministic triangle and no additional deterministic edges does not appear since Lemma \ref{lem:rank-two-triangles} 
implies that such a configuration can have rank at most $8$.

\Rem{\label{rem:rep}%
{\rm Note that up to $\text{Aut}(K_{3,3})$ there are only three representative cases. For (C1) this is obvious since all triangles are deterministic. For cases (C2) and (C3) we observe that these correspond to type $2$ and type $1$ cnc sets, respectively. Thus by Corollary~\ref{cor:action-maximal} it suffices to consider a single representative for each case. The representatives are given in  Fig.~(\ref{fig:MerminCases}).
}
}

\begin{figure}[h!]
\centering
  \includegraphics[width=0.7\linewidth]{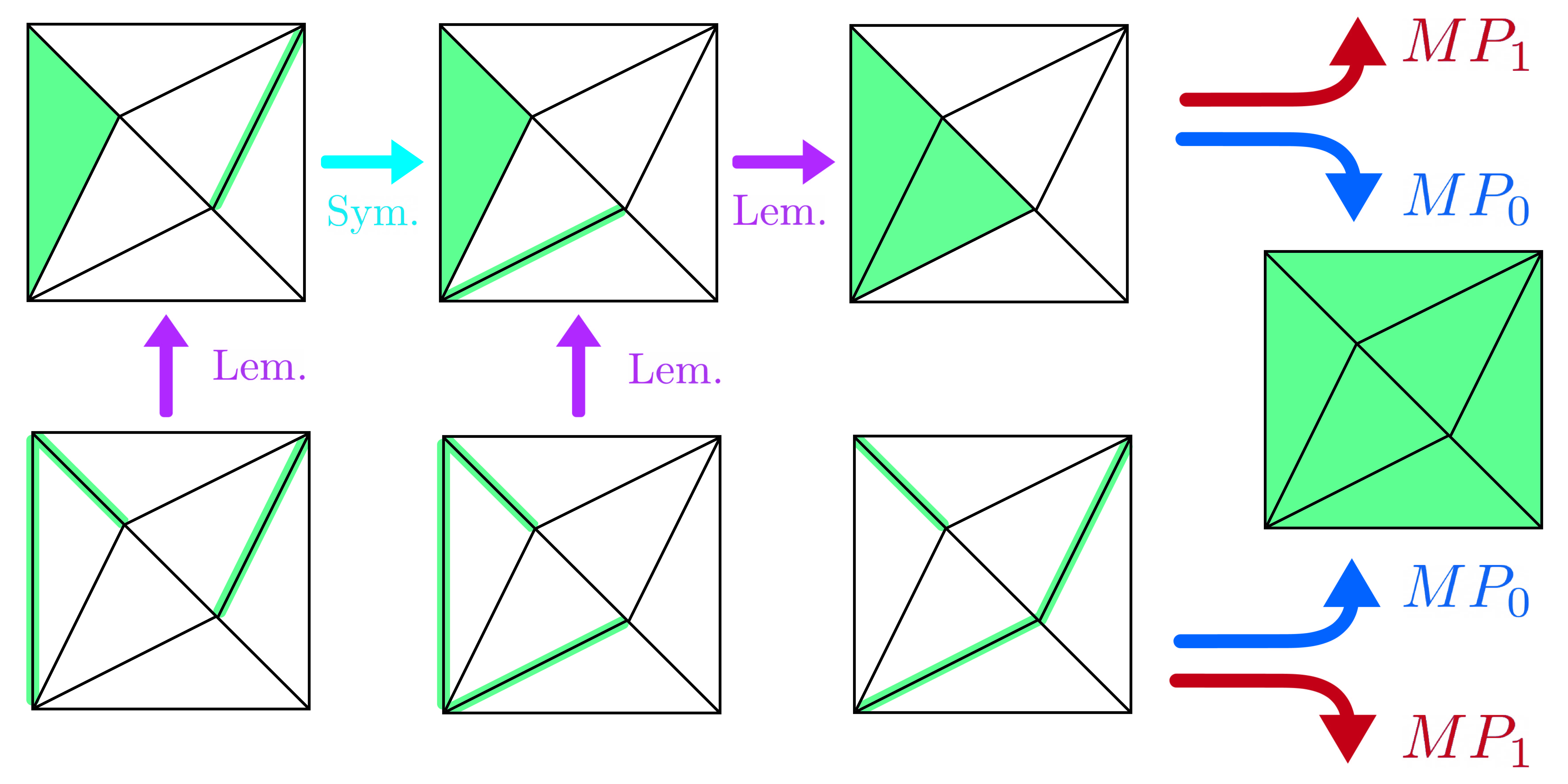}
\caption{Starting from the base case of three deterministic edges, we can obtain the vertices of $MP_{\beta}$ by repeated application of Lemma~\ref{lem:two-edges}. Up to symmetry we have three representative cases: (C1) All deterministic triangles (C2), two adjacent deterministic triangles, and (C3) three anti-commuting deterministic edges. For $\beta = 0$ all of these cases lead to a deterministic distribution. For $\beta = 1$ deterministic distributions are not allowed and we have type 1 and type 2 cnc distributions corresponding to cases (C2) and (C3), respectively.   
}
\label{fig:MerminCases}
\end{figure}



\Lem{\label{lem:C1}
Assume $p\in \MP_\beta$ is a vertex that satisfies (C1). Then $p$ belongs to $\MP_0$. No distribution in $\MP_1$ is deterministic.
}
{
\Proof{First let us note that (C1) implies that $A[Z]$ has full rank. To see this, take three mutually nonadjacent (i.e., the set of edges $C_{i}\cap C_{j}$ is empty) triangles as deterministic, which implies that all edges are deterministic. By applying Lemma~\ref{lem:det-triangle-rank} for each triangle (after an appropriate permutation of columns) we have that $A[Z]$ 
has full rank. Next observe that a set of deterministic edges 
implies a classical solution to the binary linear system $(M,\mathcal{C},\beta)$. Since this is possible only for $[\beta] = 0$, we have (C1) defines a vertex of $\MP_{0}$, but not of $\MP_{1}$.
}
}

\Lem{\label{lem:C2}
Let $p$ be a vertex of $\MP_\beta$ that satisfies (C2).
Then either
\begin{enumerate}[(i)]
\item $p$ satisfies (C1), or
\item $p$ is a type $2$ vertex of $\MP_1$.
\end{enumerate}
}
{
\Proof{
Consider a configuration for (C2), which specifies a cnc set $\Omega$ of type $2$. 
We can fix a deterministic distribution on $\Omega$. Any other choice can be dealt with similarly by the help of Lemma~\ref{lem:g1-transitive-vert}. 
Then we use the compatibility conditions, as shown below:
$$
\centering
  \includegraphics[width=.6\linewidth]{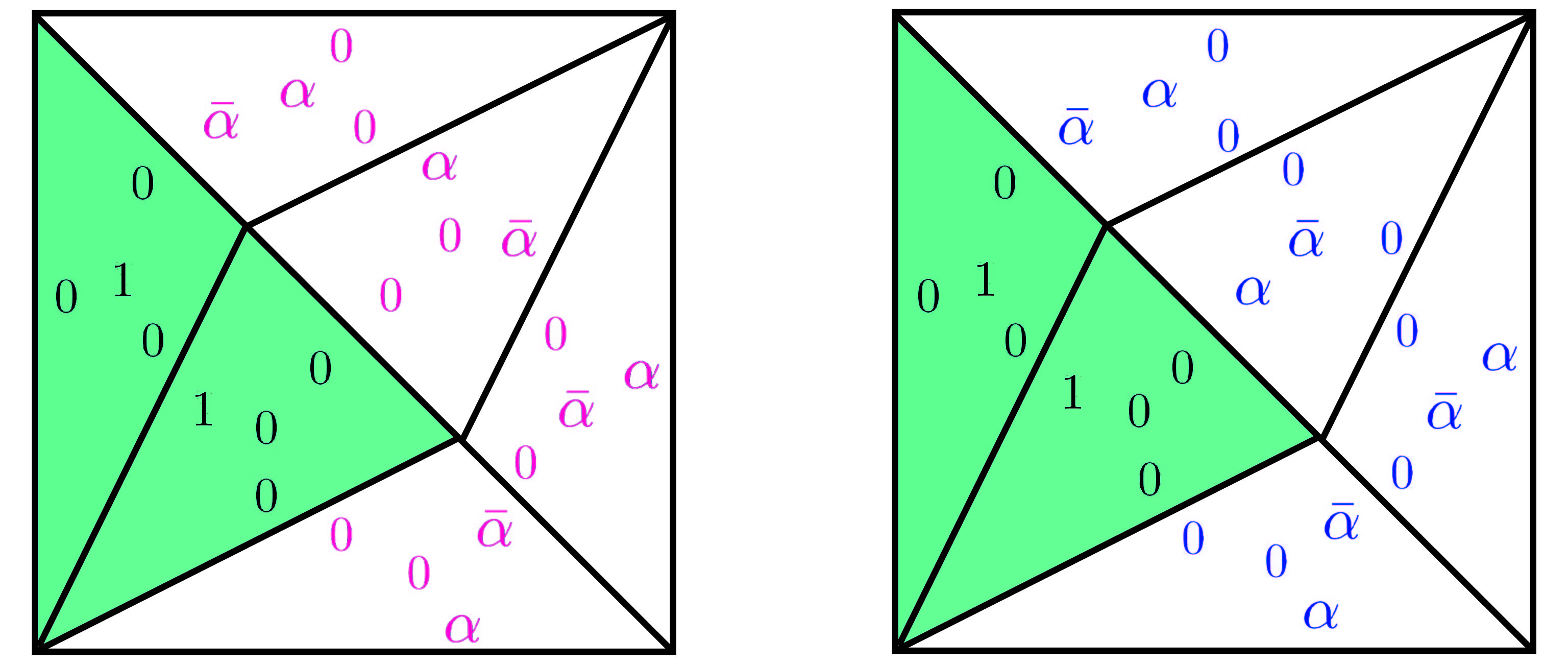}
$$
Here we have $\beta_{0}$ ($\beta_{1}$) in red (blue). For $\beta = 0$ there is a one parameter family of distributions. A vertex is specified by choosing $\alpha = 0,1$, which implies a deterministic distribution and thus reduces to (C1). For $\beta = 1$ the compatibility conditions imply that $\alpha = 1-\alpha = 1/2$. 
}
}

\Lem{\label{lem:C3}
Let $p$ be a vertex of $\MP_\beta$ that satisfies (C3).
Then either
\begin{enumerate}[(i)]
\item  $p$ satisfies (C1), or 
\item $p$ is a type $1$ vertex of $\MP_1$. 
\end{enumerate}
}  
{
\Proof{Similar to the case (C2) let us consider a configuration, choose a convenient distribution consistent with the case (C3) (other choices can be handled using symmetry, i.e., Lemma~\ref{lem:g1-transitive-vert}), and solve for the probabilities using the compatibility conditions:
$$
\centering
  \includegraphics[width=.6\linewidth]{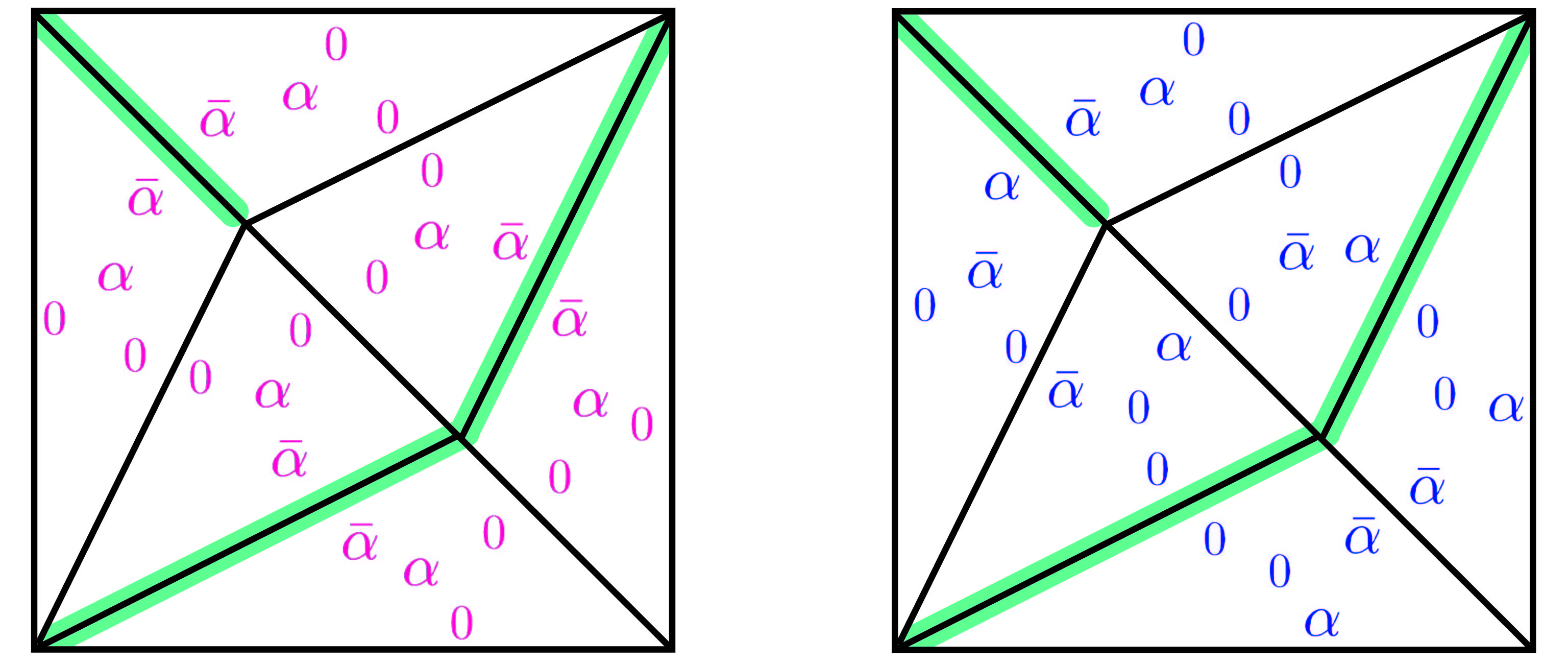}
$$
As with (C2), here for $\beta = 0$ we have a one-parameter family of distributions {(red)} where vertices are specified by $\alpha = 0,1$, reducing to the deterministic case (C3). For $\beta =1$ we have that $\alpha = 1-\alpha = 1/2$ (blue). 
}
} 
  
\begin{proof}[{\bf Proof of Theorem \ref{thm:VertexClassification}}]
To begin, the diagram in Fig.~(\ref{fig:MerminCases}) implies that we need only consider cases (C1)-(C3). Focusing first on $\beta = 0$, Lemmas~\ref{lem:C1}-\ref{lem:C3} imply that all vertices of $\MP_{0}$ are deterministic. 
These vertices are determined by the marginals on the measurements $(v_i,w_j)$ where $v_i,w_j\in \set{(0,1),(1,1)}$. Hence there are $16$ such vertices.

Turning now to $\MP_{1}$, note that by Lemma~\ref{lem:C1} that no deterministic distribution is a point of $\MP_{1}$, thus by Lemmas~\ref{lem:C2} and \ref{lem:C3}, the only vertices of $\MP_{1}$ are those of the form of (C2) and (C3). Observe that for (C2) and (C3) that $p_{\{m\}}^{0}\in \{0,1\}$ (i.e., the edge is deterministic) if and only if $m\in\Omega$, where $\Omega$ are the maximal cnc sets described in Lemma~\ref{lem:Omega-structure}, and $p_{\{m\}}^{0}=1/2$  (or $\bar{m} = 0$) for all other observables $m \notin \Omega$. For example, the deterministic edges in (C3) are described by a type $1$ cnc set since they correspond to a maximal set of anti-commuting observables. Using Lemma~\ref{lem:Omega-structure}, we know that there are $48 = 2^{3}\times 6$ type 1 and $72 = 2^{3}\times 9$ type $2$  cnc sets, which then correspond to $48$ type $1$ and $72$ type $2$ vertices of $\MP_{1}$, respectively.
\end{proof}

\section{Graph of the Mermin polytopes} 
\label{sec:Graph}

In this section, we determine the graph of $\MP_\beta$ consisting of the vertices of the polytope together with the edges connecting two neighbor vertices in the polytope.

\subsection{Graph of $\MP_0$} 

\Lem{\label{lem:g0-transitive}%
$G_{0}$ acts transitively on deterministic vertices of $\MP_{0}$. 
}
\Proof{%
Take an arbitrary deterministic vertex $p\in \MP_{0}$ and act on it by $G_{\ell}\subset G_{0}$. There are $15$ elements of $G_{\ell}$ listed in Fig.~(\ref{fig:mermin-torus-loops-a}) and the action of each permutes the outcomes of a different subset of measurements and thus generates $15$ distribution distinct from $p$. Since there are $16  = 2^{4}$ outcome assignments 
in total, we obtain all possible deterministic distributions by the action of $G_{0}$
} 

Let $q$ denote the deterministic distribution in $\MP_0$ given by $q_C^{00}=1$ for all triangles $C$; see Fig.~(\ref{fig:can-det}). We will take this as the canonical vertex of this polytope. The other vertices can be obtained by using the action of the loops as a consequence of Lemma \ref{lem:g0-transitive}. We will write 
$$
q_l = g_l\cdot q\;\;\text{ where } g_l\in G_\ell,\; l\in \ell(K_{3,3})
$$
for the remaining vertices obtained via the action of $G_\ell$ (see Eq.~(\ref{eq:G0})).

\begin{figure}[h!]
\centering
  \includegraphics[width=0.3\linewidth]{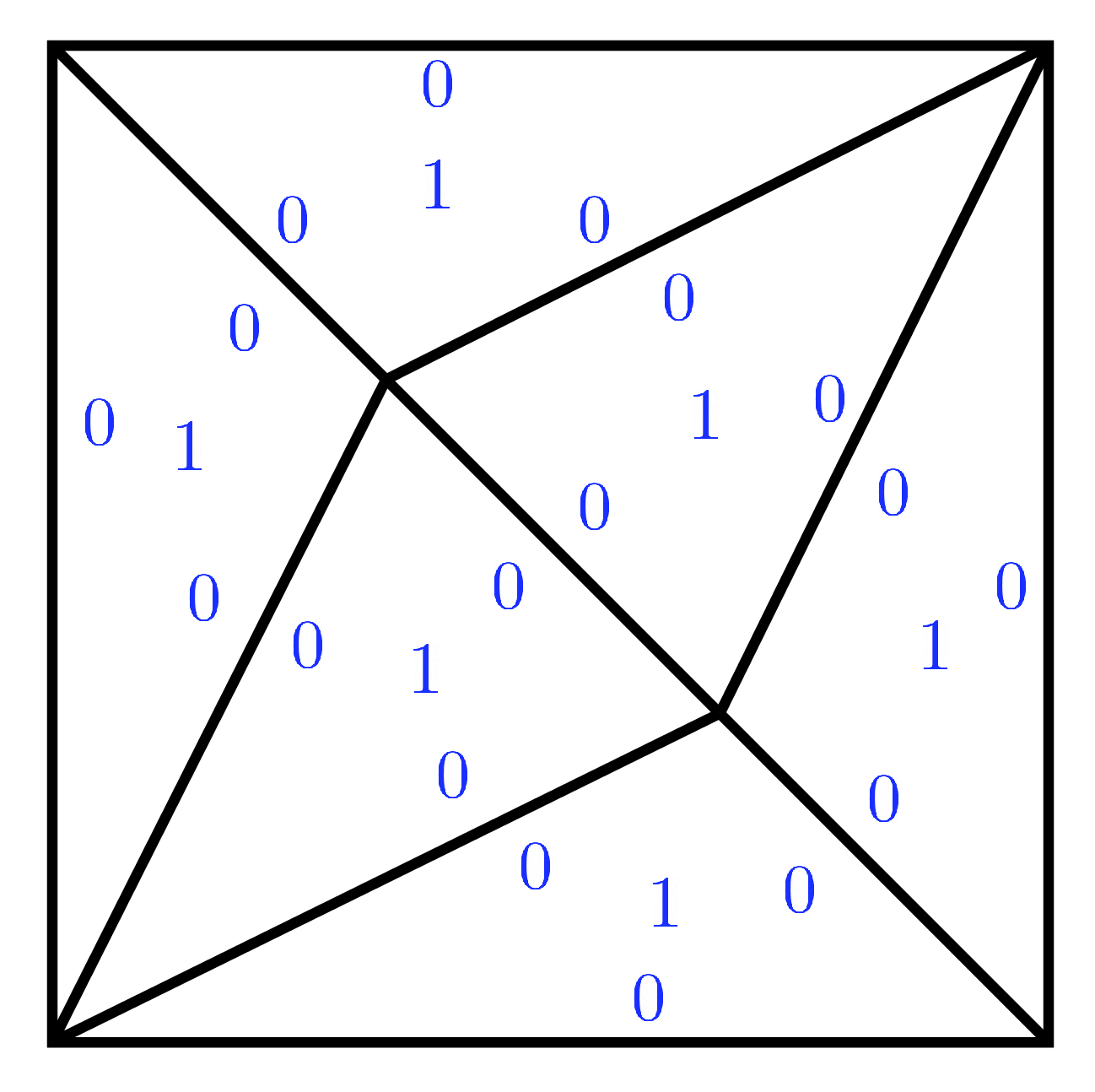}
\caption{  
}
\label{fig:can-det}
\end{figure}

\Cor{\label{cor:Stab-MP0}
Let $p$ be a vertex of $\MP_0$. Then $\Stab_{G_0}(p)$ is isomorphic to $\Aut(K_{3,3})$. Moreover, the stabilizer acts transitively on the set of remaining vertices.
}
\Proof{
By Lemma \ref{lem:g0-transitive} we know that the action of $G_0$ on the set of vertices is transitive. Therefore the stabilizers of each vertex are isomorphic.
It suffices to compute the stabilizer of the canonical vertex $q$. By definition of $q$, permutation of the contexts does not change it. That is, $\Aut(K_{3,3})\subset \Stab_{G_0}(p)$. Since there are $16$ vertices, this implies $|G_0/\Stab_{G_0}(p)|=16$ and we have  $\Aut(K_{3,3})= \Stab_{G_0}(p)$.

For the second part of the statement 
observe that the set of edges in a loop is precisely the complement
$\Omega^c$ 
of a maximal cnc set (Definition \ref{def:Omega}). Therefore there is a one-to-one correspondence between the set of loops and the set of maximal cnc sets (both types combined).
Since $\Aut(K_{3,3})$ acts transitively on the set of cnc sets (Corollary \ref{cor:action-maximal}), it also acts transitively on the set $\ell(K_{3,3})$ of loops. This implies that the action of the stabilizer, that is $\Aut(K_{3,3})$, on the vertices $\set{g_l\cdot q:\, l\in\ell(K_{3,3})}$ is transitive since $\sigma\cdot q_l = g_{\sigma\cdot l}\cdot q$ for $\sigma \in \Aut(K_{3,3})$, where $\sigma \cdot l$ is the loop obtained by the permutation action of $\sigma$.
}


\Thm{\label{thm:Graph-MP0}
The graph of $\MP_0$ is the complete graph $K_{16}$.
}
\begin{proof} 
Let us consider $q$ and another vertex $q_l=g_l\cdot q$. By Corollary \ref{cor:Stab-MP0} we can assume $l=l_{x_0}$ corresponding to flipping the outcome of $x_0$; see Lemma \ref{lem:decomposition-loop}. It suffices to show that $q$ and $q_l$ are neighbors. The distribution $p(\alpha)=\alpha q + \bar \alpha q_l$, where $0\leq \alpha\leq 1/2$, is given as follows:
$$
\includegraphics[width=.8\linewidth]{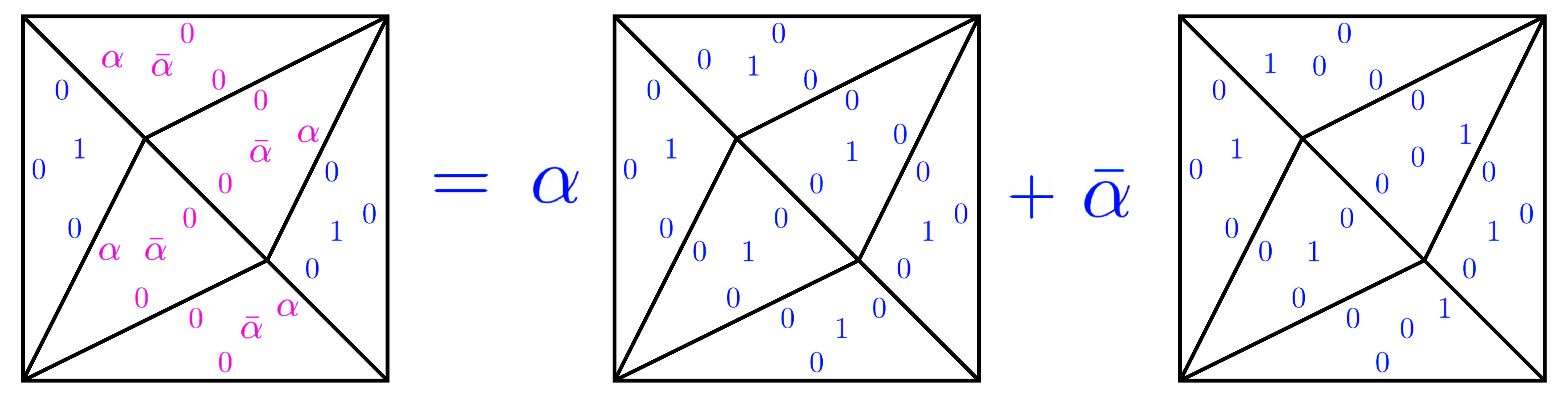}
$$
\noindent
Note that $p(\alpha)$ for $\alpha \in (0,1/2)$ specifies an edge in $\MP_0$ from $q$ to $q_l$ since the rank of $A[Z]$, where $Z$ is the set of tight inequalities, is equal to $8$. This is because, the zeros in $Z$ together with the nonsignaling conditions leaves a single parameter, that is $\alpha$. 
\end{proof}

\subsection{Graph of $\MP_1$}
\label{sec:GraphMP1}

Our goal is to describe the graph of $\MP_1$.
We will follow  a similar approach to the vertex classification. This time we consider $8$ linearly independent inequalities instead of $9$. Considering the number of deterministic edges on the torus representation is a good way to organize the cases. 
Our main technical result describes an edge between two neighboring vertices of $\MP_1$ in terms of the loops on the torus given in Fig.~(\ref{fig:mermin-torus-loops}). We begin by introducing some notation: We have seen that the complement of a loop $l\in \ell(K_{3,3})$ corresponds to a cnc set (Definition \ref{def:Omega}). Denoting a maximal cnc set that corresponds to loop $l$ by $\Omega_l$ we will write $\Omega_l^c$ for its complement, consisting of the edges that belong to the loop $l$.
A {\it signed loop} consists of a loop together with a function  
$$\varphi:\Omega_l^c \to \set{\pm 1}.$$ 
Corresponding to this function we will define a collection $p^\varphi=(p^\varphi)_{C\in \cC}$ of functions $p^\varphi:\ZZ_2^C \to \RR$ such that $\sum_{s} p^\varphi(s)=0$.
Note that this is similar to a distribution but the values sum to zero instead of one, and can be negative.
For $C=\set{x,y,z}$ our definition of $p^\varphi_C$ uses a version of Eq.~(\ref{eq:pCab}):
$$
(p^\varphi_C)^{ab} = \frac{1}{4}((-1)^a \varphi(x)+(-1)^b \varphi(y)+(-1)^{a+b+\beta(C)} \varphi(z) ).
$$


\begin{figure}[h!]
\centering
\begin{subfigure}{0.33\textwidth}
\centering
  \includegraphics[width=0.8\linewidth]{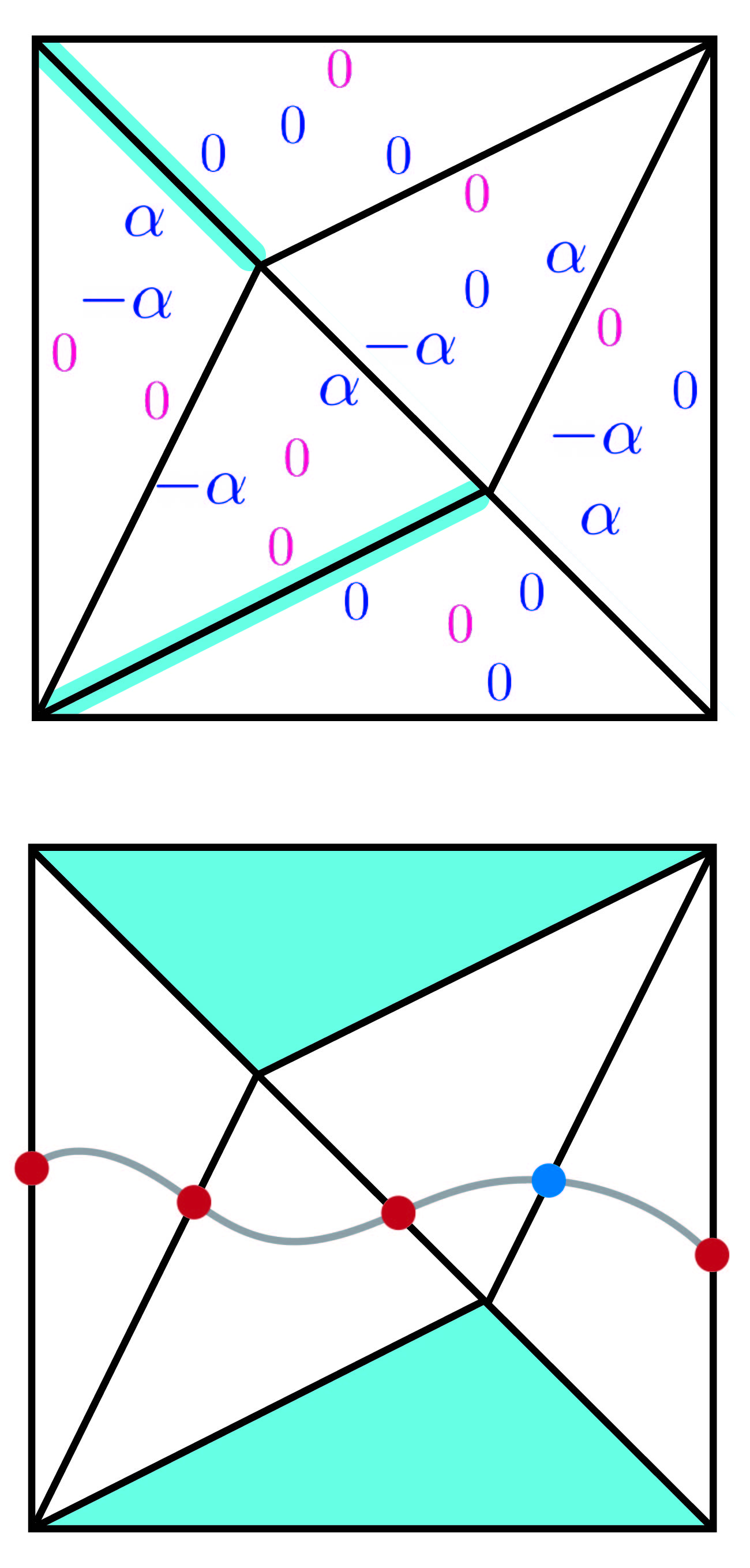}
\caption{
}
\label{fig:signed-loop-58to57}
\end{subfigure}
\begin{subfigure}{0.33\textwidth}
\centering
  \includegraphics[width=0.8\linewidth]{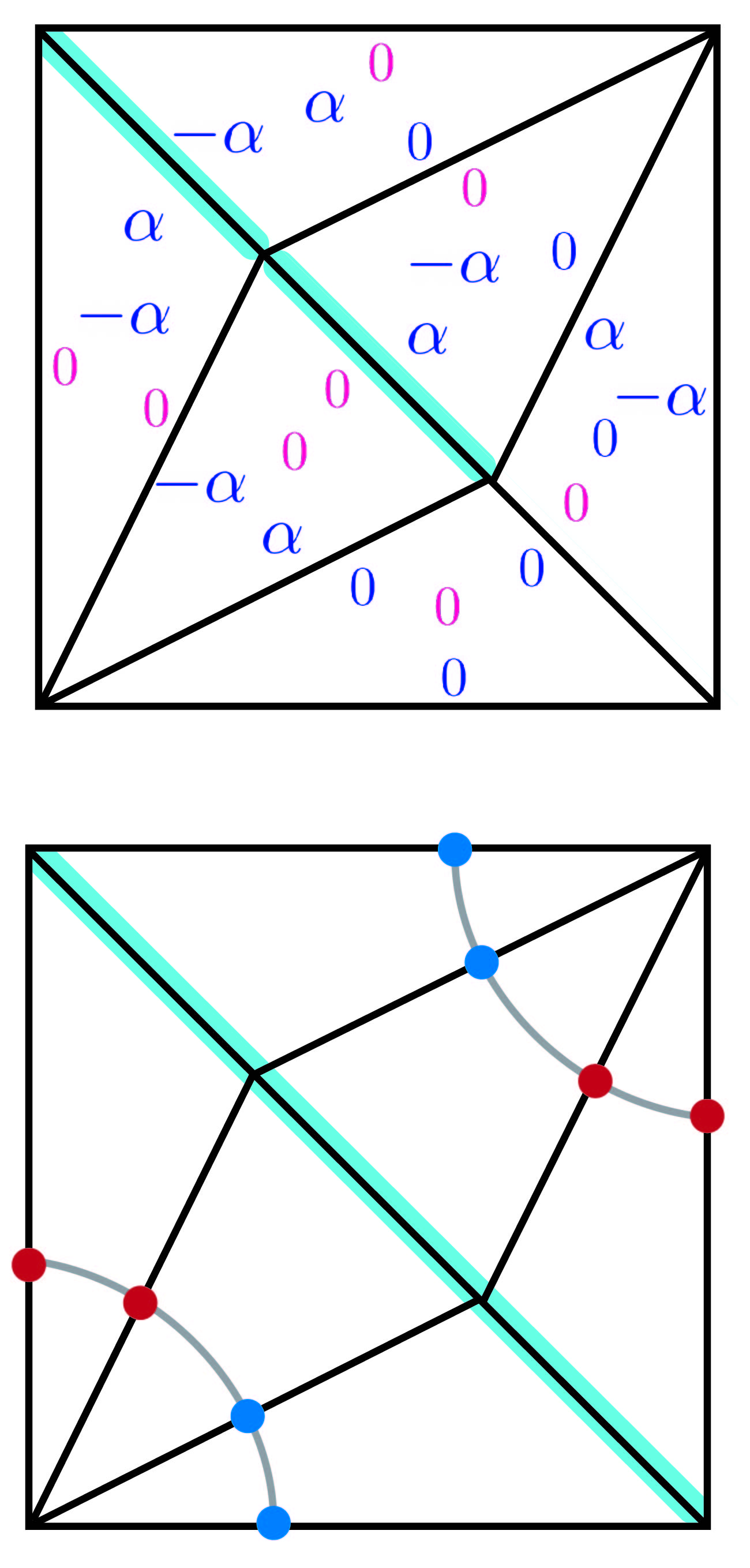}
\caption{
}
\label{fig:signed-loop-58to99}
\end{subfigure}%
\begin{subfigure}{0.33\textwidth}
\centering
  \includegraphics[width=0.8\linewidth]{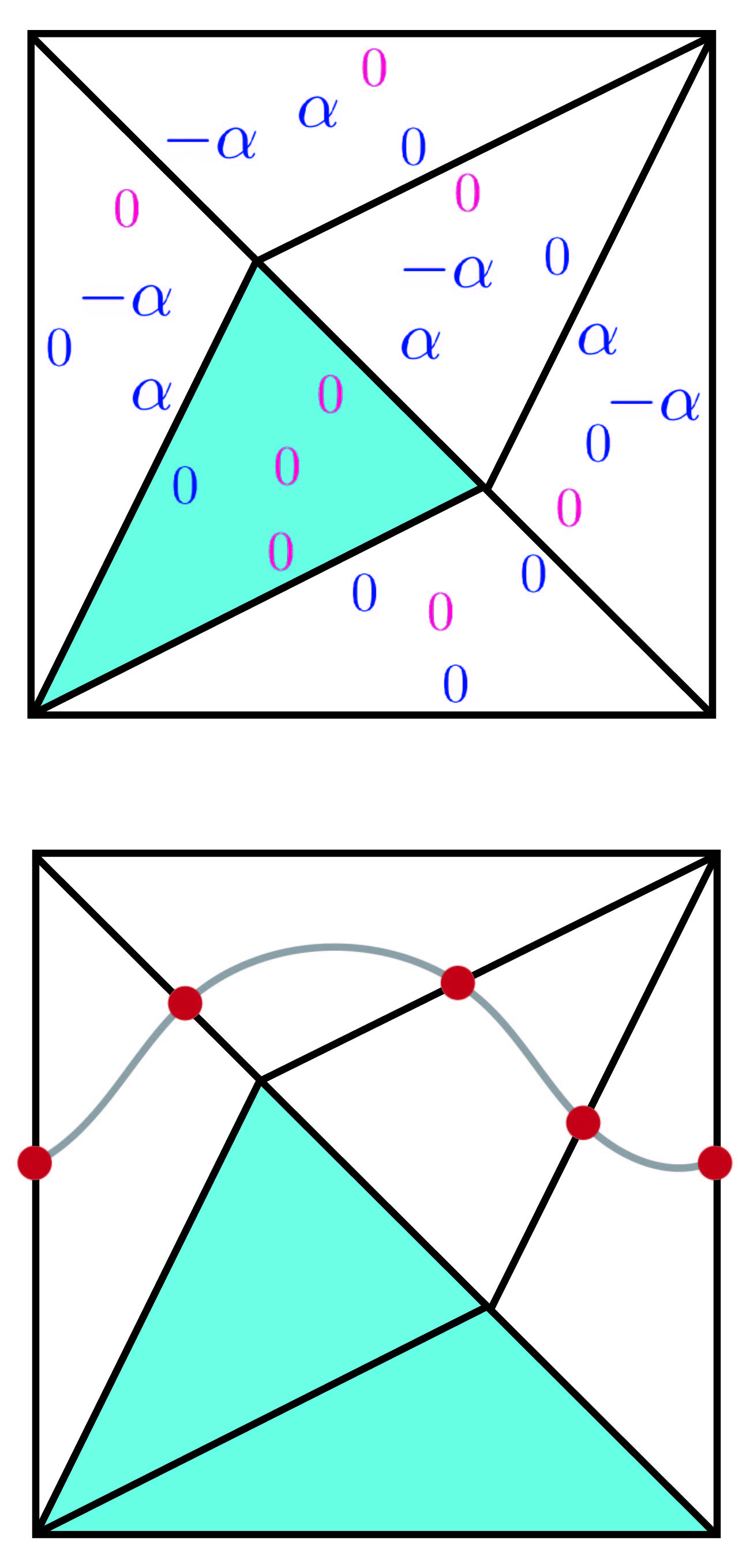}
\caption{
}
\label{fig:signed-loop-58to22}
\end{subfigure}
\caption{An edge (top) in $\MP_{1}$ between (a) canonical type 2 vertex {\VFe}  and canonical type 1 vertex $\VFS$ (b)  {\VFe}  and  {\Vnn} (c) {\VFe} and {\Vtt}, as well as the corresponding loops (bottom). The edges are parameterized by $\alpha \in [0,1/2]$ such that $\alpha = 0$ corresponds to {\VFe} and $\alpha=1/2$.
}
\label{fig:signed-loops}
\end{figure}

\Lem{\label{lem:loop-edges}
Let $p$ be a distribution in $\MP_1$ and $Z$ denote a subset of tight inequalities such that $|Z|=8$. If $\rank(A[Z])=8$ then there exists precisely two  deterministic edges. Moreover, an edge $p\in \MP_1$
between two vertices $q_1$ and $q_2$ is given by 
\begin{eqnarray}\label{eq:loop-edge}
p(\alpha) = q_1 + \alpha p^\varphi,\;\;\; \alpha \in [0,1/2],
\end{eqnarray}
for some signed loop $\varphi:\Omega_l^c \to \set{\pm 1}$, where $l\in \ell(K_{3,3})$, such that $p(1/2)=q_2$.
}
\Proof{
Let us start with the case of no deterministic edges. In this case $A[Z]\leq 6$, hence we don't have enough zeros to obtain an edge in the polytope. When there is one deterministic edge, say denoted by $z$, we consider the diamond $D$ consisting of two adjacent triangles $C,C'$ at $z$. By Lemma \ref{lem:rank-two-triangles} we have $\rank(A[Z_{C,C'}])=|Z_{C,C'}|-1$ which implies $\rank(A[Z])\leq 7$. The case of three deterministic edges, or more, with at least two of them anticommuting is studied in Section \ref{sec:mermin-vertices}. According to Lemmas \ref{lem:C1}-\ref{lem:C3} we obtain either a vertex of $\MP_1$ or a distribution that lies outside of this polytope. Remaining cases are two deterministic edges which either commute or anticommute. Note that by Lemma \ref{lem:two-edges} the commuting case also covers the three pairwise commuting deterministic edges.
To establish Eq.~(\ref{eq:loop-edge}) we note that two distributions $q_{1}$ and $q_{2}$ are connected by an edge   
if and only if they have in common 
$8$ linearly independent tight inequalities 
preserved along the edge.
Given such a set of tight inequalities, we proceed to construct $p^{\varphi}$ by placing the corresponding zeros on the torus and then use the compatibility conditions together with the fact that $\sum_{s} p^\varphi(s)=0$. 

To see how this works, let us consider a representative $\Omega = \{x,y\} \subset M$ for the case of two (a) anti-commuting and (b) commuting deterministic edges (see Fig.~(\ref{fig:58-to-neighbor})) and notice that these are both cnc sets, although not maximal. Moreover, let us choose a value assignment $s:\Omega\to\ZZ_{2}$, which by Eq.~(\ref{eq:prob-from-marginal}) determines the marginals $p_{\{x\}}$ and $p_{\{y\}}$. 
By Lemma~\ref{lem:g1-transitive-vert} the action of $G_1$ on the set of pairs $(\Omega,s)$ is transitive. Even though $\Omega$ is not maximal we can always embed it into a maximal one, extend $s$ and apply the transitivity of the action of $G_1$.

In both cases, as depicted in Fig.~(\ref{fig:Mermin-loop-58to22b}) and  (\ref{fig:Mermin-loop-58to99b}), there are $6$ linearly independent tight inequalities, thus we must choose two additional probabilities to set to zero. The possible choices are as follows: (1) Set two (or one) of the given parameters $\alpha,\beta,\gamma = 0,1$. (2) Place both remaining zeros in a single shaded triangle. (3) Place one zero in each of the shaded triangles. It is straightforward to see that both options (1) and (2) will fix the distribution to be a specific vertex, and thus will not be an edge.
 For (3) we let $p$ and $q$ denote the distribution on the shaded triangles. In Fig.~(\ref{fig:Mermin-loop-58to22b}) suppose $p$ corresponds to the triangle   whose boundary has marginals for the outcome $0$ given by $(-\alpha,-\beta,-\gamma)$ and $q$ corresponds to $(-\alpha,\beta,-\gamma)$. Then from these marginals one can compute
$$
\begin{aligned}
p^{01} &= -\alpha -p^{00} \\
p^{10} &= -\beta -p^{00} \\
p^{11} &= \gamma -p^{00}
\end{aligned}
$$
and similarly for $q$.
Using $\sum_{a,b}p^{ab}=\sum_{c,d}q^{cd}=0$ and solving for $p^{ab}$ and $q^{cd}$ we obtain
$$
\begin{aligned}
p^{00} &= (- \alpha  - \beta +  \gamma)/2\\
p^{01} &= (-\alpha + \beta -  \gamma  )/2\\
p^{10} &= ( \alpha - \beta -  \gamma )/2\\
p^{11} &= ( \alpha + \beta +  \gamma )/2
\end{aligned}\;\;\;\;\;\;\;\;\;\;
\begin{aligned}
q^{00} &= (- \alpha +  \beta +  \gamma)/2=-p^{10}\\
q^{01} &= (- \alpha -  \beta -  \gamma )/2=-p^{11}\\
q^{10} &= ( \alpha +  \beta -  \gamma )/2=-p^{00}\\
q^{11} &= ( \alpha -  \beta +  \gamma)/2=-p^{01}.
\end{aligned}
$$
If $p^{ab}$ is set to zero then we can set one of $q^{cd}$, where $(c,d)\in \ZZ_2^2-\set{a+1,b}$, equal to zero.
In this way we obtain 
a type $2$ loop. For example, setting $p^{01}=q^{10}=0$ gives the signed loop Fig.~(\ref{fig:signed-loop-58to22}) 
For this choice $p(1/2)$ is the vertex $\Vtt$. 

Fig.~(\ref{fig:Mermin-loop-58to99b}) is handled similarly. Suppose $p$ corresponds to the triangle   whose boundary has marginals 
given by $(-\alpha,-\beta,\gamma)$ and $q$ corresponds to $(\alpha,\beta,\gamma)$. Then we have
$$
\begin{aligned}
p^{00} &= (- \alpha +  \beta +  \gamma)/2\\
p^{01} &= (- \alpha -  \beta -  \gamma)/2\\
p^{10} &= ( \alpha -  \beta +  \gamma)/2\\
p^{11} &= ( \alpha +  \beta -  \gamma)/2
\end{aligned}\;\;\;\;\;\;\;\;\;\;
\begin{aligned}
q^{00} &= ( \alpha +  \beta +  \gamma)/2=-p^{01}\\
q^{01} &= ( -\alpha  + \beta -  \gamma)/2=-p^{10}\\
q^{10} &= (- \alpha -  \beta +  \gamma  )/2=-p^{11}\\
q^{11} &= ( \alpha -  \beta -  \gamma )/2=-p^{00}.
\end{aligned}
$$
This case gives either a type $1$ or a type $2$ loop. For example, setting 
$p^{10}=q^{00}=0$
gives the signed loop in Fig.~(\ref{fig:signed-loop-58to99}), and $p(1/2)$ is the vertex $\Vnn$. 
\begin{figure}[h!]
\centering
\begin{subfigure}{\textwidth}
\centering
  \includegraphics[width = 0.7\linewidth]{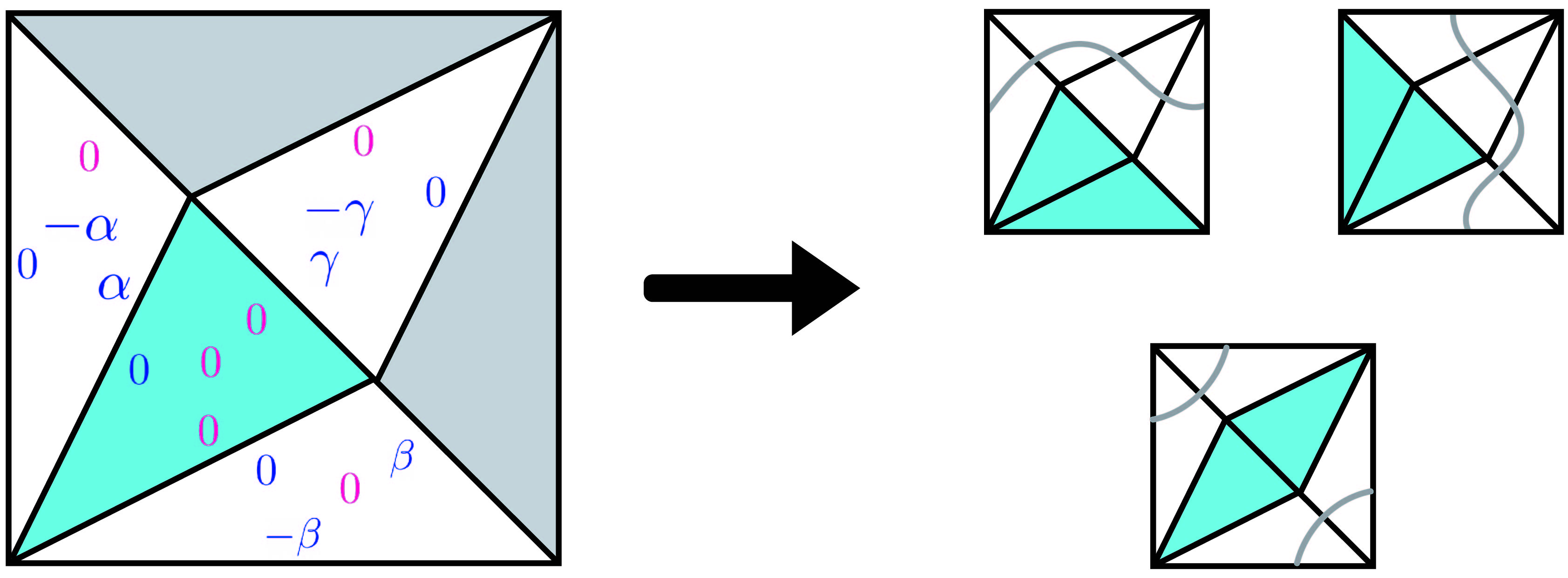}
\caption{
}
\label{fig:Mermin-loop-58to22b}
\end{subfigure}\\
\begin{subfigure}{\textwidth}
\centering
  \includegraphics[width = 0.7\linewidth]{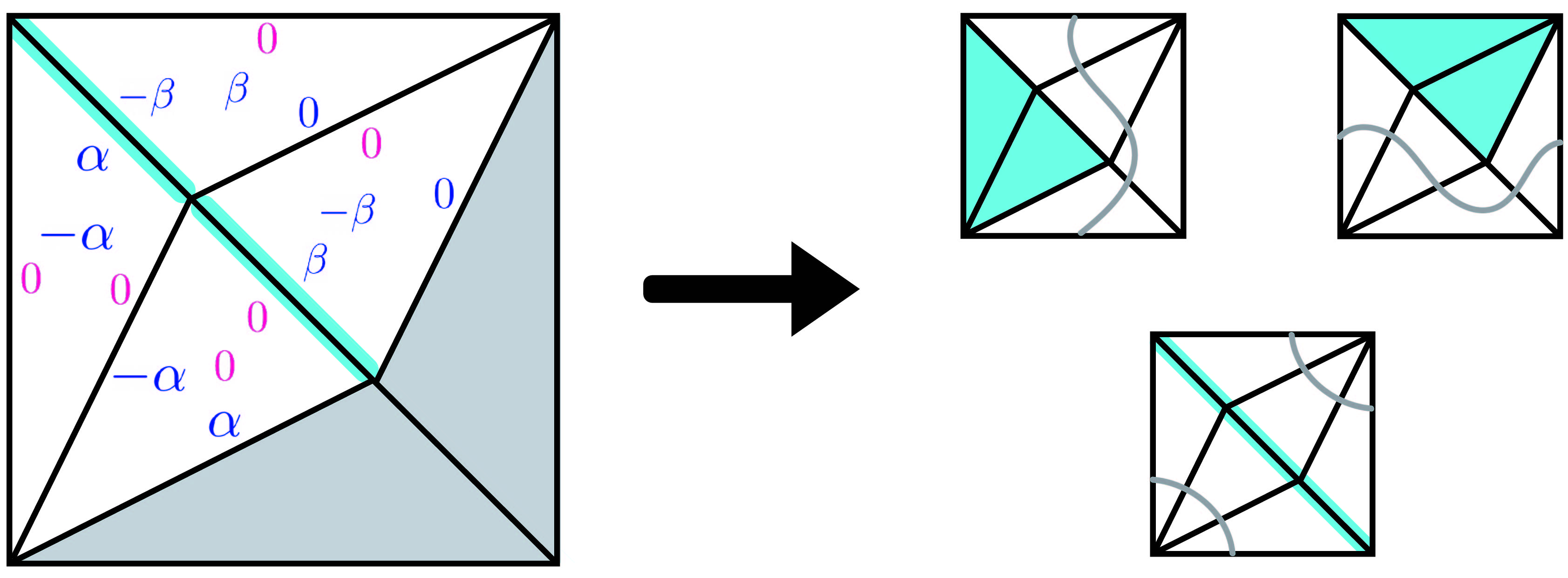}
\caption{
}
\label{fig:Mermin-loop-58to99b}
\end{subfigure}
\caption{
(a) Two commuting edges generate three possible loops; all of these are type-2.  The zeros in pink correspond to linearly independent tight inequalities. (b) Two anti-commuting edges generate three possible loops; i.e., one type-1 and two type-2.
}
\label{fig:58-to-neighbor}
\end{figure}

}


Next, we will describe the graph of $\MP_1$.
By Lemma \ref{lem:stabilizer-MP1} the action of $G_1$ on the type $1$ and $2$ vertices is transitive. Therefore to understand the local structure, i.e., the neighbors, at a given vertex we can fix one type $1$ vertex and one type $2$ vertex.
Our canonical representative for a  type $1$ vertex is $\VFS$ given in Fig.~(\ref{fig:V57-dist}), which as an operator given as follows:
\begin{equation}\label{eq:V57}
q_0=\frac{1}{4}(\mathbb{1}+X\otimes Y-Y\otimes Y+Z\otimes Y).
\end{equation}
Here we are using Lemma \ref{lem:M1-quantum} to identify points of $\MP_1$ as operators and don't distinguish them notationally from the probability distributions.
For a  type $2$ vertex our canonical choice is $\VFe$ given in Fig.~(\ref{fig:V58-dist}):
\begin{equation}\label{eq:V58}
\VFe =\frac{1}{4}(\mathbb{1}+X\otimes X+X\otimes Y+Y\otimes X-Y\otimes Y+Z\otimes Z) .
\end{equation}
 
%

 
\Lem{\label{lem:stabilizer-MP1}
Let $p$ be a vertex of $\MP_1$. 
\begin{itemize}
\item If $p$ is of type $1$ then its stabilizer is  isomorphic to the dihedral  group $D_{24}$ of order $24$. For the canonical type $1$ vertex $q_0$ we have
$$
\Stab_{G_1}(q_0) = \Span{YS\otimes X, YH\otimes H}, 
$$  
where $S$ is the phase gate and $H$ is the Hadamard gate.

\item If $p$ is of type $2$ then its stabilizer is 
isomorphic to the dihedral  group $D_{16}$ of order $16$. For the canonical type $2$ vertex $p_0$ we have
$$
\Stab_{G_1}(p_0) = \Span{X\otimes YS, \SWAP},
$$ 
where $\SWAP$ is the swap gate that permutes the parties.
\end{itemize} 
In particular, $G_1$ acts transitively on the set of type $1$ and $2$ vertices.
}
\Proof{
Proof is given in Lemma \ref{lem:stab57} and Lemma \ref{lem:stab58}. The last statement about the transitivity of the action follows from Lemma \ref{lem:g1-transitive-vert}.
}
 
%


\begin{figure}[h!]
\centering
\begin{subfigure}{0.45\textwidth}
\centering
\includegraphics[width=0.6\linewidth]{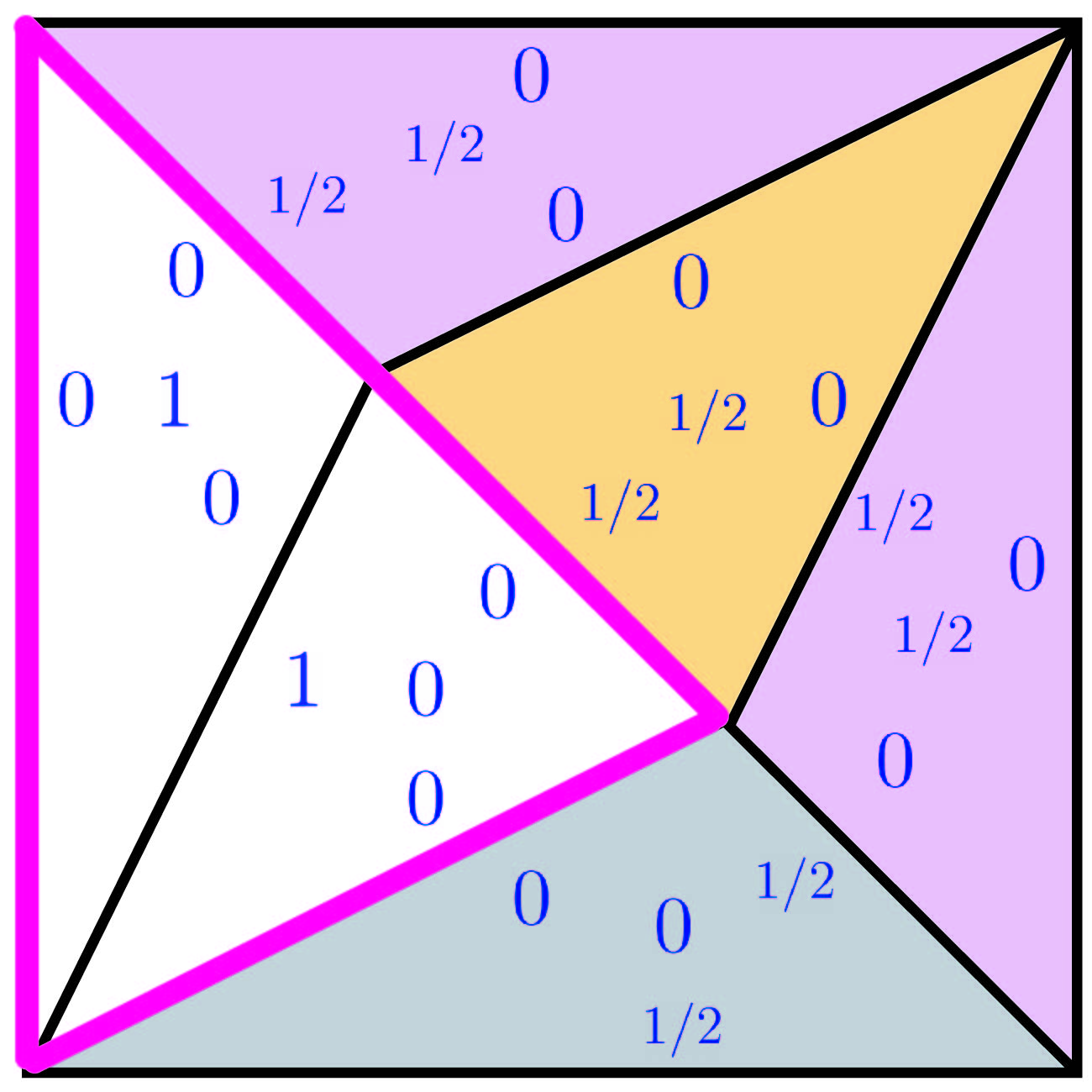}
\caption{ 
}
\label{fig:V58-dist}
\end{subfigure}
\begin{subfigure}{0.45\textwidth}
\centering
\includegraphics[width=0.6\linewidth]{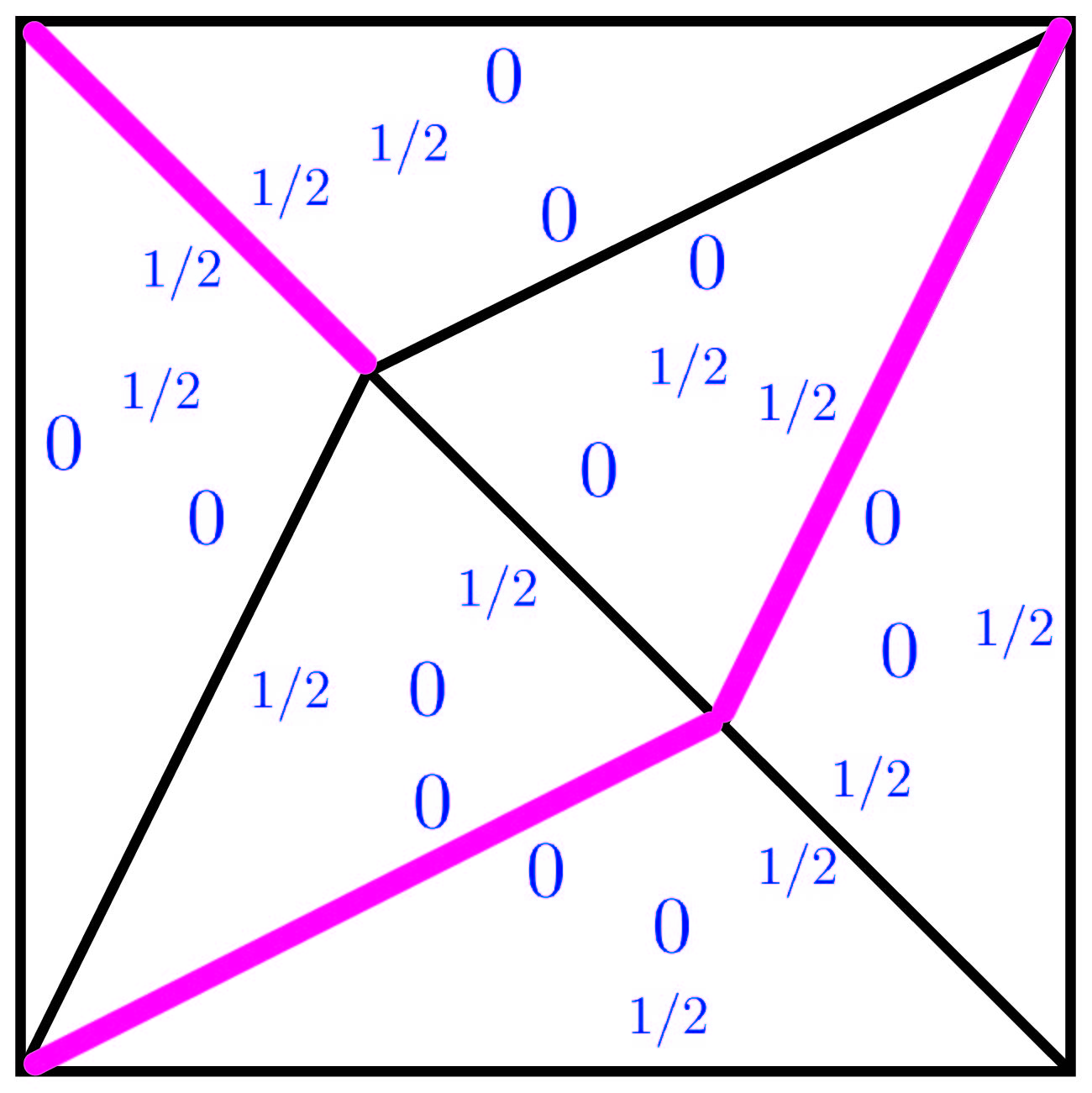}
\caption{ 
}
\label{fig:V57-dist}
\end{subfigure}
\caption{ 
}
\label{fig:canonical-dist}
\end{figure}

 For a vertex $p\in \MP_1$ let $N(p)$ denote the set of neighbor vertices of $p$. 
 
\Thm{\label{thm-MP1-graph}
The graph of $\MP_1$ consists of $120$ vertices partitioned into two kinds: $48$ type $1$ and $72$ type $2$ vertices. 
The local structure at these vertices is as follows:
\begin{itemize}
\item $N(\VFS)$ consists of $12$ type $2$ vertices given in Fig.~(\ref{fig:T2NofT1}). $\Stab_{G_1}(\VFS)$ acts transitively on these neighbors.
\item $N(\VFe)$ consists of $8$ type $1$ vertices and $16$ type $2$ vertices given in Fig.~(\ref{fig:T1NofT2}) and (\ref{fig:T2NofT2}); respectively. $\Stab_{G_1}(\VFe)$ acts transitively on the type $1$ neighbors, whereas the type $2$ neighbors break into  two orbits.
\end{itemize}  
} 
\Proof{Vertices of $\MP_1$ are classified in part (ii) of Theorem \ref{thm:VertexClassification}. 
Lemma \ref{lem:loop-edges} shows that edges of $\MP_1$ are described by signed loops.
To describe the local structure of the graph at a vertex we consider the canonical vertices in Eq.~(\ref{eq:V57}) and (\ref{eq:V58}), since by Lemma \ref{lem:stabilizer-MP1} $G_1$ acts transitively on each type of vertex. 

\begin{figure}[h!]
\centering
\includegraphics[width=.9\linewidth]{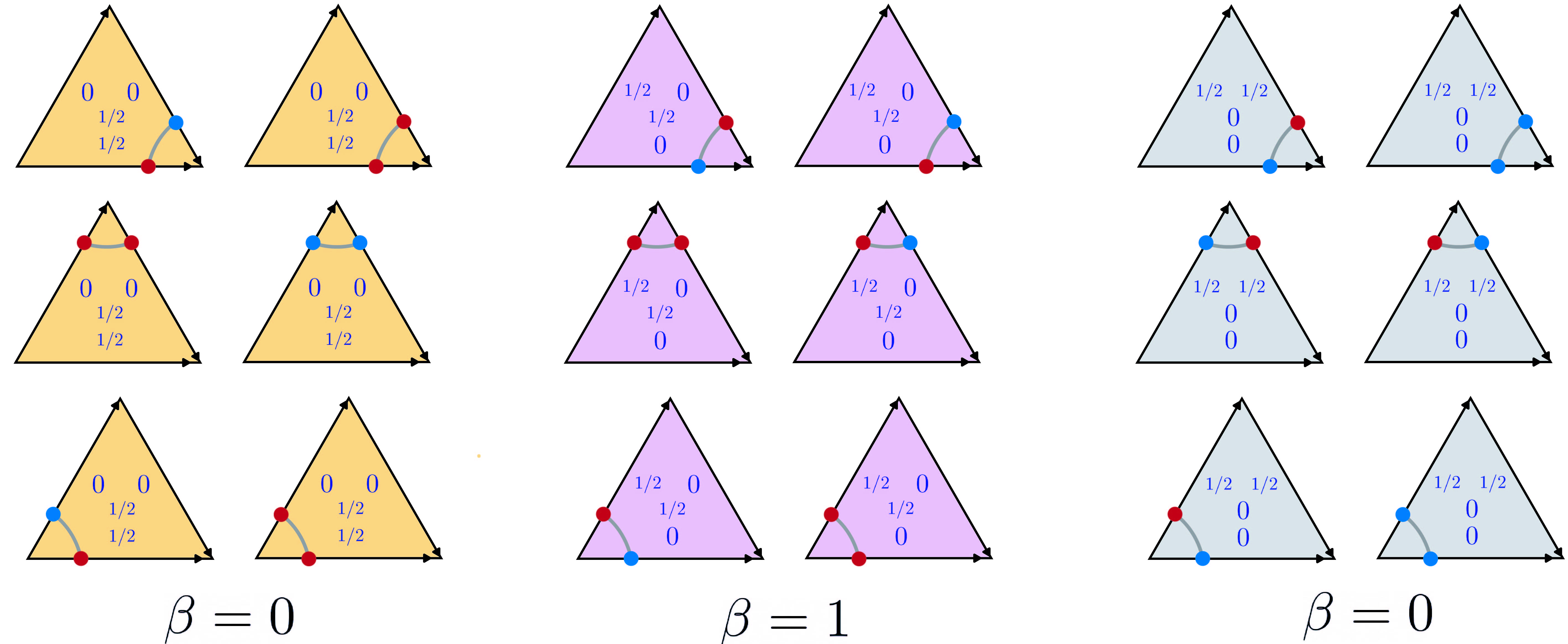}
\caption{
}
\label{fig:triangle-loop}
\end{figure}

Our strategy is to find the signed loops such that $p(\alpha)$ in Eq.~(\ref{eq:loop-edge}) gives $p(0)=\VFS$ or $\VFe$. Let us start with $\VFe$, the corresponding distribution is given in Fig(\ref{fig:V58-dist}). 
Let $\Omega$ denote the maximal cnc set corresponding to $\VFe$ and $\Omega^c$ denote its complement. 
We will partition a loop $l$ into two parts $\Omega_l \cap \Omega$ and $\Omega_l \cap \Omega^c$. 
The restriction of $\varphi$ to $\Omega_l \cap \Omega$ is determined by the outcome assignment corresponding to $\VFe$. We begin by considering the restriction of $\varphi$ to $\Omega_l \cap \Omega^c$. 
The region $\Omega^c$ consists of four triangles. Each of these triangles has exactly one deterministic edge. Let $C$ be one of those four triangles.  The intersection $\Omega_l\cap C$ is either empty or consists of two edges. 
There are two choices for the restriction of the sign $\varphi$ to this intersection, which is dictated by the distribution on $C$. All the possibilities are given in Fig.~(\ref{fig:triangle-loop}).
Observe that $\Omega_l \cap \Omega$ can be given by one of the following possibilities:
$$
\includegraphics[width=0.45\linewidth]{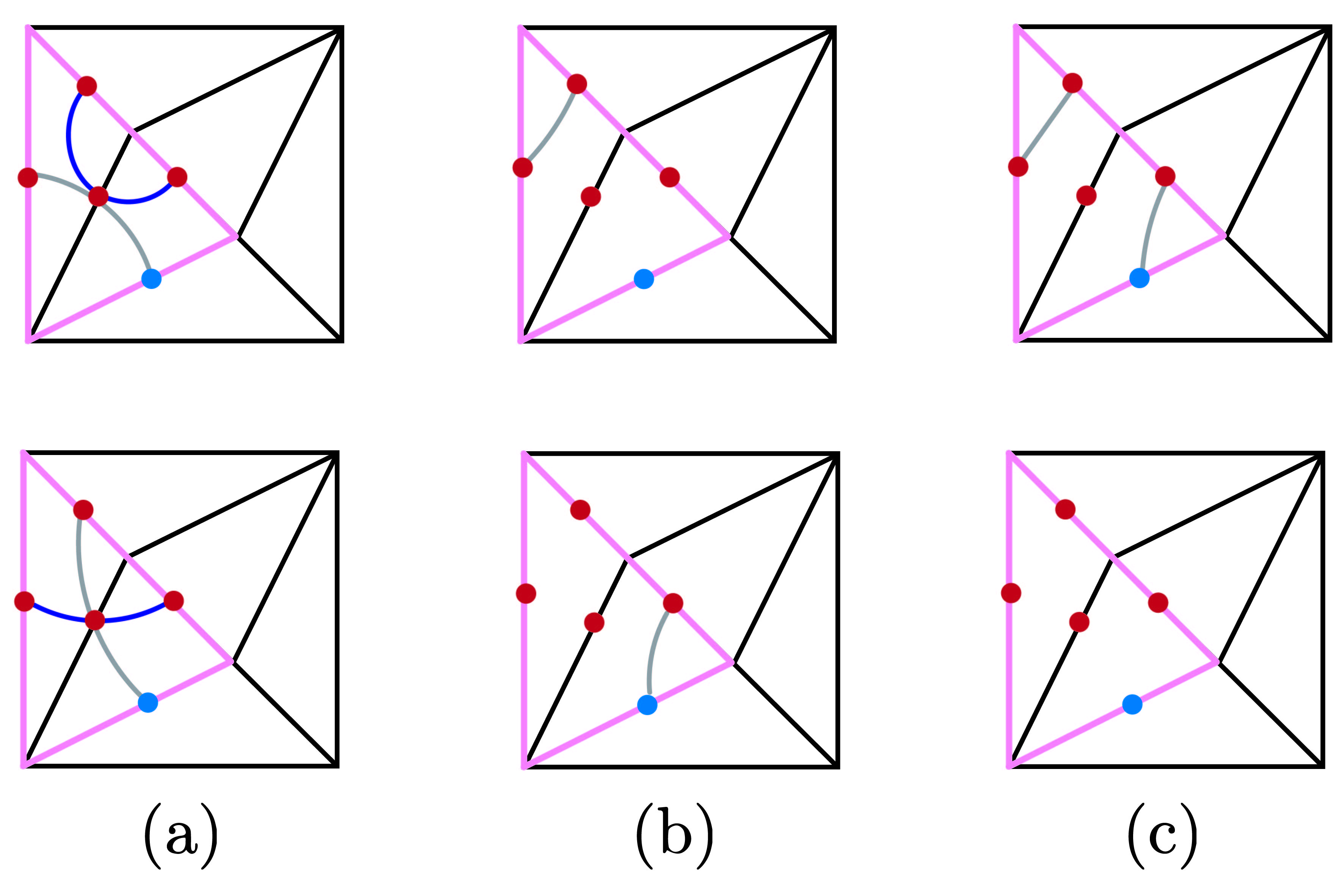}
$$
We analyze each case.
\begin{enumerate}
\item[(a)] There are two ways to complete the paths to a loop. The sign on $\Omega_l\cap \Omega^c$ is determined by two adjacent triangles $C,C'$. There are two possibilities for the sign on $(\Omega_l\cap \Omega^c)\cap (C\cup C')$. If $l$ is type $2$ then
we obtain the type $2$ neighbors given in the first two columns of  Fig.~(\ref{fig:T2NofT2}). The action of $\Stab_{G_1}(\VFe)$ is transitive by Lemma \ref{lem:G1-transitive-V58} on these neighbors; see Table (\ref{tab:V58-type2-neighbors}). A representative vertex in this orbit is
\begin{equation}\label{eq:V99}
\Vnn=\frac{1}{4}(\one + X\otimes X + X\otimes Y - Y\otimes Z + Z\otimes X - Z\otimes Y).
\end{equation} 
If $l$ is type $1$ we obtain  the type $1$ neighbors in   Fig.~(\ref{fig:T1NofT2}). By Lemma \ref{lem:G1-transitive-V58} $\Stab_{G_1}(\VFe)$ acts transitively; see Table (\ref{tab:V58-type1-neighbors}). A representative vertex in this orbit is $\VFS$ given in Eq.~(\ref{eq:V57}).

\item[(b)] This is similar to (a): Two ways to complete to a loop and two choices for the sign on the complement. We obtain the type $2$ neighbors in the last two columns of Fig.~(\ref{fig:T2NofT2}). The action of $\Stab_{G_1}(\VFe)$ is transitive by Lemma \ref{lem:G1-transitive-V58} on these neighbors; see Table (\ref{tab:V58-type2-neighbors}). A representative vertex in this orbit is
\begin{equation}\label{eq:V22}
 \Vtt =\frac{1}{4}(\one + X\otimes X - Y\otimes Y - Y\otimes Z - Z\otimes Y + Z\otimes Z).
\end{equation} 

\item[(c)] Top figure: There are two ways to complete to a loop. The sign on the complement is determined by two nonadjacent triangles. Hence there are four possibilities for the sign on the complement. We obtain the signed loops in Fig.~(\ref{fig:T1NoNofT2}).
By Lemma \ref{lem:not-neigh} the vertices at $p(1/2)$ are not neighbors of $\VFe$. Also in the proof of this lemma we see that $\Stab_{G_1}(\VFe)$ acts transitively; see Table (\ref{tab:V58-type1-non-neighbors}). Our representative vertex is
\begin{equation}\label{eq:V28}
\Vte=\frac{1}{4}(\one + Z\otimes X - Z\otimes Y + Z\otimes Z).
\end{equation}
Bottom figure: There is a unique loop on $\Omega^c$. However, no sign is compatible with the restrictions onto the triangles given in Fig.~(\ref{fig:triangle-loop}). This loop does not produce an edge in the graph that initiates from $\VFe$; see Lemma \ref{lem:not-neigh}.
\end{enumerate}
The distributions connecting $\VFe$ to $\VFS$, $\Vnn$ and $\Vtt$ are given in Fig.~(\ref{fig:signed-loop-58to57}), (\ref{fig:signed-loop-58to99}) and (\ref{fig:signed-loop-58to22}); respectively. 

For $\VFS$ given in Fig~(\ref{fig:V57-dist}) the argument is similar. Let $p(\alpha)$ be a path obtained from a signed loop such that $p(0)=\VFS$. The distribution $p(1/2)$ will consists of triangles with a single deterministic edge on the boundary. Hence it is a vertex of type $1$. However, we need to determine whether $p(\alpha)$ is an edge in $\MP_1$.
There are three cases to consider.
\begin{enumerate}
\item[(a)] $l$ is of type $2$: Then $p(\alpha)$ will be obtained from $\VFS$ by swapping a $1/2$ with $0$ in each triangle. This means that the common set of zeros between $\VFS$ and $p(1/2)$ is $6$. Therefore $p(\alpha)$ cannot be an edge.
\item[(b)] $l$ is of type $1$ and intersects two of the edges in $\set{X\otimes Y, Z\otimes Y, Y\otimes Y}$: Similarly $p(1/2)$ is a type $1$ vertex. Looking at the common zeros we see that there are $8$. However, as in the proof of Lemma \ref{lem:not-neigh} we can argue that Lemma \ref{lem:rank-two-triangles} to a pair of adjacent triangles to reduce the rank by $1$. This implies that the path $p(\alpha)$ is not an edge in $\MP_1$. 
\item[(b)] $l$ is of type $1$ and intersects one of the edges in $\set{X\otimes Y, Z\otimes Y, Y\otimes Y}$: Then $p(1/2)$ is a type $2$ vertex as listed in Fig.~(\ref{fig:T2NofT1}).  By Lemma \ref{lem:G1-transitive-V57} $\Stab_{G_1}(\VFS)$ acts transitively on them. The distribution connecting $\VFS$ to $\VFe$ is given in Fig.~(\ref{fig:signed-loop-58to57}). 
\end{enumerate}

}

\begin{figure}[h!]
\centering
\begin{subfigure}{.33\textwidth}
  \centering
  \includegraphics[scale = 0.055]{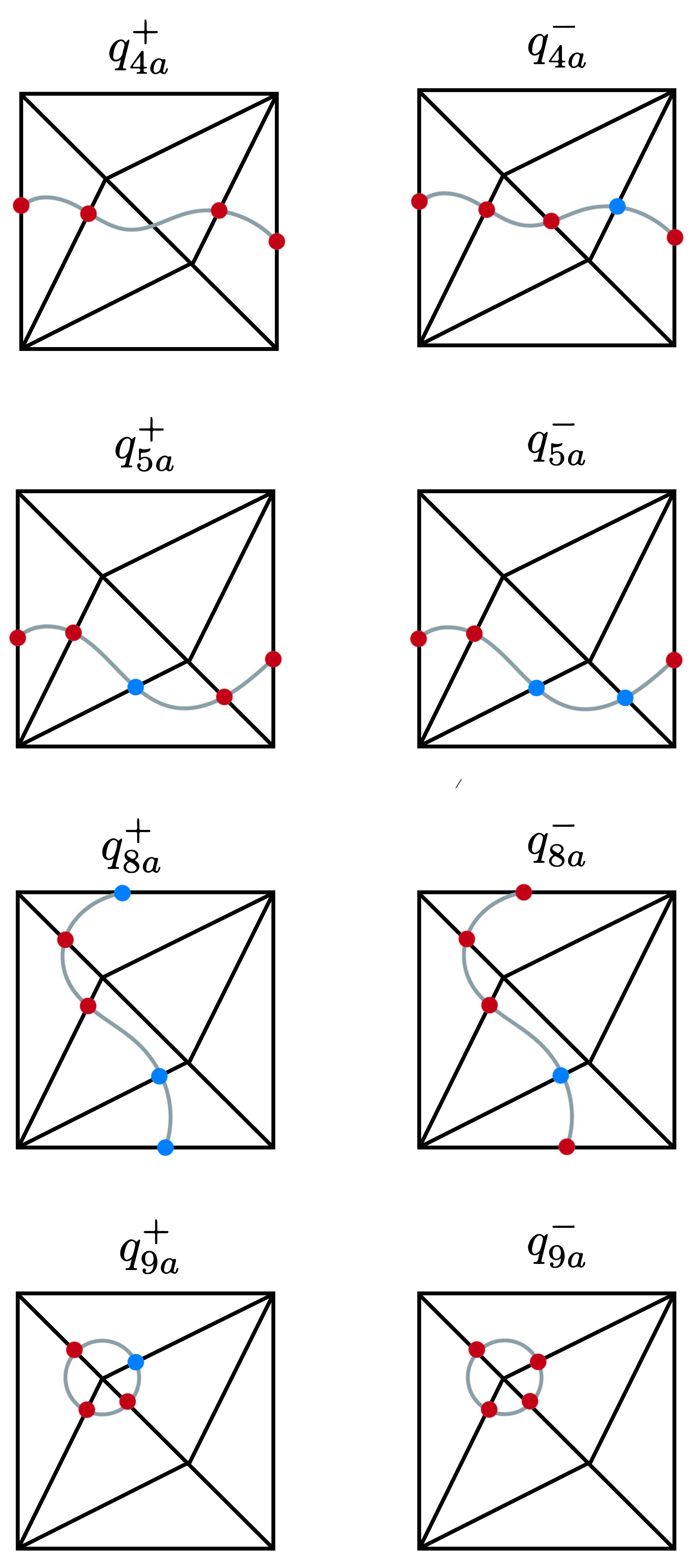}
  \caption{}
  \label{fig:T1NofT2}
\end{subfigure}%
\vspace{0.5cm}
\begin{subfigure}{.66\textwidth}
  \centering
  \includegraphics[scale = 0.055]{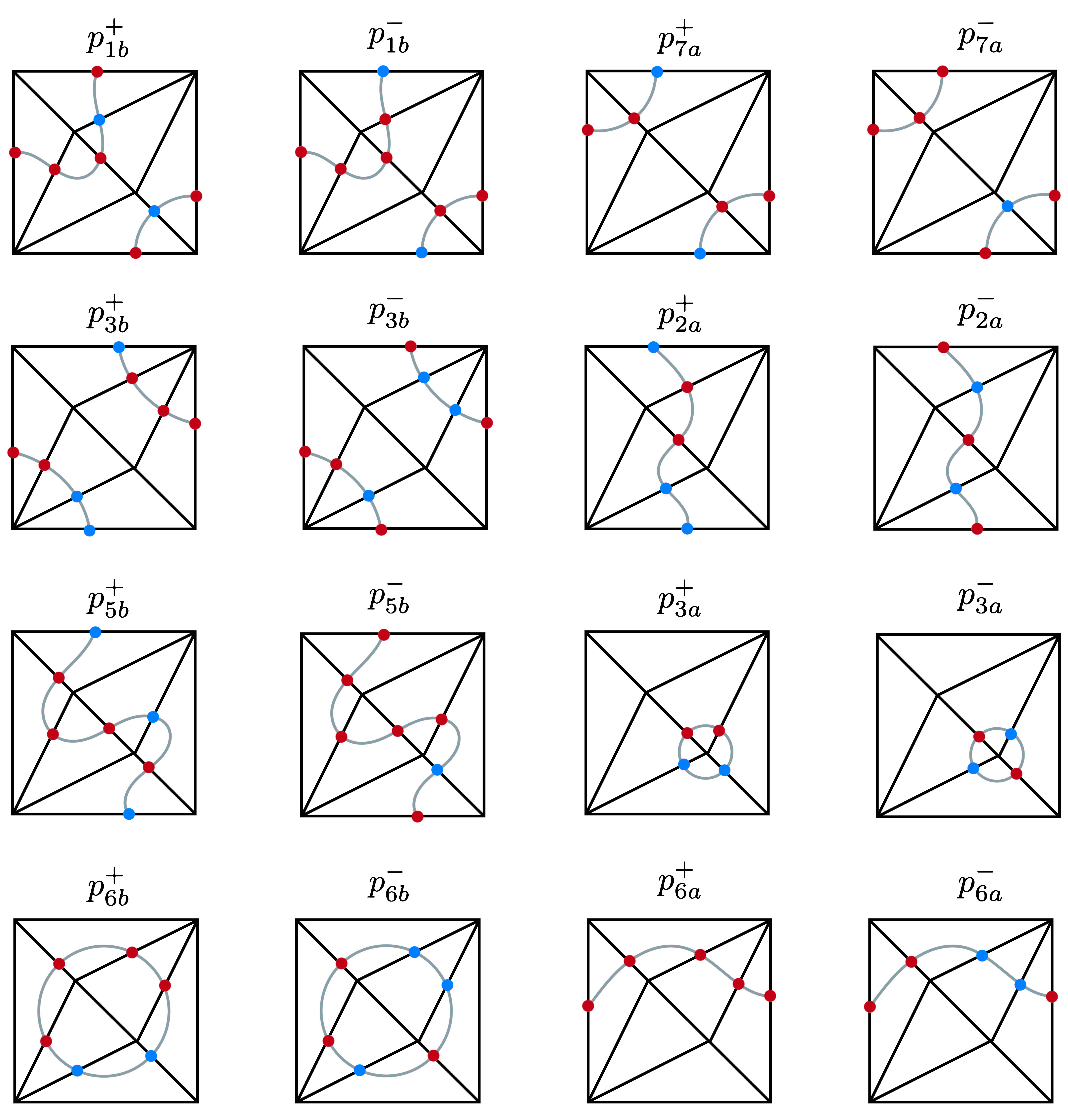}
  \caption{}
  \label{fig:T2NofT2}
\end{subfigure}
\caption{ (a) Type $1$ neighbors of $\VFe$.  (b) Type $2$ neighbors of $\VFe$. 
}
\label{fig:NofT2}
\end{figure}

\begin{figure}[h!]
\centering
\begin{subfigure}{.66\textwidth}
  \centering
  \includegraphics[scale = 0.055]{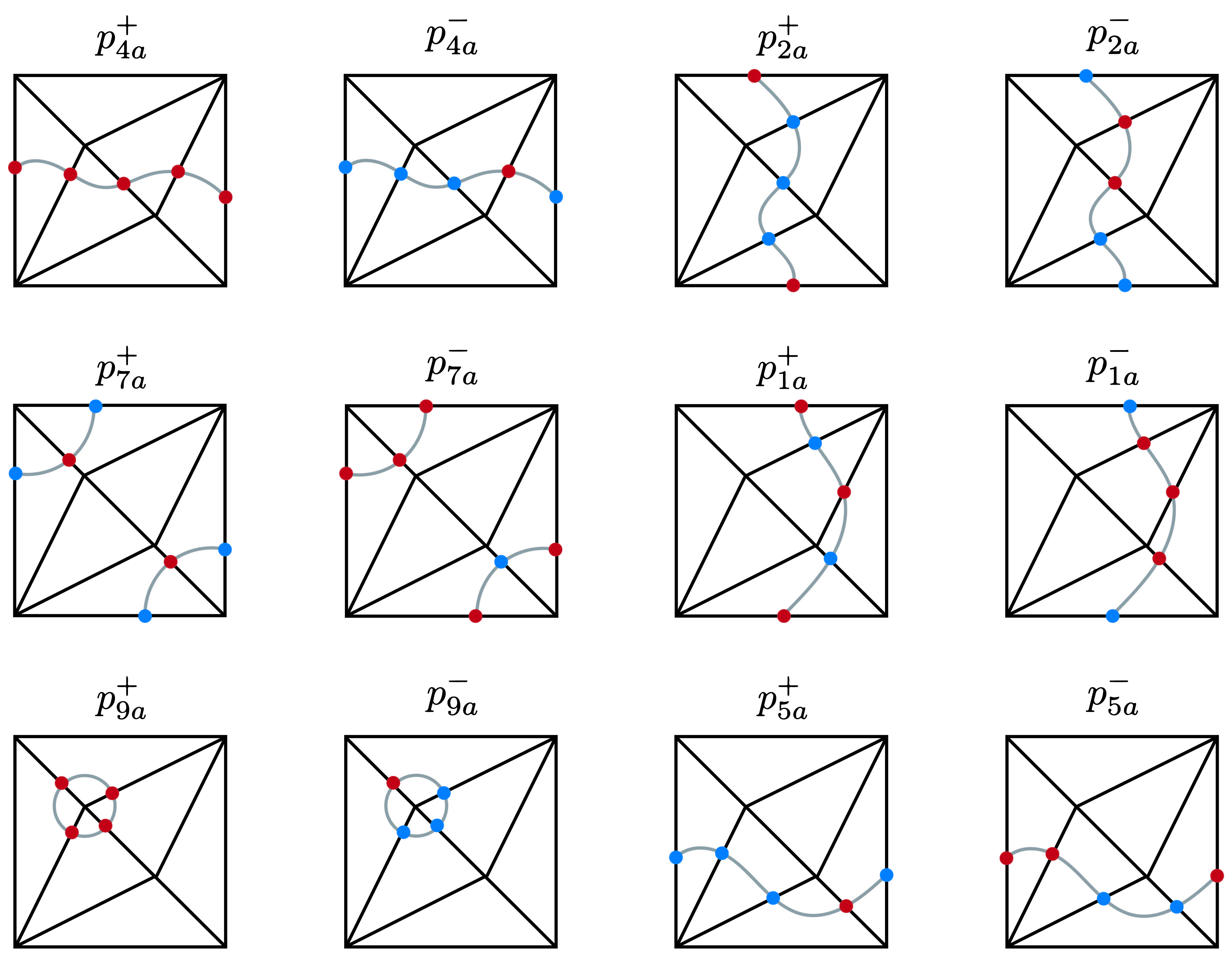}
  \caption{}
  \label{fig:T2NofT1}
\end{subfigure} 
\begin{subfigure}{.33\textwidth}
  \centering
  \includegraphics[scale = 0.055]{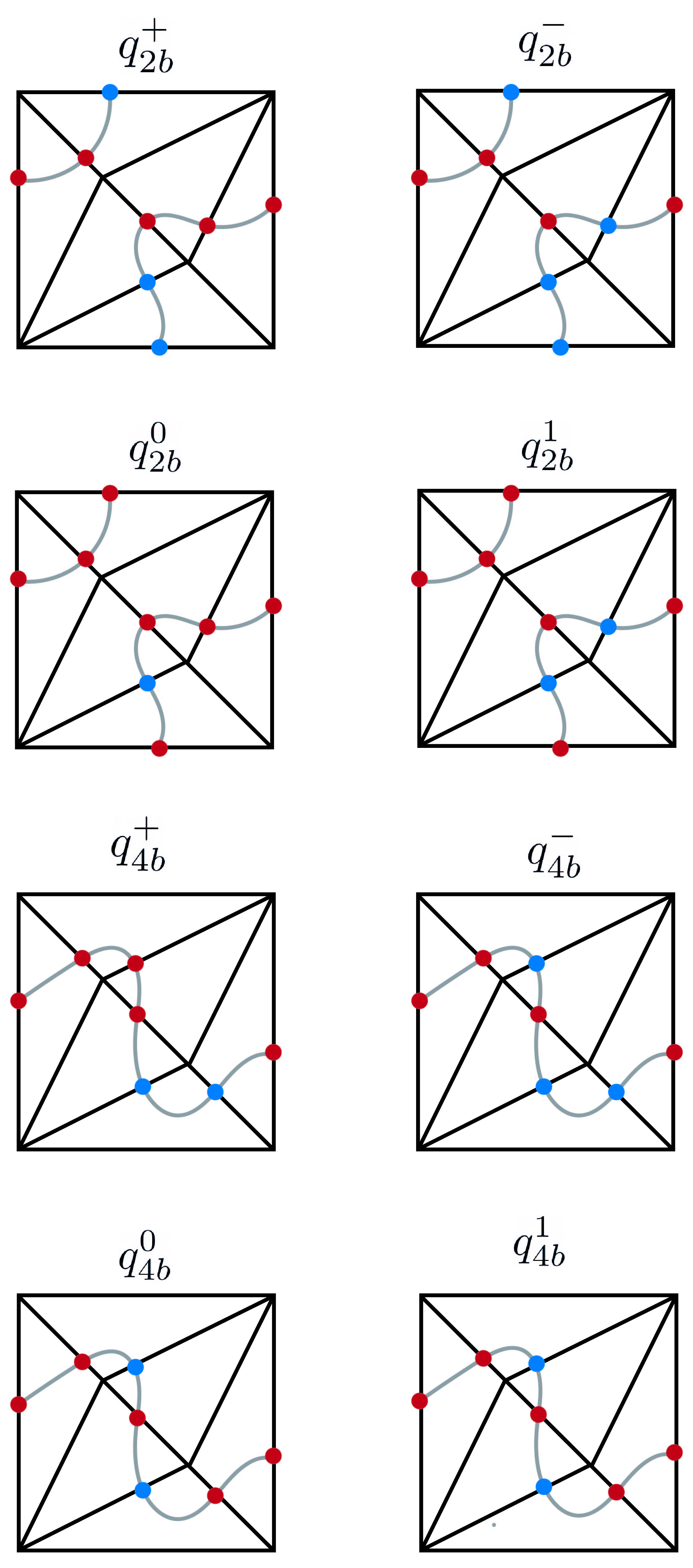}
  \caption{}
  \label{fig:T1NoNofT2}
\end{subfigure}
\caption{ (a) Type $2$ neighbors of $\VFS$. 
The neighbors $\sVfT$, $\sVTz$, $\sVee$, $\sVfn$, $\sVFo$, $\sVfe$ coincide with the type $2$ neighbors $\VfT$, $\VTz$, $\Vee$, $\Vfn$, $\VFo$, $\Vfe$ of the canonical type $2$ vertex $\VFe$; respectively.
 (b) Signed loops that connect $\VFe$ to a type $1$ vertex that is not a neighbor.
}
\label{fig:NofT1-NoNofT2}
\end{figure}







\section{Applications}\label{sec:applications}
Mermin polytopes $\MP_{\beta}$, besides having an interesting structure in their own right, also have utility in understanding aspects of quantum foundations ($\MP_{0}$) as well as quantum computation ($\MP_{1}$). We explore these topics here.

\subsection{A new topological proof of Fine's theorem}\label{sec:fine-thm}
Here we combine the current results on Mermin polytopes together with the topological framework of \cite{okay2022simplicial} to provide a novel proof of Fine's theorem \cite{fine1982hidden}. Before proceeding to the precise statement of Fine's theorem, however, we recall from Section~\ref{sec:mp0-case} that the CHSH scenario consists of four measurements $x_{i}$, $y_{j}$ and four measurement contexts consisting of pairs $\{x_{i},y_{j}\}$, where $i,j\in\ZZ_{2}$. Our first goal will be to represent this scenario topologically.

In the simplicial approach to contextuality, first introduced in \cite{okay2022simplicial}, and discussed briefly in Section~\ref{sec:top-rep}, measurement contexts are represented by simplicies (triangles) and to each simplex we associate a probability distribution. The collection of distributions on each simplex constitutes a \emph{simplicial distribution}, which generalizes the notion of nonsignaling distributions. In particular, a well-studied class of measurement scenarios are the bipartite scenarios which in the simplicial framework are given by collections of triangles (i.e., $2$-simplicies) where edges (i.e., $1$-simplicies) represent measurements; not necessarily local. Nonsignaling (or compatibility) constraints are then formalized as the gluing of triangles along edges. For instance, the $(2,1,2)$ Bell scenario is just a single triangle as in Fig.~(\ref{fig:simple-scenarios}), while the so-called diamond scenario consists of two triangles glued along a single edge; see Fig.~(\ref{fig:diamond}).
The diamond scenario will prove useful for our proof of Fine's theorem.  
\begin{figure}[h!]
\centering
\includegraphics[width=.15\linewidth]{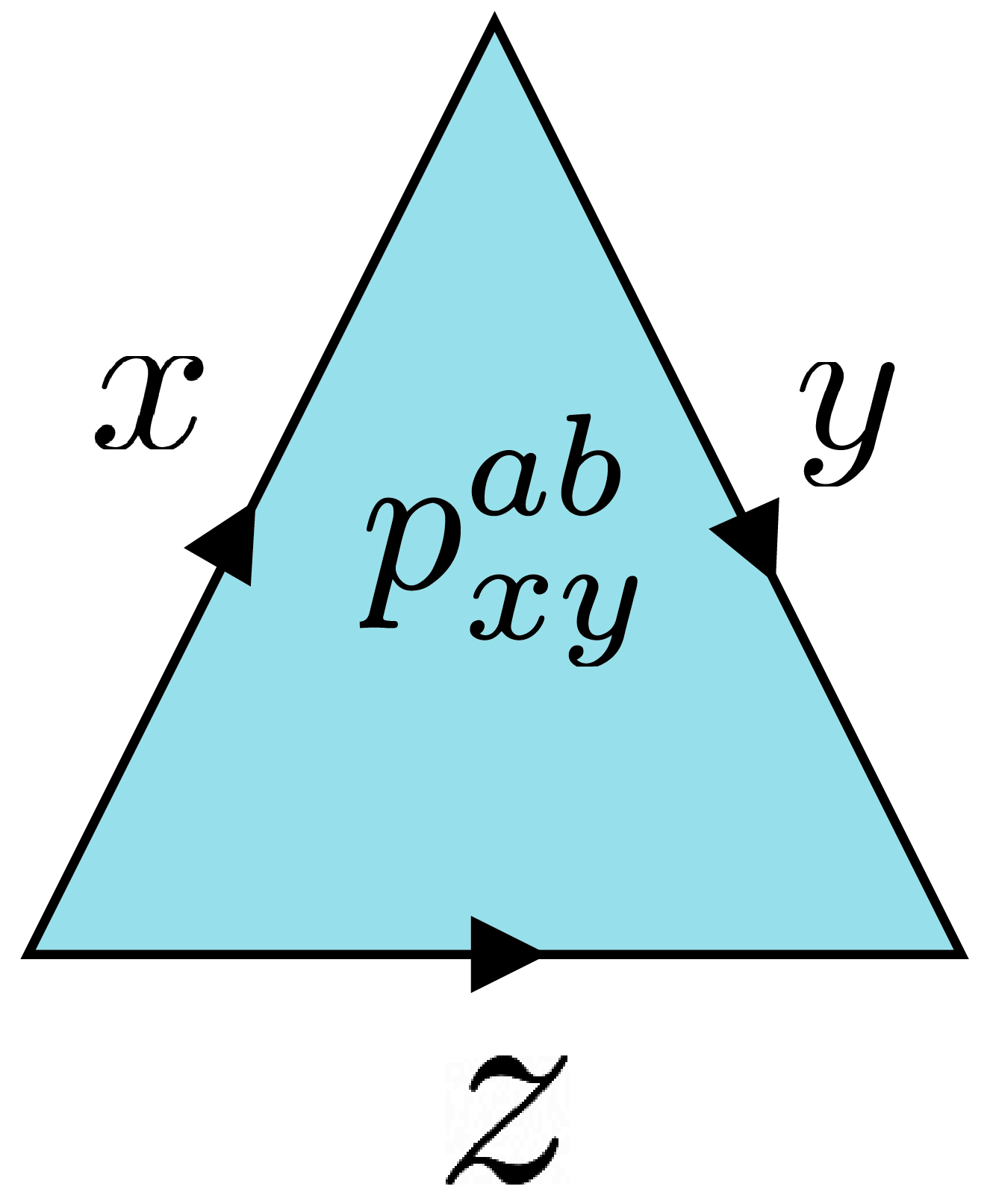}
\caption{
The $(2,1,2)$ Bell scenario in the simplicial setting.  
}
\label{fig:simple-scenarios}
\end{figure}


A topological representation of the CHSH scenario is given by four triangles glued along their $x_{i}$, $y_{j}$ edges. 
We assemble these four triangles into 
a punctured torus as in Fig.~(\ref{fig:punctured-torus}). That is, as a Mermin scenario with the $\{x_{0}\oplus y_{0},x_{1}\oplus y_{1},z\}$ and $\{x_{0}\oplus y_{1},x_{1}\oplus y_{0},z\}$ contexts removed. For convenience we denote the CHSH scenario as $T_{0}$ and the Mermin scenario as $T$.
\begin{figure}[h!]
\centering
  \includegraphics[width=.3\linewidth]{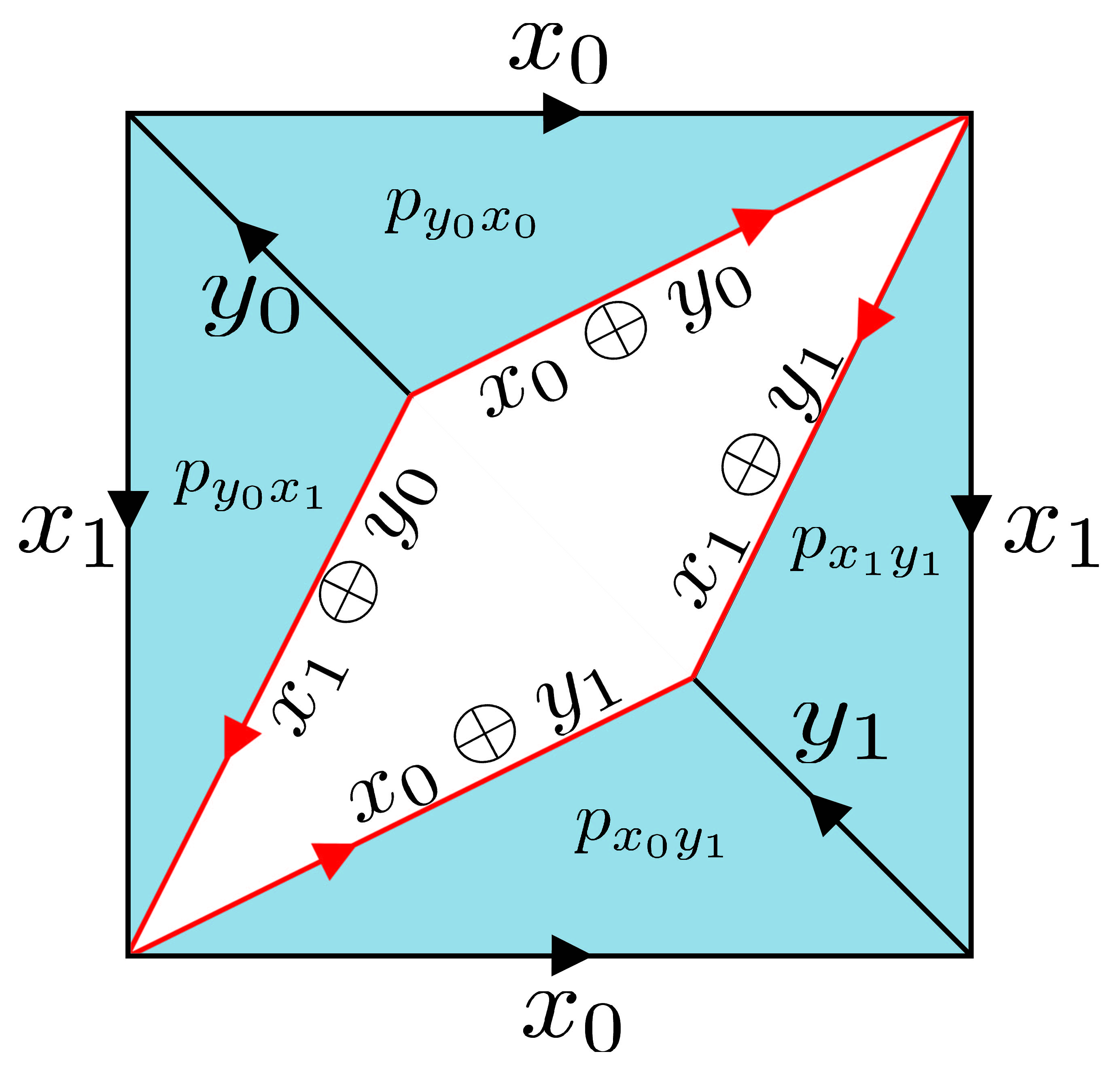}
\caption{ %
CHSH scenario represented topologically as a punctured torus.  
}
\label{fig:punctured-torus}
\end{figure}

\noindent Before we analyze this scenario, let us establish some terminology.
\Def{%
\cite[Def. 3.10]{okay2022simplicial}
A simplicial distribution $p$ is called \emph{noncontextual} if it can be written as a convex combination of \emph{deterministic} distributions.  
Otherwise we call it \emph{contextual}.
}




\noindent 
This notion of contextuality specializes to the usual notion for the CHSH scenario.
As is well-known, the CHSH scenario is contextual since there are distributions, the so-called Popescu-Rohrlich boxes \cite{popescu1994quantum}, which cannot be written as a probabilistic mixture of deterministic distributions. It was established by CHSH \cite{clauser1969proposed} that necessary for a distribution on the CHSH scenario to be noncontextual is that the following CHSH inequalities be satisfied:
\begin{equation}\label{eq:CHSH234}
\begin{aligned}
0\leq p_{x_0\oplus y_0}^0 + p_{x_0\oplus y_1}^0 + p_{x_1\oplus y_0}^0 - p_{x_1\oplus y_1}^0 \leq 2&\\
0\leq p_{x_0\oplus y_0}^0 + p_{x_0\oplus y_1}^0 - p_{x_1\oplus y_0}^0 + p_{x_1\oplus y_1}^0 \leq 2&\\
0\leq p_{x_0\oplus y_0}^0 - p_{x_0\oplus y_1}^0 + p_{x_1\oplus y_0}^0 + p_{x_1\oplus y_1}^0 \leq 2&\\
0\leq -p_{x_0\oplus y_0}^0 + p_{x_0\oplus y_1}^0 + p_{x_1\oplus y_0}^0 + p_{x_1\oplus y_1}^0 \leq 2&.
\end{aligned}
\end{equation}
Fine \cite{fine1982hidden,fine1982joint} then established the sufficiency of these inequalities:

\Thm{[Fine]\label{thm:fine}%
A distribution on the CHSH scenario is noncontextual if and only if the CHSH inequalities are satisfied.
}

To provide a new proof of Fine's theorem we will rely on a couple of key observations. One is that $T_{0}$ can be embedded into $T$ by inclusion, which allows us to study the CHSH scenario via the Mermin scenario. The other is the following immediate consequence of the vertex classification of $\MP_0$.


\Cor{\label{cor:Mermin-noncontextual}
Any distribution on the Mermin torus, whose topological realization is given in Fig.~(\ref{fig:mermin-scenario}), is noncontextual.
}
\Proof{
The distributions on the Mermin torus satisfying the nonsignaling conditions given in Eq.~(\ref{eq:Mermin-ns}) constitute the polytope $\MP_0$. In Theorem \ref{thm:VertexClassification} part (1) we have seen that all the vertices of this polytope are deterministic. Therefore any distribution on the Mermin torus can be written as a probabilistic mixture of deterministic distributions.
}

Next, we prove Proposition \ref{pro:noncontextual-torus}, which is stated in a more topological form below.

\begin{figure}[h!] 
  \centering
  \includegraphics[width=.5\linewidth]{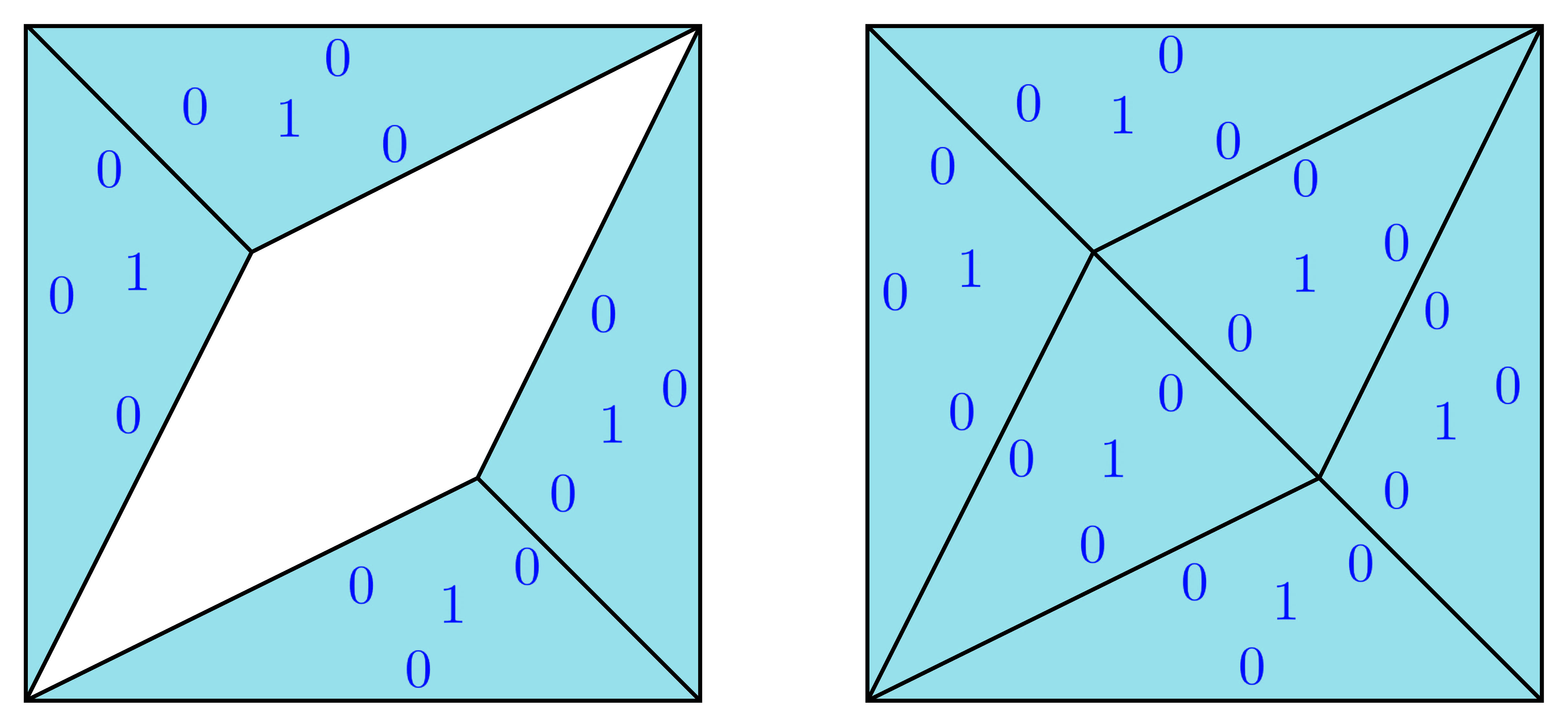}
\caption{(Left) The deterministic distribution $\delta^s$ on the punctured torus corresponding to the outcome assignment $s:x_0\mapsto 1,\;x_1,y_j\mapsto 0$. (Right) The extension $\tilde \delta^s$ of the distribution to the torus. 
}
\label{fig:extension}
\end{figure}  
\Pro{\label{pro:simpcont-4.7} 
A distribution $p$ on the punctured torus $T_0$ extends to a distribution on the torus $T$ if and only if $p$ is noncontextual.
}
\begin{proof} 
This result is a special case of the extension result proved in \cite[Pro. 4.7]{okay2022simplicial}.
The key observation is that an outcome assignment 
$s:\set{x_i,y_j:\,i,j,\in \ZZ_2}\to \ZZ_2$ specifies both a deterministic distribution $\delta^s$ on $T_0$ and a deterministic distribution on $T$, which we denote by $\tilde \delta^s$. See Fig.~(\ref{fig:extension}). Assume that $p$ is noncontextual. We can express $p$ as a probabilistic mixture $\sum_s \lambda(s)\, \delta^s $ of deterministic distributions. Then $\tilde p$ defined as the probabilistic mixture $\sum_s \lambda(s)\, \tilde \delta^s $ is the desired extension.
Conversely, assume that $p$ extends to a distribution $\tilde p$ on the torus. By Corollary
 \ref{cor:Mermin-noncontextual} 
 every distribution in $\MP_0$ is noncontextual, i.e., can be expressed as a probabilistic mixture of deterministic distributions $\tilde \delta^s$.
 Then restricting onto $T_0$ we can write  $p$ as a probabilistic mixture of $\delta^s$.
Thus $p$ is noncontextual.
\end{proof}

 Since an extension from $T_{0}$ to $T$ amounts to filling in the diamond whose boundary is given by the measurements $x_{i}\oplus y_{j}$, $i,j\in\ZZ_{2}$, it is useful to establish 
the following fact:

\begin{lem} 
\label{lem:diamond-chsh} 
A distribution $p$ on the boundary of the diamond scenario extends to the diamond if and only if $p$ satisfies the CHSH inequalities in Eq.~(\ref{eq:CHSH234}).
\end{lem} 
\Proof{%
This is proved in \cite[Pro 4.9]{okay2022simplicial}, we include the proof here for the convenience of the reader.
For our purposes we will assume that the diamond $Z$ is such that the triangles are glued along their XOR edge; see Fig.~(\ref{fig:diamond-extension}).
\begin{figure}[h!]
\centering
\begin{subfigure}{.49\textwidth}
  \centering
  \includegraphics[width=.3\linewidth]{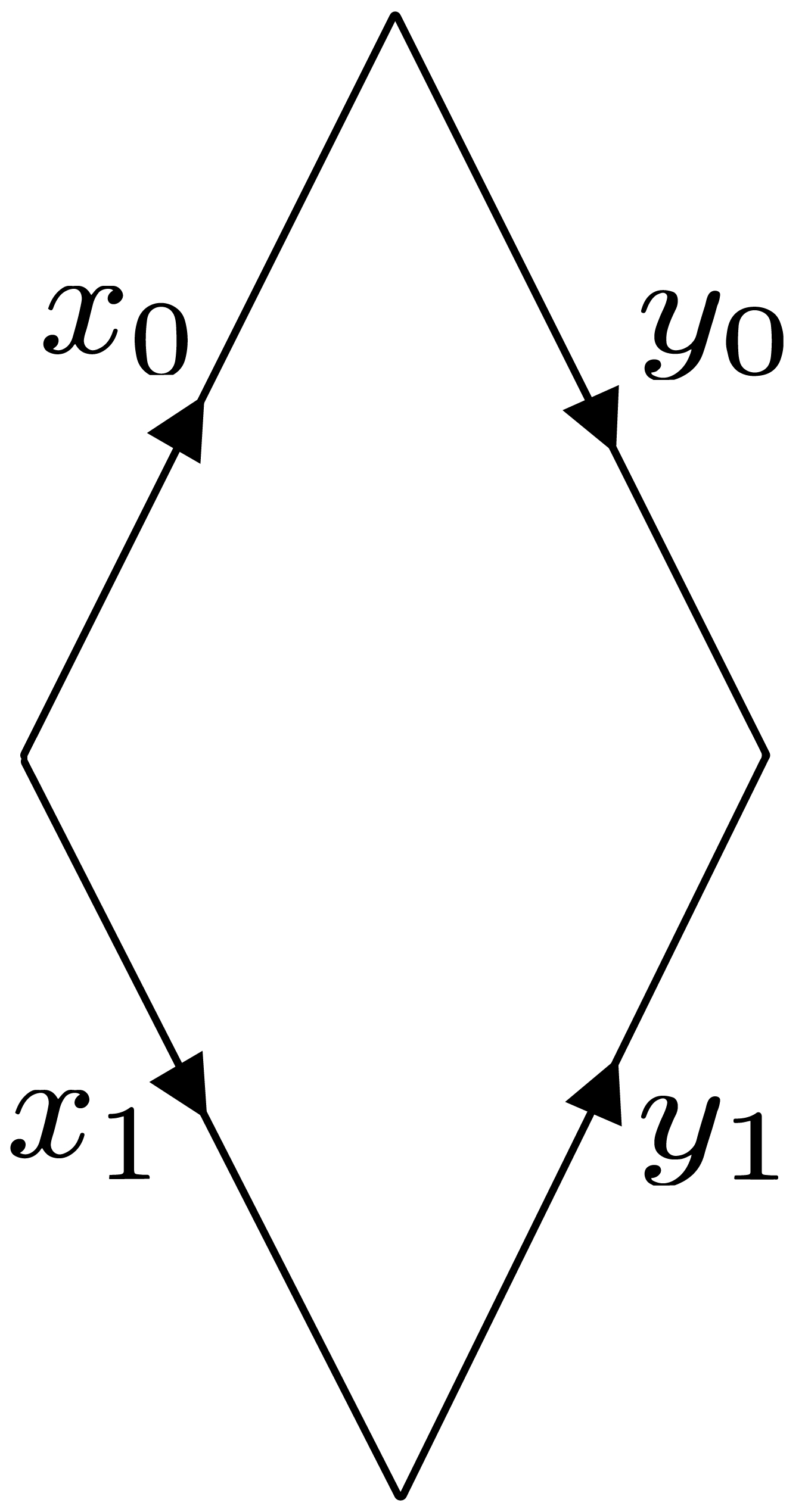}
  \caption{}
  \label{fig:diamond-ext}
\end{subfigure}
\begin{subfigure}{.49\textwidth}
  \centering
  \includegraphics[width=.3\linewidth]{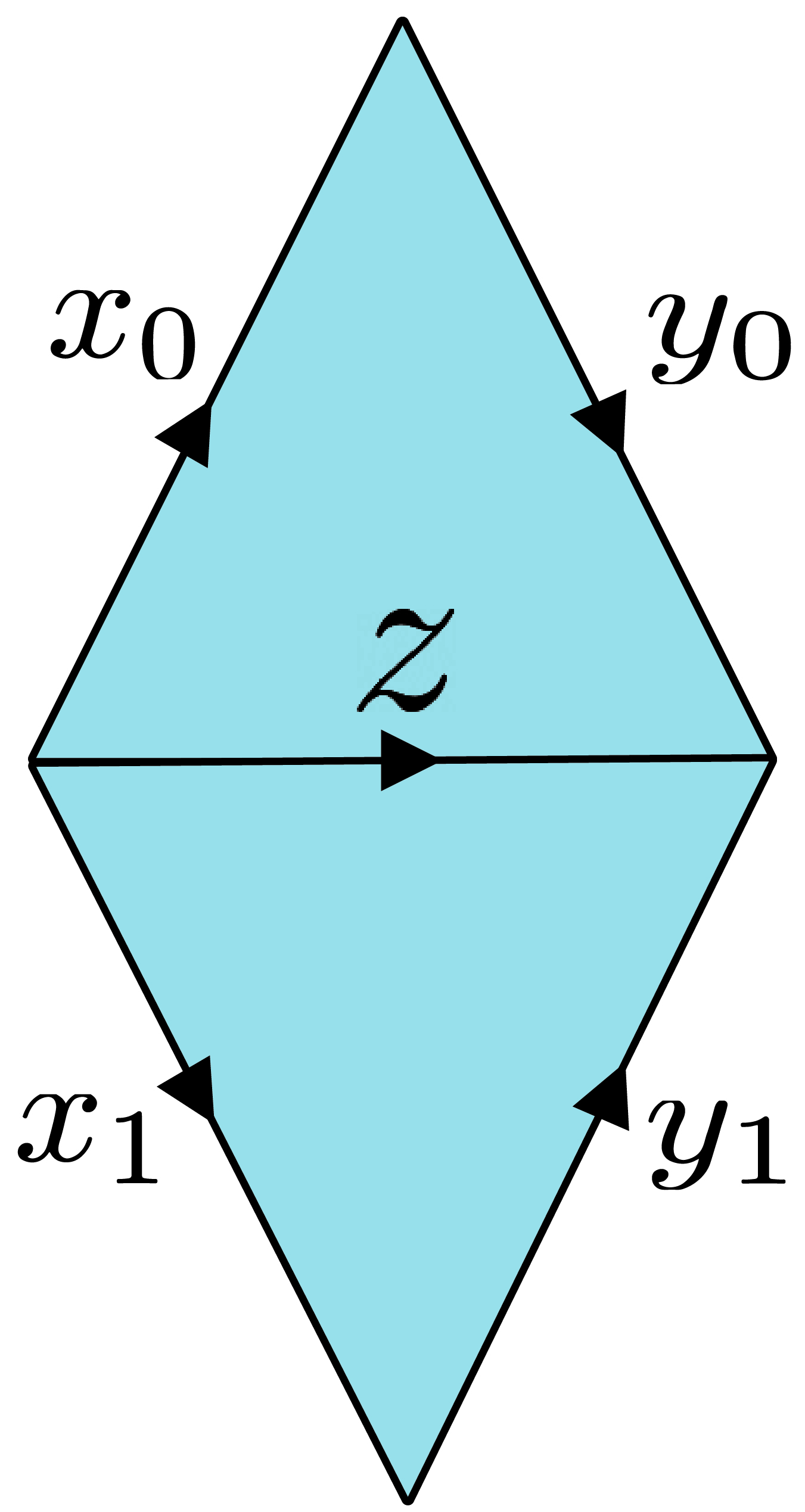}
  \caption{}
  \label{fig:diamond}
\end{subfigure}
\caption{(a) The boundary of the diamond. (b) Topological representation of the diamond scenario.
}
\label{fig:diamond-extension}
\end{figure}

The argument for the other choices is similar. The distribution $p_{\partial Z}$ on the boundary of the diamond is specified by $(p_{x_{0}}^0,p_{y_{0}}^0,q_{x_{1}}^0,q_{y_{1}}^0)\in [0,1]^4$. On the other hand, a distribution $p_{Z}$ on the diamond, requires compatible distributions $p_{xy}^{ab}$ and $q_{vw}^{rs}$, which by using Eq.~(\ref{eq:dist-by-edge}) can be specified by $(p_{x_{0}}^0,p_{y_{0}}^0,q_{x_{1}}^0,q_{y_{1}}^0,q_{z}^{0})$, where $q_{z}^{0}$ is the marginal along the common edge. It is possible to extend from   $\partial Z$ to $Z$ if and only if there exists a $q_{z}^{0}$ such that all $p_{x_{0}y_{0}}^{ab}, q_{x_{1}y_{1}}^{rs} \geq 0$. This occurs precisely when 
\begin{eqnarray}
\max\set{|p_{x_{0}}^0+p_{y_{0}}^0-1|,|q_{x_{1}}^0+q_{y_{1}}^0-1|} \leq  q_{z}^{0} \leq \min\set{1-|p_{x_{0}}^0-p_{y_{0}}^0|,1-|q_{x_{1}}^0-q_{y_{1}}^0|}.\notag
\end{eqnarray}
By Fourier-Motzkin elimination this single inequality is equivalent to the following four 
$$
\begin{aligned}
|p^0_{x_{0}}+p^0_{y_{0}}-1| &\leq 1-|p^0_{x_{0}}-p^0_{y_{0}}|\\
|p^0_{x_{0}}+p^0_{y_{0}}-1| &\leq 1-|q^0_{x_{1}}-q^0_{y_{1}}|\\
|q^0_{x_{1}}+q^0_{y_{1}}-1| &\leq 1-|p^0_{x_{0}}-p^0_{y_{0}}|\\
|q^0_{x_{1}}+q^0_{y_{1}}-1| &\leq 1-|q^0_{x_{1}}-q^0_{y_{1}}|,
\end{aligned}
$$
in addition to the trivial inequalities corresponding to $0\leq p_{i_{0}}^{0}, q_{i_{1}}^{0}\leq 1$, where $i = x,y$. Expanding the absolute values gives the inequalities
\begin{equation}\label{eq:CHSH-prob}
\begin{aligned}
0&\leq  p_{x_{0}}^0 + p_{y_{0}}^0 + q_{x_{1}}^0 - q_{y_{1}}^0 \leq 2 \\
0&\leq  p_{x_{0}}^0 + p_{y_{0}}^0 - q_{x_{1}}^0 + q_{y_{1}}^0 \leq 2 \\
0&\leq  p_{x_{0}}^0 - p_{y_{0}}^0 + q_{x_{1}}^0 + q_{y_{1}}^0 \leq 2 \\
0&\leq  -p_{x_{0}}^0 + p_{y_{0}}^0 + q_{x_{1}}^0 + q_{y_{1}}^0 \leq 2. 
\end{aligned}
\end{equation}
These equations are formally identical to the CHSH inequalities appearing in Eq.(\ref{eq:CHSH234}).
}



\begin{proof}[{\bf Proof of Theorem~\ref{thm:fine}}]
Let $p$ be a distribution on $T_0$ and $p_\partial$ denote the restriction (marginalization) of $p$ to the boundary of $T_0$. 
Observe that the torus is obtained from the punctured torus by filling in the diamond in the middle. Therefore $p$ extends to $T$ if and only if $p_\partial$ extends to the diamond.
Combining this observation with Proposition \ref{pro:simpcont-4.7} and Lemma \ref{lem:diamond-chsh} gives the desired result. 
\end{proof}

 \subsection{Decomposing the $2$-qubit $\Lambda$-polytope}


In this section, we provide a decomposition for $\Lambda_2$, the $\Lambda$-polytope for $2$-qubits, using the Mermin polytope $\MP_1$. This decomposition will provide valuable insight into the vertex enumeration problem for $\Lambda$-polytopes. This problem is a fundamental mathematical obstacle in the complexity analysis of the $\Lambda$-simulation algorithm introduced in \cite{zurel2020hidden}. 

\begin{figure}[h!]
\centering
\begin{subfigure}{.49\textwidth}
  \centering
  \includegraphics[width=.7\linewidth]{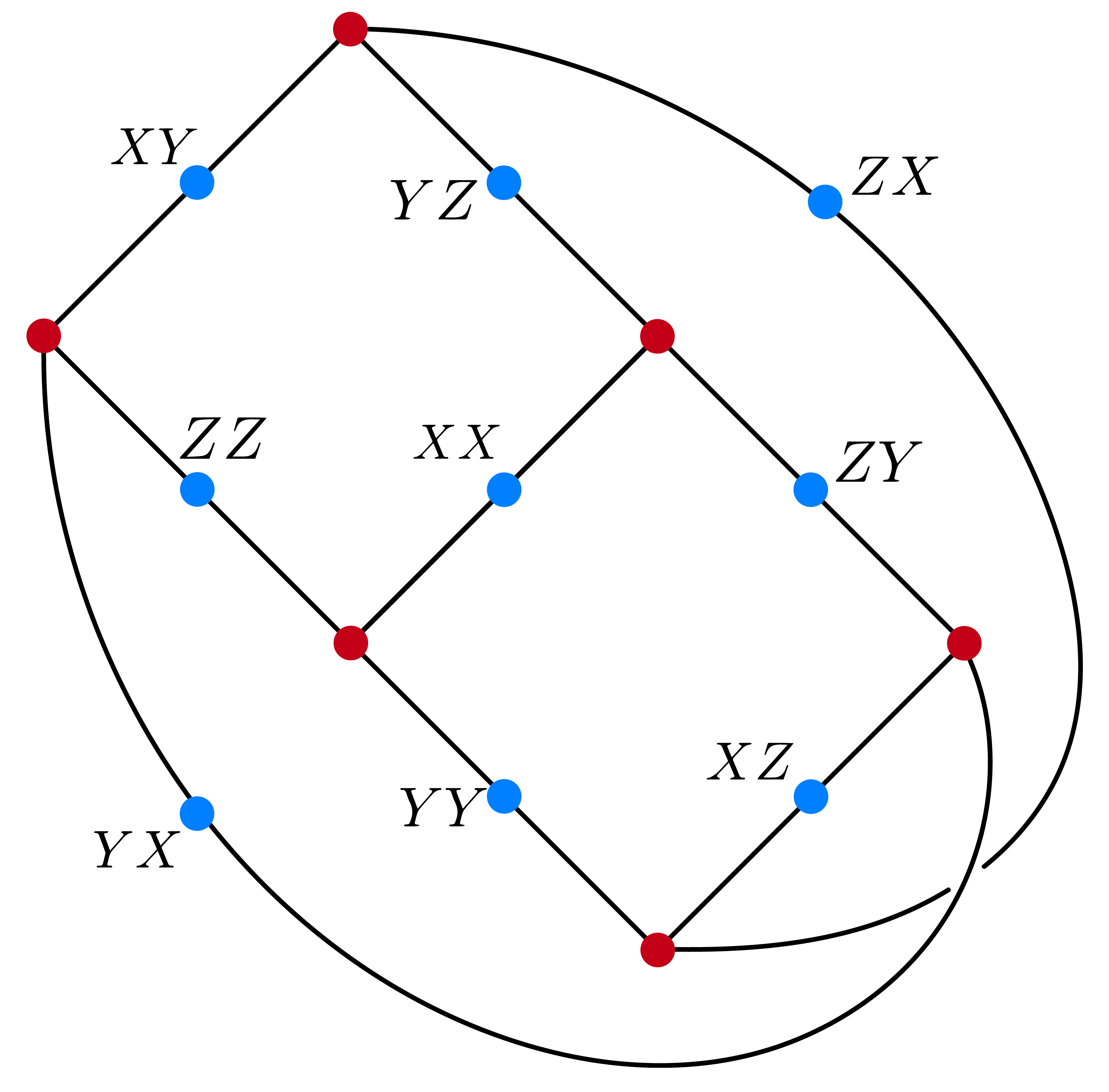}
  \caption{}
  \label{fig:I2-nloc}
\end{subfigure}%
\begin{subfigure}{.49\textwidth}
  \centering
  \includegraphics[width=.7\linewidth]{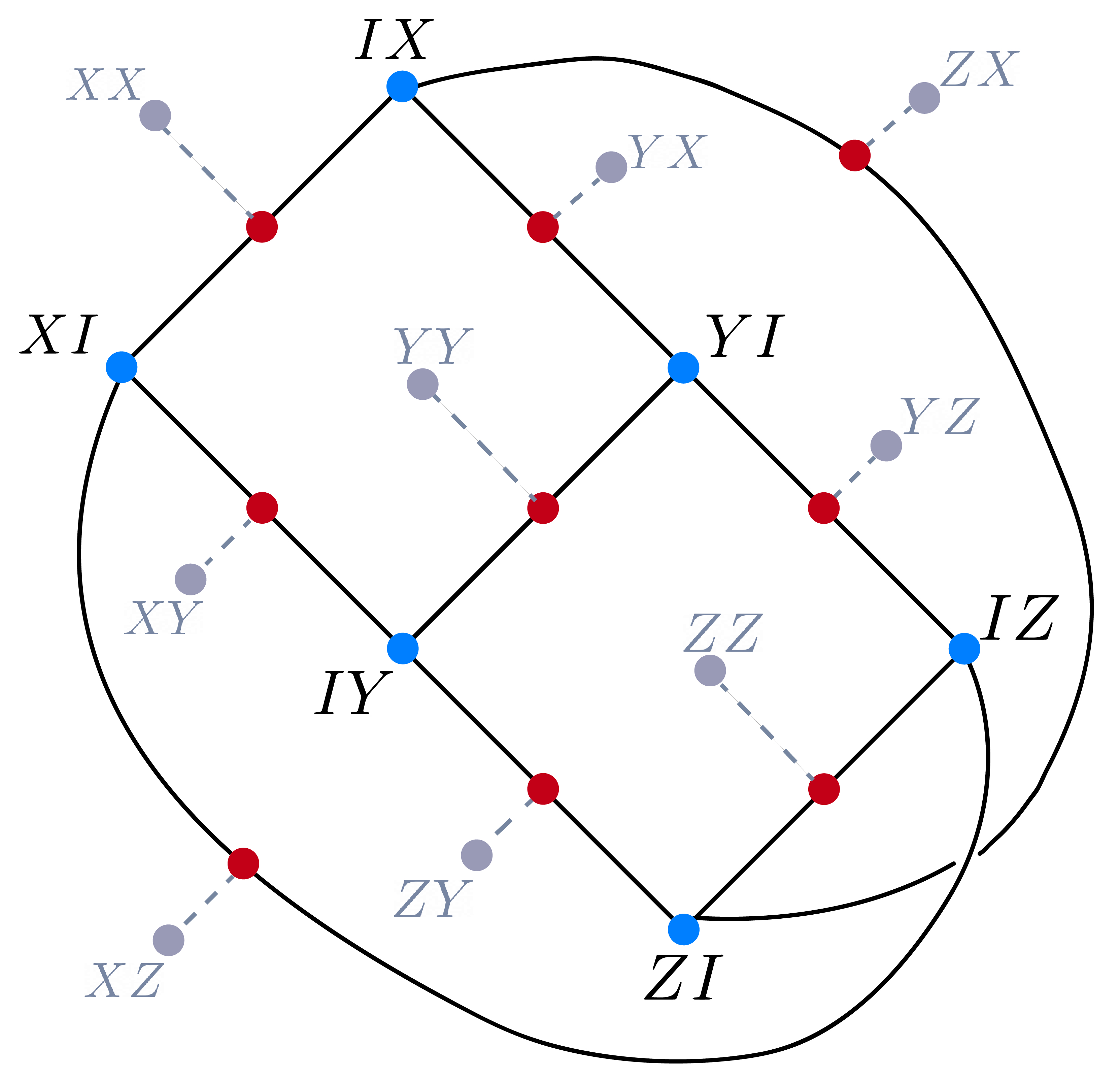}
  \caption{}
  \label{fig:I2-loc}
\end{subfigure}
\caption{ (a) Subspaces in $\iI_2^\nloc$ and their intersections (zero subspace is omitted).  
(b)  Subspaces in $\iI_2^\loc$ and their intersections (nonlocal operators do not belong to this set, they are only indicated to reveal the connection to the nonlocal part).
}
\label{fig:I2}
\end{figure}

Recall the set $\sS_2$ of $2$-qubit stabilizer states and the (non)local version from Eq.~(\ref{eq:stabilizer-loc-nloc}).
The $2$-qubit $\Lambda$-polytope is defined as follows:
$$
\Lambda_2 = \set{X\in \Herm((\CC^2)^{\otimes 2})):\, \Tr(X)=1,\; \Tr(X \Pi)\geq 0,\;\forall\, \Pi\in \sS_2  }
$$ 
Our decomposition will be derived from the local vs. nonlocal decomposition of Pauli operators introduced in Section \ref{sec:MP1}.  
Let us write $E^\loc$ and $E^\nloc$ for the subsets of $E$ corresponding to the local and nonlocal Pauli operators; respectively. This gives us the following decomposition:
$$
E = \set{0} \sqcup E^\loc \sqcup E^\nloc.
$$
Let $\iI$ denote the set of maximal isotropic subspaces in $E$. 
This set also decomposes into a local $\iI^\loc$ and a nonlocal part $\iI^\nloc$; see Fig.~(\ref{fig:I2}).
Recall that the Mermin scenario $(M,\cC)$ can be identified with $(E^\nloc,\iI^\nloc)$ via the map in Eq.~(\ref{eq:embedding-M-E}).  
The function $\beta_1: \cC \to \ZZ_2$ extends to a function $\tilde \beta: \iI\to \ZZ_2$ where $\tilde \beta(C)=0$ for $C\in \iI^\loc$. We begin with a result that is a local version of Lemma \ref{lem:M1-quantum}. 
We define the local $2$-qubit $\Lambda$-polytope: 
$$
\Lambda_2^\loc = \set{X\in \Herm((\CC^2)^{\otimes 2})):\, \Tr(X)=1,\; \Tr(X \Pi)\geq 0,\;\forall\, \Pi\in \sS_2^\loc  }.
$$ 
Note that by definition 
$
\Lambda_2 \subset \Lambda_2^\loc.
$
This local polytope is, in fact, a well-known nonsignaling polytope. The $(2,3,2)$ Bell scenario consists of
\begin{itemize}
\item the measurement set $M_{232}=
\set{x_i,y_j:\, i,j,\in \ZZ_3}$,
\item the collection $\cC_{232}$ of contexts $C_{ij}$, where $i,j\in \ZZ_3$, given by
$$
C_{ij} = \set{x_i,y_j}.
$$
\end{itemize}

\Lem{\label{lem:Lambda-local}
The local polytope $\Lambda^\loc_2$ can be identified with  the nonsignaling polytope $\NS_{232}$ of the $(2,3,2)$ Bell scenario. 
}
\Proof{
The argument is similar to the proof of Lemma \ref{lem:M1-quantum}. 
An operator $X\in \Lambda_2^\loc$ specifies a distribution $p_X$ in $\NS_{\iI^\loc,0}$ (see Eq.~(\ref{eq:NSbeta}) for the definition of the polytope) via the Born rule.  
For the bijection we need an inverse map, which comes from by first marginalizing to a single measurement and then computing the expectation $\Span{A}_X$ of the corresponding Pauli operator. 
The identification of $\NS_{\iI^\loc,0}$ with the nonsignaling polytope $\NS_{232}$ follows from realizing the measurements in the Bell scenario as quantum mechanical measurements
\begin{equation}\label{eq:embedding-local}
\begin{aligned}
x_0 \mapsto X\otimes \one,\;\;\; x_1 \mapsto Y\otimes \one,\;\;\; x_2 \mapsto Z\otimes \one\\
y_0 \mapsto \one\otimes X,\;\;\; y_1 \mapsto \one\otimes Y,\;\;\; y_2 \mapsto \one\otimes Z.
\end{aligned}
\end{equation}
}

Let $\MP_1^\RR$ denote the Mermin polytope for quasiprobability distributions; see Definition \ref{def:MerminPolytope}. 
We introduce an important map
\begin{equation}\label{eq:ext}
\ext: \NS_{232} \to \MP_1^\RR
\end{equation}
using the identification of Lemma \ref{lem:Lambda-local}. 
For the explicit description of the $\ext$ map  we need to extend the $(2,3,2)$ Bell scenario $(M_{232},\cC_{232})$  
by including the nonlocal measurements. 
Formally we introduce an extended scenario:
\begin{itemize}
\item $\tilde M= M_{232} \sqcup \set{x_i\oplus y_j\,:\,i,j\in \ZZ_3}$,
\item $\tilde \cC =\set{\tilde C_{ij}\,:\, i,j\in \ZZ_3 }$ where $\tilde C_{ij}=C_{ij} \sqcup \set{x_i\oplus y_j}$.
\end{itemize}  
Now we are ready to describe the $\ext$ map explicitly.
Let $d=\set{d_{C_{ij}}}_{i,j\in \ZZ_3}$   be a nonsignaling distribution defined on the $(2,3,2)$ Bell scenario. 
We define $\tilde d$ as a nonsignaling distribution on $(\tilde M,\tilde \cC)$ by setting 
$$
\tilde d_{\tilde C_{ij}}(s) = \left\lbrace
\begin{array}{ll}
d_{C_{ij}}(s|_{C_{ij}}) & s(x_i\oplus y_j) = s(x_i) + s(y_j) \\
0 & \text{otherwise}.
\end{array}
\right.
$$
The mapping in Eq.~(\ref{eq:embedding-local}) can be used to define an embedding $M_{232} \subset E$. Together with the embedding of Eq.~(\ref{eq:embedding-M-E}) we   obtain a local vs nonlocal decomposition
$$
E = M_{232} \sqcup M. 
$$
With this convention we will give the explicit form of the $\ext$ map on a context of the form 
$$C=\set{(v,w), (v',w'), (v+v',w+w') }$$
 with $\omega((v,w),(v',w'))=0$. 
For $s\in \ZZ_2^{C}$ we set $a=s(v , w)$, $b=s(v' , w')$ and $c=a+b+\beta((v,w),(v',w'))$.
Then we have
\begin{equation}\label{eq:ext-def}
(\ext\, d)_C(s) =  \frac{1}{2}\left( \tilde d|_{\set{(v, w)}}(a)  +\tilde d|_{\set{(v' , w')}}(b) -\tilde d|_{\set{(v+w,v'+w')}}(c+1)     \right)
\end{equation}
if $s(v+w,v'+w') = s(v , w) + s(v', w')+\beta((v,w),(v',w'))$ and $(\ext\,d)_C(s)=0$ otherwise.  

\Thm{\label{thm:Lambda2NS}
The polytope $\Lambda_2$ is precisely the subpolytope of the nonsignaling polytope $\NS_{232}$ for the $(2,3,2)$ Bell scenario given by those distributions that map to a probability distribution in $\MP_1$ under the $\ext$ map given in Eq.~(\ref{eq:ext-def}).
}
\Proof{
Using the identification  given in  Lemma \ref{lem:Lambda-local} and the operator-theoretic description of $\MP_1^\RR$ in Lemma \ref{lem:M1-quantum} 
 the $\ext $ map given in Eq.~(\ref{eq:ext}) can be identified with the map  
\begin{equation}\label{eq:local-to-Mermin}
\Lambda_2^\loc \to \MP_1^\RR
\end{equation}
obtained by sending $X$  to the operator $\bar X$ such that $\Span{\bar X}_A = \Span{X}_A$ for nonlocal Pauli operators $A$ (including $\one$) and $\Span{\bar X}_A = 0$ for the remaining local Pauli operators.  
Those operators $X\in \Lambda_2^\loc$ which give a probability distribution on the Mermin scenario, instead of a quasiprobability distribution, are precisely those that come from $\Lambda_2$. 
}

Theorem \ref{thm:Lambda2NS} gives a description of $\Lambda_2$ in terms of well-understood polytopes: $\NS_{232}$ whose vertices are described in \cite{jones2005interconversion} and $\MP_1$ described in Theorem \ref{thm:VertexClassification}.

\section{Conclusion}\label{sec:conclusion}
Motivated by a classic example of contextuality known as Mermin's square \cite{mermin1993hidden}, and its topological realization given in \cite{Coho}, in this paper we considered variations of this scenario parametrized by a function $\beta$
and studied the corresponding nonsignaling polytopes $\MP_{\beta}$. We showed that these polytopes 
fall into two equivalence classes, determined by $[\beta]$, which has a cohomological interpretation \cite{Coho}. Among our main results is the characterization of the vertices of $\MP_{\beta}$. We demonstrated that all vertices of $\MP_{0}$ are deterministic, which facilitates a novel proof of Fine's theorem \cite{fine1982hidden,fine1982joint}. On the other hand, $\MP_{1}$ has two types of vertices, both of which are cnc \cite{raussendorf2020phase}.
We also described the graphs associated with the polytopes. In the case of $\MP_1$, the edges in this graph are essentially given by the loops on the Mermin torus.
These loops correspond to complements of cnc sets and play a significant role throughout the paper.

An important connection is established between $\MP_{1}$ and computation through the notion of $\Lambda$-simulation \cite{zurel2020hidden}. Indeed, if one restricts to just measurements of non-local Pauli operators then one can define a simulation algorithm for $\MP_{1}$ in the spirit of \cite{zurel2020hidden}; although, since the vertices of $\MP_{1}$ are cnc, all resulting quantum computations can be efficiently simulated classically \cite{raussendorf2020phase}. Alternatively, here we have established that $\MP_{1}$ corresponds precisely to the non-local part of the polytope $\Lambda_{2}$ \cite{zurel2020hidden,okay2021extremal}, with the local part being related to $\NS_{232}$, the polytope associated with the $(2,3,2)$ Bell scenario \cite{jones2005interconversion}. We expect that this decomposition will be important in understanding the combinatorial structure of $\Lambda_{2}$; an important first step in analyzing the complexity of classical simulation based on $\Lambda$-polytopes.

An interesting but yet unexplored topic of research is the study more generally of polytopes associated with measurement scenarios, or topological spaces with non-trivial $[\beta]$. An interesting example of this is the well-known Mermin star scenario \cite{mermin1993hidden}, which also has a topological realization as a torus \cite{Coho}. Particularly appealing about this line of research is that the Mermin's star is closely related to the so-called Greenberger-Horne-Zeilinger (GHZ) paradox \cite{greenberger1989going}, which can be exploited for computational advantage; see \cite{anders2009computational,raussendorf2013contextuality}.

\paragraph{Acknowledgments.}
This  work is supported by the Air Force Office of Scientific Research under award
number FA9550-21-1-0002.


\begin{thebibliography}{10}

\bibitem{bohm1952suggested}
D.~Bohm, ``A suggested interpretation of the quantum theory in terms of
  `hidden' variables. {I},'' {\em Physical review}, vol.~85, no.~2, p.~166,
  1952.

\bibitem{schmid2020unscrambling}
D.~Schmid, J.~H. Selby, and R.~W. Spekkens, ``Unscrambling the omelette of
  causation and inference: The framework of causal-inferential theories,'' {\em
  arXiv preprint arXiv:2009.03297}, 2020.

\bibitem{caticha2022entropic}
A.~Caticha, ``Entropic dynamics and quantum `measurement','' {\em arXiv
  preprint arXiv:2208.02156}, 2022.

\bibitem{bell1966problem}
J.~S. Bell, ``On the problem of hidden variables in quantum mechanics,'' {\em
  Reviews of Modern physics}, vol.~38, no.~3, p.~447, 1966.

\bibitem{kochen1975problem}
S.~Kochen and E.~P. Specker, ``The problem of hidden variables in quantum
  mechanics,'' in {\em The logico-algebraic approach to quantum mechanics},
  pp.~293--328, Springer, 1975.

\bibitem{budroni2021quantum}
C.~Budroni, A.~Cabello, O.~G{\"u}hne, M.~Kleinmann, and J.-{\AA}. Larsson,
  ``Quantum contextuality,'' {\em arXiv preprint arXiv:2102.13036}, 2021.

\bibitem{mermin1993hidden}
N.~D. Mermin, ``Hidden variables and the two theorems of {J}ohn {B}ell,'' {\em
  Reviews of Modern Physics}, vol.~65, no.~3, p.~803, 1993.
\newblock doi:
  \href{https://link.aps.org/doi/10.1103/RevModPhys.65.803}{10.1103/RevModPhys.65.803}.

\bibitem{cleve2014characterization}
R.~Cleve and R.~Mittal, ``Characterization of binary constraint system games,''
  in {\em International Colloquium on Automata, Languages, and Programming},
  pp.~320--331, Springer, 2014.

\bibitem{Coho}
C.~Okay, S.~Roberts, S.~D. Bartlett, and R.~Raussendorf, ``Topological proofs
  of contextuality in quantum mechanics,'' {\em Quantum Information \&
  Computation}, vol.~17, no.~13-14, pp.~1135--1166, 2017.
\newblock doi:
  \href{https://doi.org/10.26421/QIC17.13-14-5}{10.26421/QIC17.13-14-5}. arXiv:
  \href{https://arxiv.org/abs/1701.01888}{1701.01888}.

\bibitem{okay2022simplicial}
C.~Okay, A.~Kharoof, and S.~Ipek, ``Simplicial quantum contextuality,'' {\em
  arXiv preprint arXiv:2204.06648}, 2022.

\bibitem{clauser1969proposed}
C.~Horne, M.~Horne, A.~Shimony, and H.~Richard, ``Proposed experiment to test
  local hidden-variable theories,'' {\em Physical Review Letters}, p.~880,
  1969.

\bibitem{fine1982hidden}
A.~Fine, ``Hidden variables, joint probability, and the bell inequalities,''
  {\em Physical Review Letters}, vol.~48, no.~5, p.~291, 1982.

\bibitem{zurel2020hidden}
M.~Zurel, C.~Okay, and R.~Raussendorf, ``Hidden variable model for universal
  quantum computation with magic states on qubits,'' {\em Physical Review
  Letters}, vol.~125, no.~26, p.~260404, 2020.

\bibitem{gawrilow2000polymake}
E.~Gawrilow and M.~Joswig, ``polymake: a framework for analyzing convex
  polytopes,'' {\em Polytopes — Combinatorics and Computation}, p.~43–73,
  2000.

\bibitem{okay2021extremal}
C.~Okay, M.~Zurel, and R.~Raussendorf, ``On the extremal points of the
  {$Lambda$}-polytopes and classical simulation of quantum computation with
  magic states,'' {\em arXiv preprint arXiv:2104.05822}, 2021.

\bibitem{raussendorf2020phase}
R.~Raussendorf, J.~Bermejo-Vega, E.~Tyhurst, C.~Okay, and M.~Zurel,
  ``Phase-space-simulation method for quantum computation with magic states on
  qubits,'' {\em Physical Review A}, vol.~101, no.~1, p.~012350, 2020.

\bibitem{jones2005interconversion}
N.~Jones and L.~Masanes, ``Interconversion of nonlocal correlations,'' {\em
  Physical Review A}, p.~43–73, 2005.

\bibitem{fine1982joint}
A.~Fine, ``Joint distributions, quantum correlations, and commuting
  observables,'' {\em Journal of Mathematical Physics}, vol.~23, no.~7,
  pp.~1306--1310, 1982.

\bibitem{abramsky2011sheaf}
S.~Abramsky and A.~Brandenburger, ``The sheaf-theoretic structure of
  non-locality and contextuality,'' {\em New Journal of Physics}, vol.~13,
  no.~11, p.~113036, 2011.
\newblock doi:
  \href{https://doi.org/10.1088/1367-2630/13/11/113036}{10.1088/1367-2630/13/11/113036}.
  arXiv: \href{https://arxiv.org/abs/1102.0264}{1102.0264}.

\bibitem{sreekumar2021automorphism}
K.~Sreekumar and K.~Manilal, ``Automorphism groups of some families of
  bipartite graphs.,'' {\em Electron. J. Graph Theory Appl.}, vol.~9, no.~1,
  pp.~65--75, 2021.

\bibitem{ziegler2012lectures}
G.~M. Ziegler, {\em Lectures on polytopes}, vol.~152.
\newblock Springer Science \& Business Media, 2012.

\bibitem{chvatal1983linear}
V.~Chvatal, V.~Chvatal, {\em et~al.}, {\em Linear programming}.
\newblock Macmillan, 1983.

\bibitem{godsil2001algebraic}
C.~Godsil and G.~F. Royle, {\em Algebraic graph theory}, vol.~207.
\newblock Springer Science \& Business Media, 2001.

\bibitem{popescu1994quantum}
S.~Popescu and D.~Rohrlich, ``Quantum nonlocality as an axiom,'' {\em
  Foundations of Physics}, vol.~24, no.~3, pp.~379--385, 1994.

\bibitem{greenberger1989going}
D.~M. Greenberger, M.~A. Horne, and A.~Zeilinger, ``Going beyond bell’s
  theorem,'' in {\em Bell’s theorem, quantum theory and conceptions of the
  universe}, pp.~69--72, Springer, 1989.

\bibitem{anders2009computational}
J.~Anders and D.~E. Browne, ``Computational power of correlations,'' {\em
  Physical Review Letters}, vol.~102, no.~5, p.~050502, 2009.

\bibitem{raussendorf2013contextuality}
R.~Raussendorf, ``Contextuality in measurement-based quantum computation,''
  {\em Physical Review A}, vol.~88, no.~2, p.~022322, 2013.

\bibitem{bray2011short}
J.~Bray, M.~Conder, C.~Leedham-Green, and E.~O'Brien, ``Short presentations for
  alternating and symmetric groups,'' {\em Transactions of the American
  Mathematical Society}, vol.~363, no.~6, pp.~3277--3285, 2011.

\end{thebibliography}

\appendix

\section{Proof of Proposition \ref{pro:MPbeta-cohomology} }\label{sec:ProofPro}

In this section we will prove Proposition \ref{pro:MPbeta-cohomology}. 
For this we will introduce a generalized version of the Mermin polytope (Definition \ref{def:MerminPolytope}). Recall the $K_{3,3}$ graph associated to the Mermin scenario with vertex set $\cC=\cC^\hor\sqcup \cC^\ver$ and edge set $M$; see Fig.~(\ref{fig:mermin-scenario-k33}). We begin with generalizing the definition of $\beta$. Let $R$ denote the set of pairs $(C,m)\in \cC\times M$ such that $m\in C$. We will consider incidence weights on $K_{3,3}$, that is functions of the form $\beta:R\to \ZZ_2$. Let us write $K_{3,3}^\beta$ to indicate the weight. 

\begin{defn}\label{def:MerminPolyGeneral}
{\rm
Let $\widetilde \MP_\beta$ denote the polytope 
given by the set of functions
$$
p: R\to \RR_{\geq 0}
$$
satisfying the following conditions:
\begin{enumerate}[(a)]
\item $\sum_{m\in C} p(C,m) \leq 1$ for all $C\in \cC$,
\item For a context $C\in \cC$ define $p_C:C\to \RR_{\geq 0}$ by 
$$
p_C(m) = \left\lbrace
\begin{array}{cc}
\sum_{m'\in C-\set{m}} p(C,m') & \beta(C)=1, \\
1-\sum_{m'\in C-\set{m}} p(C,m') & \beta(C)=0. \\ 
\end{array}
\right.
$$
Then for all $m\in M$ and $C,C'\in \cC$ such that $m\in C\cap C'$ we require that
$$
p_C(m) = p_{C'}(m).
$$
\end{enumerate}
}
\end{defn}

\begin{figure}[h!]
\centering
\begin{subfigure}{.33\textwidth}
  \centering
  \includegraphics[width=.6\linewidth]{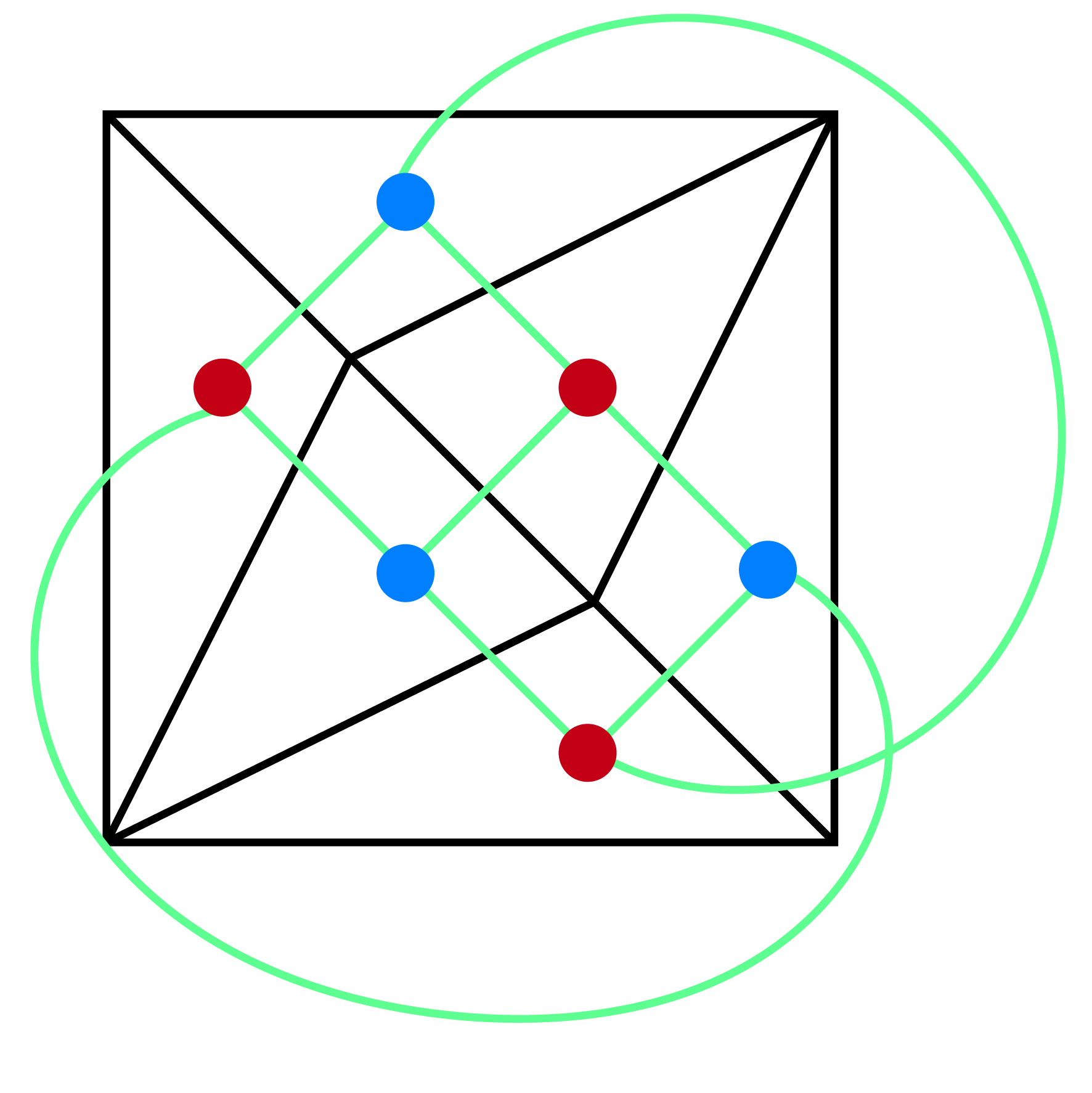}
  \caption{}
  \label{fig:K33-beta0}
\end{subfigure}%
\begin{subfigure}{.33\textwidth}
  \centering
  \includegraphics[width=.6\linewidth]{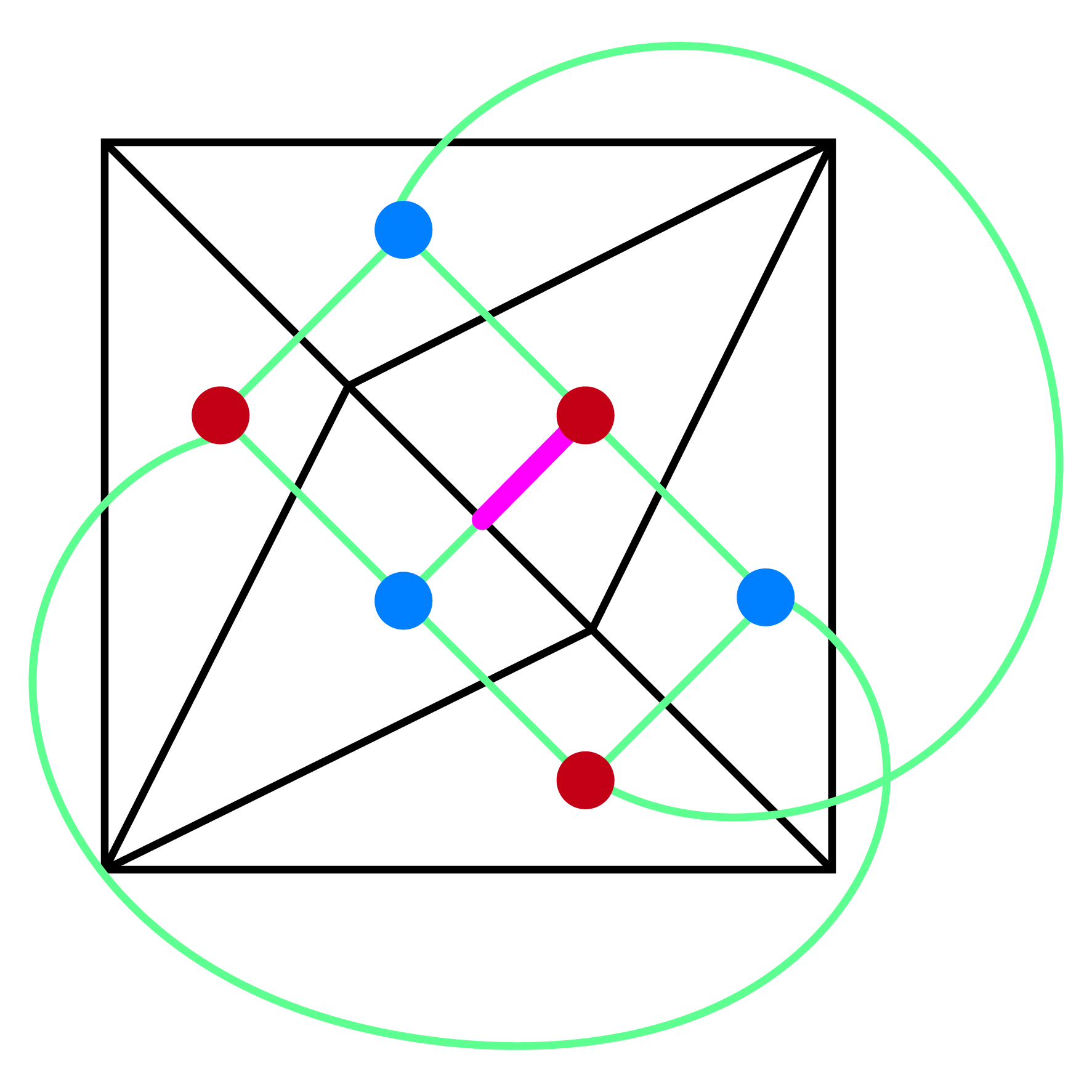}
  \caption{}
  \label{fig:K33-beta1}
\end{subfigure}
\begin{subfigure}{.33\textwidth}
  \centering
  \includegraphics[width=.6\linewidth]{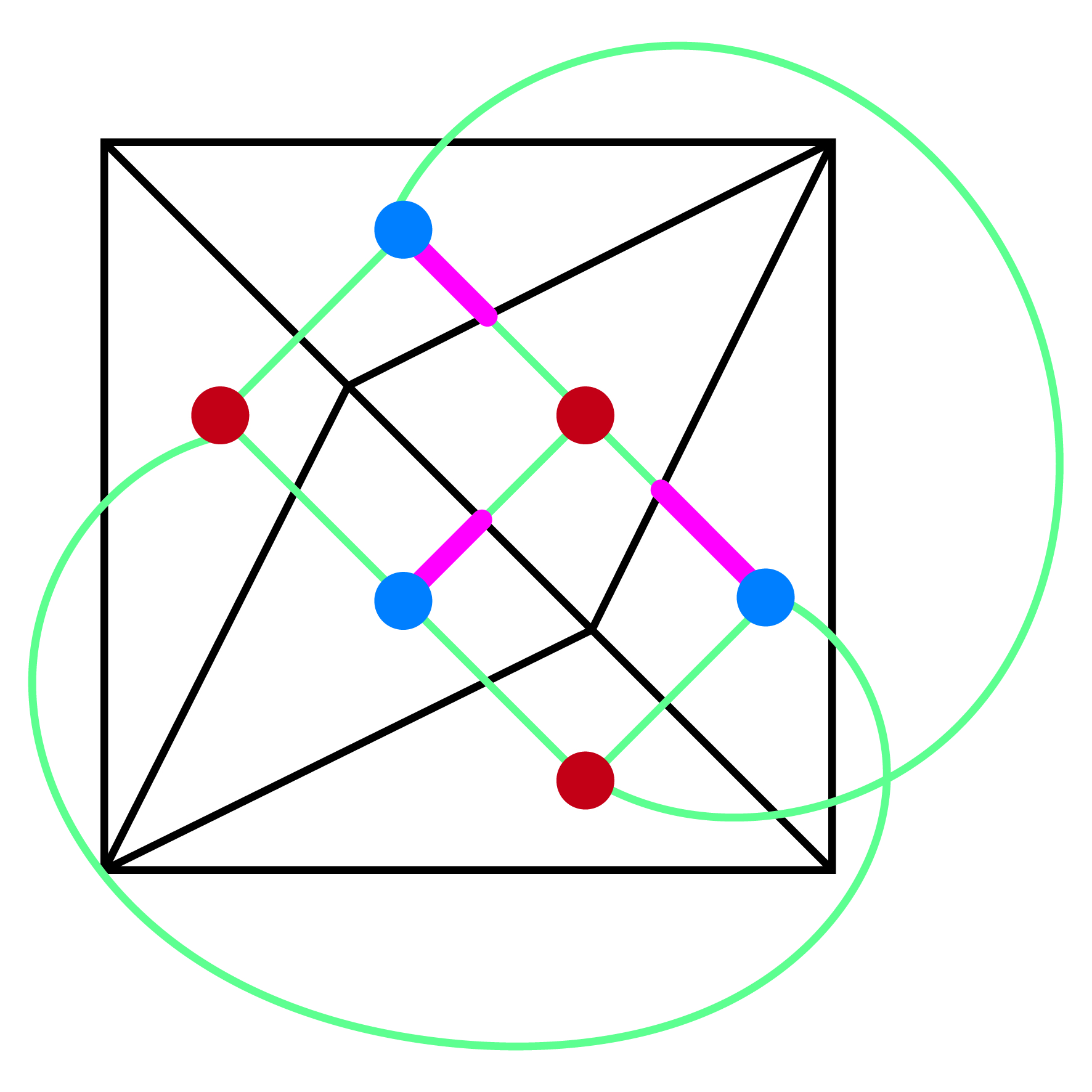}
  \caption{}
  \label{fig:K33-beta3}
\end{subfigure}
\caption{$K_{3,3}^\beta$ for three different choices of $\beta$. Pink color on the part of an edge $x$ incident to $C$ implies that $\beta(C,x)=1$, otherwise $\beta$ takes the value of zero. 
}
\label{fig:K33-beta}
\end{figure}

To have a better idea of this definition consider a context $C=\set{x,y,z}$ and let 
\begin{equation}\label{eq:abcd}
a=p(C,x),\;\;\;
b=p(C,y),\;\;\;
c=p(C,z),\;\;\;
d=1-(a+b+c).
\end{equation}
 Condition (a) says that $p=\set{a,b,c,d}$ is a probability distribution. The choice of $\beta$ at $C$ determines the way $p$ marginalizes to each single measurement. This is given by condition (b). For example, if $\beta(C,x)=1$ then we have $p_x^0 = b+c$, but if $\beta(C,x)=0$ then $p_x^0=a+d=1-(b+c)$. In the notation of (b) the value $p_C(x)$ coincides with $p_x^0$; similarly for $y$ and $z$. 
This definition generalizes $\MP_\beta$; for example the $\beta$ choices given in Fig.~(\ref{fig:mermin-scenario-beta}) can be captured by the weights given in Fig.~(\ref{fig:K33-beta}).

\Lem{\label{lem:beta-rules}
Let $\beta'$ be an incidence weight on $K_{3,3}$ defined in the same way as $\beta$ except possible at
a single context as in (1) and (2), or at two contexts as in (3).
\begin{enumerate}[(1)]
\item There is a single context on which $\beta$ and $\beta'$ are defined as one of the following: 
$$
\includegraphics[width=.6\linewidth]{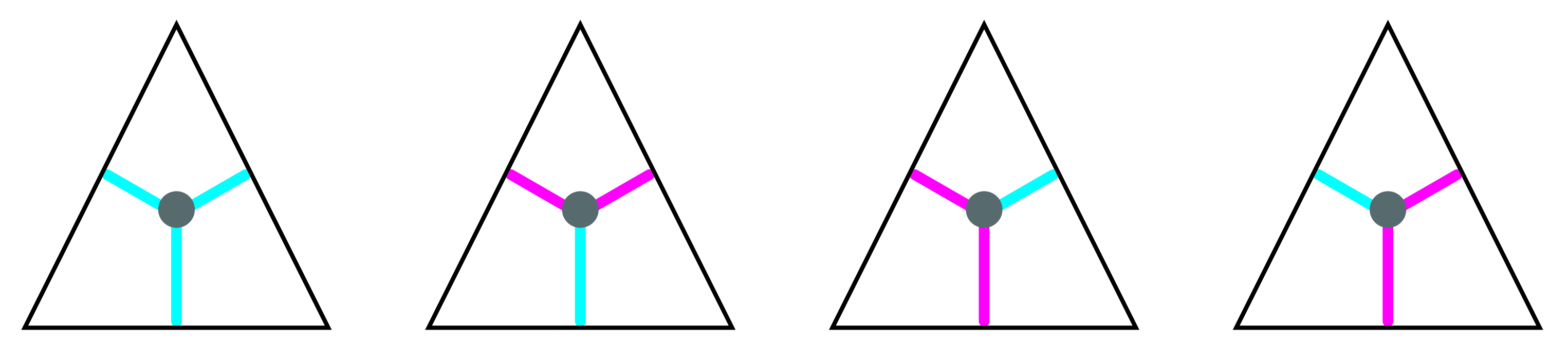}
$$ 
\item There is a single context on which $\beta$ and $\beta'$ are defined as one of the following:
$$
\includegraphics[width=.6\linewidth]{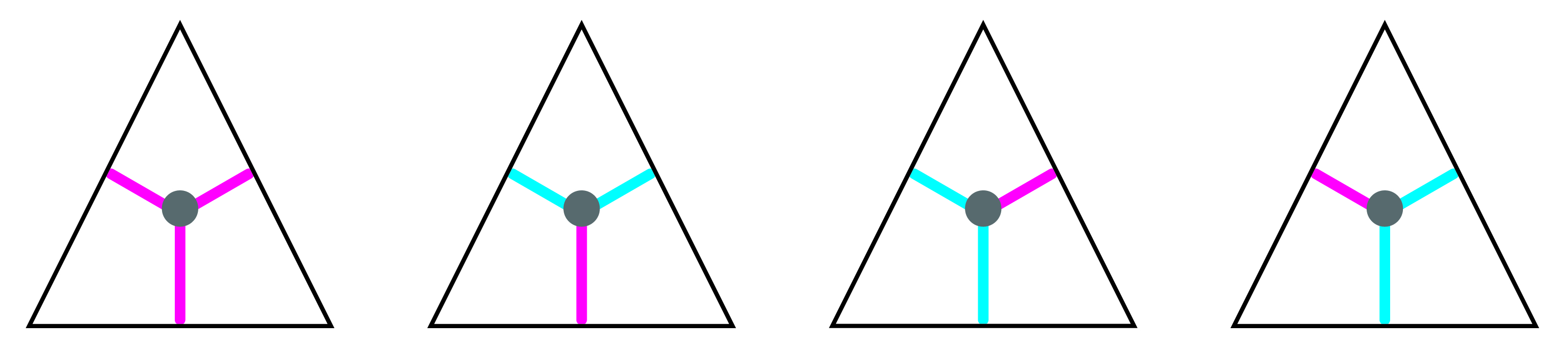}
$$ 
\item There are two contexts on which $\beta$ and $\beta'$ are defined as one of the following: 
$$
\includegraphics[width=.3\linewidth]{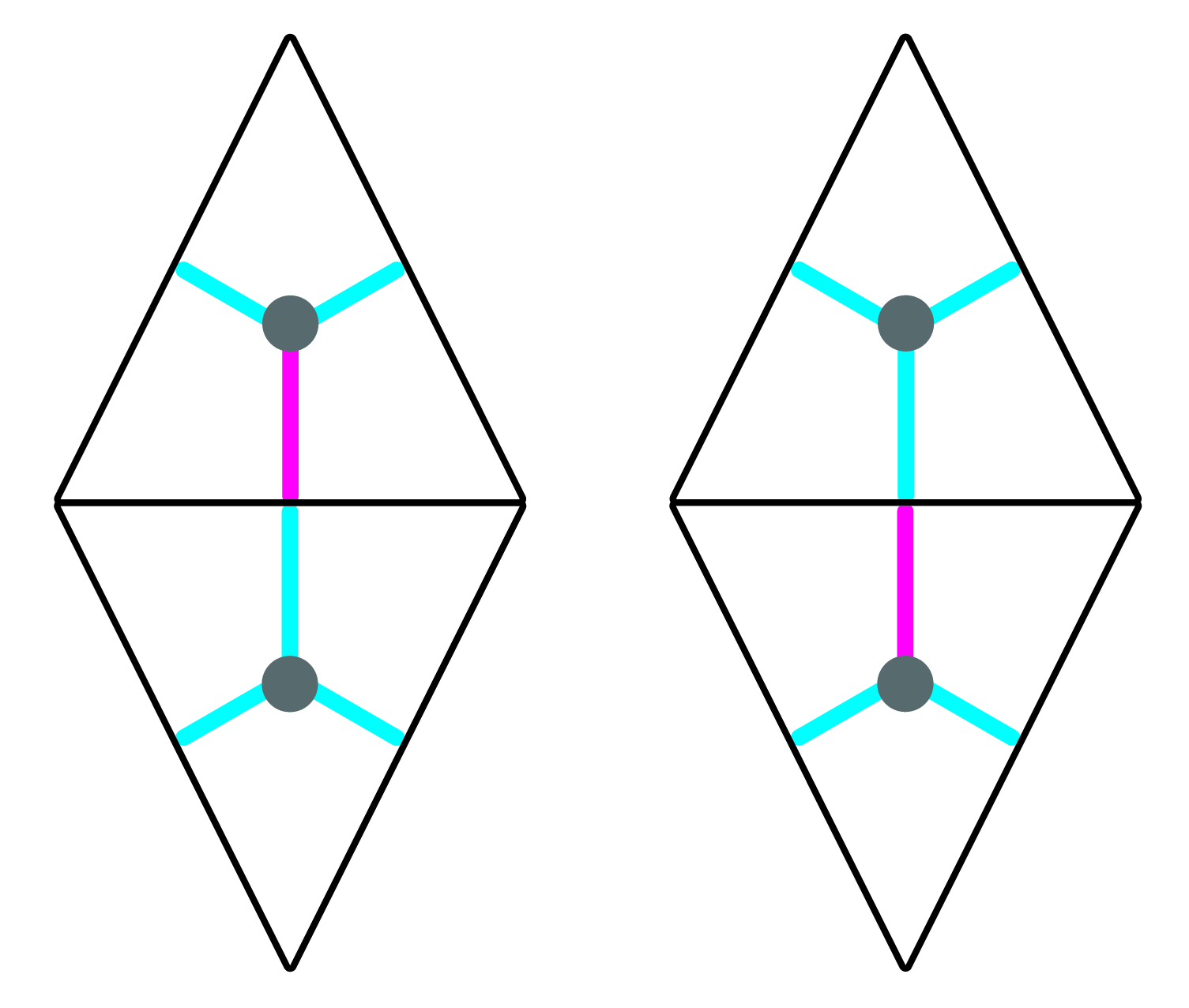}
$$ 
\end{enumerate}
At each case $\MP_\beta$ is combinatorially isomorphic to $\MP_{\beta'}$.
} 
\Proof{
The isomorphism is given by permuting the probability coordinates inside the contexts. Let $C=\set{x,y,z}$ denote the context in cases (1) and (2). (See Fig.~(\ref{fig:simple-scenarios}) for the labeling convention.) In case (1) we can obtain the isomorphism between the polytope corresponding to the first context and the next ones (from right to left) by flipping the outcome of $x$, $y$ and both $x$ and $y$; respectively.
Case (2) is similar.
In (3)  the constrained imposed at the common edge is the same in both cases, hence they specify the same polytope.
}

Proposition \ref{pro:MPbeta-cohomology} is a consequence of the following more general result.

\Pro{\label{pro:tildeMPbeta}
$\widetilde\MP_\beta$ is combinatorially isomorphic to $\widetilde\MP_{\beta'}$ if and only if 
$$
\sum_{(C,m)\in R}\beta(C,m)  = \sum_{(C,m)\in R}\beta'(C,m) \mod 2.
$$ 
}
\Proof{
The main idea of the proof is to use Lemma \ref{lem:beta-rules} to show that every case is either isomorphic to Fig.~(\ref{fig:K33-beta0}) or Fig.~(\ref{fig:K33-beta1}).
First we observe that applying the transformations in (1) and (2) in Lemma \ref{lem:beta-rules}
we can assume that at each context $C$ either $\beta(C,m)=0$ for all $m\in C$ or $\beta(C,m)=1$ for exactly one $m\in C$ and zero otherwise. Applying the transformation  (3)  we can assume that every context where $\beta$ is nonzero is adjacent. Furthermore, again using (3) we can cancel a pair of adjacent contexts with $\beta=1$ by first rotating  $\beta(C,m)=1$ once using (1) and then applying (3) to obtain a pair of contexts where on one of them $\beta=0$ and on the other there are two measurements for which $\beta(C,m)=1$. Using (1) the remaining context with two nonzero $\beta$'s can be replaced  by $\beta=0$. This procedure terminates either at Fig.~(\ref{fig:K33-beta0}), or, after successive application of the transformation in (3),  at Fig.~(\ref{fig:K33-beta1}).
}

\section{Proof of Proposition~\ref{pro:G0-G1}}\label{sec:proof-g0-g1}

Let $\Sigma_n$ denote the symmetric group on $n$ letters.

\Pro{\label{prop-pres}
The group presentation of $\Cl_1$ and $G_1$ are given as follows:
$$\Cl_1=\Span{H,S}\cong \langle h,s\suchthat h^2=s^4=(hs)^3=1\rangle\cong \Sigma_4,$$
where $s$ and $h$ correspond to $S$ and $H$, and
\begin{equation}\label{eq:G1-presentation}
G_1=\langle \Cl_1\times \Cl_1,\SWAP\rangle\cong \langle h,s,w\suchthat w^2=h^2=s^4=(hs)^3=[h,whw]=[h,wsw]=[s,whw]=[s,wsw]=1\rangle,
\end{equation}
where $s$, $h$ and $w$ correspond to $S\otimes \one$, $H\otimes \one$ and $\SWAP$; respectively. 
}
\Proof{The group $\langle h,s\suchthat h^2=s^4=(hs)^3=1\rangle$ is isomorphic to $\Sigma_4$ (cf. \cite[Theorem 8.1]{bray2011short}). In particular, it has order $24$. Sending $H\mapsto h$ and $S\mapsto s$ defines a group homomorphism from $\Cl_1$ since  $H^2=S^4=(HS)^3=\one$. 
By comparing the orders of the groups we see that this is an isomorphism.


By the first part, the presentation of $\Cl_1\times \Cl_1$ is given by
$$\Cl_1\times \Cl_1\cong \langle h_1,s_1,h_2,s_2\suchthat h_1^2=s_1^4=(h_1s_1)^3=h_2^2=s_2^4=(h_2s_2)^3=[h_1,h_2]=[h_1,s_2]=[s_1,h_2]=[s_1,s_2]=1\rangle$$
where we identify $h_1,s_1,h_2,s_2$ with $H\otimes\one,S\otimes\one,\one\otimes H,\one\otimes S$; respectively. The presentation of $G_1$ is obtained by adding one more generator namely $w$ (which identify with $\SWAP$) and add relations that correspond to the action of $w$ on $h_1,s_1,h_2$ and $s_2$. 
Thus we have
\begin{multline*}
G_1\cong \langle h_1,s_1,h_2,s_2,w\suchthat w^2=h_1^2=s_1^4=(h_1s_1)^3=h_2^2=s_2^4=(h_2s_2)^3=1\\
[h_1,h_2]=[h_1,s_2]=[s_1,h_2]=[s_1,s_2]=wh_1wh_2^{-1}=ws_1ws_2^{-1}=1\rangle
\end{multline*}
By relations $wh_1wh_2^{-1}=ws_1ws_2^{-1}=1$, we can remove the generators $h_2$ and $s_2$. Then we obtain
\begin{multline*}
G_1\cong \langle h_1,s_1,w\suchthat w^2=h_1^2=s_1^4=(h_1s_1)^3=(wh_1w)^2=(ws_1w)^4=(wh_1wws_1w)^3=1\\
[h_1,wh_1w]=[h_1,ws_1w]=[s_1,wh_1w]=[s_1,ws_1w]=1\rangle
\end{multline*}
Note that the relations $(wh_1w)^2=1$, $(ws_1w)^4=1$ and $(wh_1wws_1w)^3$ can be obtained from  $h_1^2=w^2=1$, $w^2=s_1^4=1$ and $w^2=(h_1s_1)^3=1$, respectively. Thus those three relations can be removed. Finally, we obtain
$$G_1=\langle \Cl_1\times \Cl_1,\text{SWAP}\rangle\cong \langle h,s,w\suchthat w^2=h^2=s^4=(hs)^3=[h,whw]=[h,wsw]=[s,whw]=[s,wsw]=1\rangle$$
where we identify $h,s,w$ with $H\otimes\one,S\otimes\one$ and $\SWAP$, respectively.
}

\begin{proof}[{\bf Proof of Proposition \ref{pro:G0-G1}}]

We will construct a function $\phi:G_1\to G_0$, show that it is a group homomorphism, and  makes the following diagram  commute:
$$
	\begin{tikzpicture}[node distance=2.3cm, auto]
	\node (0) {$0$};
	\node (1) [right of=0] {$\mathbb{Z}^4$};
	\node (2) [right of=1] {$G_1$};
	\node (3) [right of=2] {$p(G_1)\subset Sp_4(\mathbb{Z}_2$)};
	\node (4) [right of=3] {$1$};
	\node (5) [below of=0] {$0$};
	\node (6) [below of=1] {$G_l\cong\mathbb{Z}_2^4$};
	\node (7) [below of=2] {$G_0$};
	\node (8) [below of=3] {$\Aut(K_{3,3})$};
	\node (9) [below of=4] {$1$};
	\draw[->] (0) to node {}(1);
	\draw[->] (1) to node {$\iota$}(2);
	\draw[->] (2) to node {$p$}(3);
	\draw[->] (3) to node {}(4);
	\draw[->] (5) to node {}(6);
	\draw[->] (6) to node {$\iota'$}(7);
	\draw[->] (7) to node {$p'$}(8);
	\draw[->] (8) to node {}(9);
	\draw[->] (2) to node {$\phi$}(7);
	\draw[->] (1) to node {$f$}(6);
	\draw[->] (3) to node {$g$}(8);
	\end{tikzpicture}
$$
Since $G_1$ is a subgroup of $\Cl_2$, the top row of the group extension corresponds to decomposing $G_1$ into the symplectic part and the Pauli part. Define the following sets:
$$
\begin{aligned}
C_1&=\{Y\otimes X,X\otimes Y,Z\otimes Z\},\\
C_2&=\{X\otimes X,Y\otimes Y,Z\otimes Z\},\\
C_3&=\{X\otimes Z,Z\otimes X,Y\otimes Y\},\\
C_4&=\{X\otimes Y,Y\otimes Z,Z\otimes X\},\\
C_5&=\{X\otimes X,Y\otimes Z,Z\otimes Y\},\\
C_6&=\{X\otimes Z,Z\otimes Y,Y\otimes X\}.
\end{aligned}
$$
The group $\Aut(K_{3,3})$ permutes these sets, hence we think of it as a subgroup of $\Sigma_6$.
We define $f$ and $\phi$ as follows
$$
\begin{aligned}
    && \one \otimes X \mapsto & l_{6b} \\
f:  && \one \otimes Z \mapsto & l_{2b}\\
    && X\otimes \one \mapsto & l_{3b}\\
    && Z\otimes \one\mapsto & l_{4b}
\end{aligned}\;\;\; \text{ and }\;\;\;
\begin{aligned}
    && H\otimes \one \mapsto &  (l_{5b},(1\, 6) (2\,3) (4\,5) )\\
\phi:  && S\otimes \one \mapsto & (l_{3b},(1\,2) (3\,4) (5\,6) )\\
    && \SWAP \mapsto & (l_{0},(4\,6) )
\end{aligned}
$$
where we write $l_0$ for the trivial element of $G_l$.
Note that $\phi$ factors through $g$ and its surjective.
 It is clear that $f$ is an isomorphism. 
 It remains to show that $\phi$ is group homomorphism and the left square of the diagram commutes.
By Proposition \ref{prop-pres}, we know that the group presentation of $G_1$ is given by Eq.~(\ref{eq:G1-presentation}).
We show that $\phi$ is a group homomorphism by checking that it respects all the relations.

We will need the products $\phi(w)\phi(h)\phi(w)$ and $\phi(w)\phi(s)\phi(w)$:
$$\begin{aligned}
\phi(w)\phi(h)\phi(w)&=(l_{0},(46))(l_{5b},(16)(23)(45))(l_{0},(46))\\
&=(l_{1b},(1456)(23))(l_{0},(46))\\
&=(l_{1b},(14)(23)(56))
\end{aligned}$$
$$\begin{aligned}
\phi(w)\phi(s)\phi(w)&=((l_{0},(46))(l_{3b},(12)(34)(56))(l_{0},(46))\\
&=(l_{6b},(12)(3654))(l_{0},(46))\\
&=(l_{6b},(12)(36)(45))
\end{aligned}$$
We check the commutation relation $[\phi(h),\phi(w)\phi(h)\phi(w)]=1$:
$$\begin{aligned}
\phi(h)(\phi(w)\phi(h)\phi(w))&=(l_{5b},(16)(23)(45))(l_{1b},(14)(23)(56))\\
&=(l_{5b}+l_{1b},(15)(46))\\
&=(l_{6a},(15)(46))\\
\\
(\phi(w)\phi(h)\phi(w))\phi(h)&=(l_{1b},(14)(23)(56))(l_{5b},(16)(23)(45))\\
&=(l_{1b}+l_{5b},(15)(46))\\
&=(l_{6a},(15)(46))\\
\end{aligned}$$
The remaining commutation relations $[\phi(h),\phi(w)\phi(s)\phi(w)]=[\phi(s),\phi(w)\phi(h)\phi(w)]=[\phi(s),\phi(w)\phi(s)\phi(w)]=1$ can be checked similarly.
Next, we  check the remaining relations:
$$\begin{aligned}
(\phi(h)\phi(s))^3&=((l_{5b},(16)(23)(45))(l_{3b},(12)(34)(56)))^3\\
&=(l_{5b}+l_{4b},(16)(23)(45)(12)(34)(56))\\
&=(l_{3b},(135)(264))(l_{3b},(135)(264))^2\\
&=(l_{3b}+l_{5b},(153)(246))(l_{3b},(135)(264))\\
&=(l_{4b},(153)(246))(l_{3b},(135)(246))\\
&=(l_{4b}+l_{4b},())\\
&=(l_{0},())
\end{aligned}$$
The relations $\phi(w)^2=\phi(h)^2=\phi(s)^4=1$ can be checked similarly.

Finally, we need to check the left square commutes. First, we express all generators of $\mathbb{Z}_2^4\subset G_1$ using $H\otimes\one,S\otimes\one$ and $\SWAP$:
$$\begin{aligned}
X\otimes \one&=(H\otimes\one)(S\otimes\one)^2(H\otimes\one)\\
Z\otimes \one&=(S\otimes\one)^2\\
\one\otimes X&=\SWAP (X\otimes\one) \SWAP\\
\one\otimes Z&=\SWAP (Z\otimes\one) \SWAP
\end{aligned}$$
Then we calculate the image of each generator:
$$\begin{aligned}
\phi\circ\iota(Z\otimes\one) &= (l_{3b},(12)(34)(56))(l_{3b},(12)(34)(56))\\
&=(l_{3b}+l_{5b},())\\
&=(l_{4b},())\\
&=\iota'\circ f(Z\otimes\one)\\
\end{aligned}$$
and
$$\begin{aligned}
\phi\circ\iota(X\otimes\one)&= (l_{5b},(16)(23)(45))(l_{4b},())(l_{5b},(16)(23)(45))\\
&=(l_{5b}+l_{3b} ,(16)(23)(45))(l_{5b},(16)(23)(45))\\
&=(l_{4b},(16)(23)(45))(l_{5b},(16)(23)(45))\\
&=(l_{4b}+l_{5b} ,())\\
&=(l_{3b},())\\
&=\iota'\circ f(X\otimes\one).
\end{aligned}$$
We can verify $\phi\circ\iota(\one\otimes X)=\iota'\circ f(\one\otimes X)$ and $\phi\circ\iota(\one\otimes Z)=\iota'\circ f(\one\otimes Z)$ in a similar way.

\end{proof}

\section{Stabilizers of $\MP_1$ vertices}
\label{sec:Stabilizer}

\subsection{Stabilizers of type $1$ and $2$ vertices}

In this section we describe the stabilizers of the vertices of $\MP_1$ in the  group $G_1\subset \Cl_2$. Recall that $\Cl_2$ is the quotient of the normalizer of the Pauli group by the central subgroup. When we consider a unitary as an element of the Clifford group, we mean the equivalence class up to a scalar, even though this is not indicated in our notation for the sake of simplicity. 
For the computation of the stabilizers it suffices to choose a representative from each type of vertices. We choose $\VFS$ (type $1$) and $\VFe$ (type $2$).
For the description of the stabilizers we will need   
the dihedral group whose presentation is given as follows:
\begin{equation}\label{eq:dih}
D_{2n}=\langle a,b\suchthat a^n=b^2=(ba)^2=1\rangle.
\end{equation}

\begin{lem}{\label{lem:stab57}}
The stabilizer of $\VFS$ is given by
$$
\Stab_{G_1}(\VFS)=\langle Q,R \rangle\cong D_{24}
$$
where 
$Q = YS \otimes X$ and  $R = YH\otimes H$.
\end{lem}
\begin{proof}
Let $K=\langle Q,R\rangle$. 
It is straight-forward to verify that $K$ is contained in the stabilizer by explicitly checking that the vertex is fixed by $Q$ and $R$.
Hence $K\subset \Stab_{G_1}(\VFS)$. 
Since there are $48$ type $1$ vertices and $G_1$ acts transitively on them by Lemma \ref{lem:g1-transitive-vert} we have
$$
\left|\frac{G_1}{\Stab_{G_1}(\VFS)}\right| =48,
$$
which implies that $|\Stab_{G_1}(\VFS)|=24$. We finish the proof by showing that  $K\cong D_{24}$. 
Let $A=QR$ then one can verify  that 
$$
K = \Span{A,R \suchthat A^{12} = R^2 = (RA)^2=\one} \subset G_1.
$$
\end{proof}

\begin{lem}{\label{lem:stab58}}
The stabilizer of $\VFe$ is given by
$$
\Stab_{G_1}(\VFe)=\langle M,\SWAP\rangle\cong D_{16}
$$
where 
$M=X\otimes YS$.
\end{lem}
\begin{proof}
Proof is similar to Lemma \ref{lem:stab57}.
Let $L=\Span{M,\SWAP}$.
First one verifies that the given generators fix $\VFe$, which implies that $L$ is contained in the stabilizer. Transitivity of the action of $G_1$ on the set of type $2$ vertices (Lemma \ref{lem:g1-transitive-vert}) implies that $|\Stab_{G_1}(\VFe)|=16$. To conclude the proof we observe that 
$$L=\Span{N,\SWAP \suchthat N^6 = \SWAP^2 = (\SWAP N)^2=\one}\subset G_1,$$ 
where $N=M\,\SWAP$. 
Therefore $L\cong D_{16}$.  
\end{proof}

\subsection{Stabilizer action on the neighbors}
\begin{lem}{\label{dih}}
Consider the generator $a\in D_{2n}$ in the presentation of $D_{2n}$; see Eq.~(\ref{eq:dih}).
If $a^{n/2}\not\in G$, then either $D_{2n}\cap G=\{1\}$ or there exists a unique $i\in\{0,\cdots,n-1\}$ such that $D_{2n}\cap G=\langle a^ib\rangle\cong C_2$.
\end{lem}
\begin{proof}
Observe that any non-trivial subgroup of $\langle a\rangle$ will contains $a^{n/2}$. 
Since $a^{n/2}\not\in G$, it follows that $a^i\not\in G$ for all $i\in\{0,1,\cdots,n-1\}$ (otherwise $a^{n/2}\in \langle a^i\rangle\subset G$, which is a contradiction). Thus either $D_{2n}\cap G$ is trivial or $D_{2n}\cap G$ is generated by elements of form $a^ib$ where $i\in\{0,1,\cdots,n-1\}$. 
Let $g=a^ib$ and $h=a^jb$ be two distinct elements. 
We have $gh=a^{i-j}$, which is a non-trivial elements in $\langle a\rangle$. Thus either $a^ib\not\in G$ for all $i\in\{0,1,\cdots,n-1\}$, or there exists an unique $k\in\{0,1,\cdots,n-1\}$ such that $a^k b\in G$. 
This proves the statement.
\end{proof}

We will consider the following type $2$ neighbors of $\VFe$: $\Vtt$ given in Eq.~(\ref{eq:V22}), $\Vnn$ in Eq.~(\ref{eq:V99}) and $\Vte$ in Eq.~(\ref{eq:V28}).

 
\begin{lem}{\label{intersection}}
Let $N=M\,\SWAP$ where $M=X\otimes YS$. We have
\begin{enumerate}
\item $\Stab_{G_1}(\VFe)\cap \Stab_{G_1}(\Vtt)=\langle \SWAP\rangle\cong C_2$,
\item $\Stab_{G_1}(\VFe)\cap \Stab_{G_1}(\Vte)=\langle N\SWAP\rangle\cong C_2$,
\item $\Stab_{G_1}(\VFe)\cap \Stab_{G_1}(\VFS)=\langle N^{-1}\SWAP\rangle\cong C_2$,
\item $\Stab_{G_1}(\VFe)\cap \Stab_{G_1}(\Vnn)=\langle N\SWAP\rangle\cong C_2$.
\end{enumerate}
\end{lem}
\begin{proof}
The table below shows the action of $N,N^4,N^{-1},\SWAP,N \SWAP$ and $N^{-1}\SWAP$ on the non-local Pauli operators. For simplicity we omit the tensor product notation.

 \begin{table}[h!] 
\centering
\begin{tabular}{|l|l|l|l|l|l|l|l|l|l|}
\hline
$A$: non-local Pauli & XX & XY & XZ & YX & YY & YZ & ZX & ZY & ZZ \\ \hline
$N A N^\dagger$ & XY & -YY & -ZY & XX  & -YX  & -ZX  & -XZ  & YZ  & ZZ  \\ \hline
$N^4 A (N^4)^\dagger$ & XX & XY & -XZ & YX  & YY  & -YZ & -ZX  & -ZY  & ZZ  \\ \hline
$N^\dagger  A N$ & YX & XX & -ZX & -YY  & -XY  & ZY  & -YZ  & -XZ  & ZZ  \\ \hline
$(\SWAP)  A (\SWAP)^\dagger$ & XX & YX & ZX & XY  & YY  & ZY & XZ  & YZ  & ZZ  \\ \hline
$(N\SWAP) A (N\SWAP)^\dagger$ & XY & XX & -XZ & -YY  & -YX  & YZ & -ZY  & -ZX  & ZZ  \\ \hline
$(N^\dagger \SWAP) A (N^\dagger \SWAP)^\dagger$ & YX & -YY & -YZ & XX  & -XY  & -XZ & -ZX  & ZY  & ZZ  \\ \hline
\end{tabular}
\caption{The action of some  unitaries in $\Stab_{G_1}(\VFe)$. 
}\label{tab:A-nonlocalPauli}
\end{table} 

Using table we can show that 
$N^4$ does not fix $\Vtt, \VFS$, $\Vte$ and $\Vnn$. 
On the other hand, $\SWAP$, $N^{-1}\SWAP$, and $N (\SWAP)$ fixes the vertices $\Vtt$, $\VFS$,  and $\Vte$ respectively, and $N(\SWAP)$ fixes $\Vnn$.
Then the statement follows from Lemma \ref{dih}.
\end{proof}

\begin{lem}\label{lem:G1-transitive-V57}
$Stab_{G_1}(\VFS)$ acts transitively on the set of neighbor vertices of $\VFS$.
\end{lem}
\begin{proof}
By Lemma \ref{intersection}, we have $Stab_{G_1}(\VFS)\cap  Stab_{G_1}(\VFe)\cong C_2$.  
Then the orbit of $\VFS$ under the $Stab_{G_1}(\VFS)$ action 
has $|D_{24}|/|C_2|=12$ elements, which is the whole set of neighbors of $\VFS$.
\end{proof}

\begin{lem}\label{lem:G1-transitive-V58}
The action of $Stab_{G_1}(\VFe)$ on the set of neighbor vertices of $\VFe$ breaks into three orbits with representatives given by $\VFS$ (type $1$), $\Vtt$ and $\Vnn$ (both type $2$).
\end{lem}

\begin{proof}
By Lemma \ref{intersection}, we have $Stab_{G_1}(\VFe)\cap  Stab_{G_1}(\VFS)\cong C_2$. 
Then the orbit of $\VFe$ under the action of the stabilizer 
has $|D_{16}|/|C_2|=8$ elements, which is the whole set of type $1$ neighbors of $\VFe$.
By Lemma \ref{intersection}, we have
\[Stab_{G_1}(\VFe)\cap Stab_{G_1}(\Vnn)\cong Stab_{G_1}(\VFe) \cap Stab_{G_1}(\Vtt)\cong C_2.\]
Since there are $16$ type $2$ neighbors of $\VFe$, the orbit of $Stab_{G_1}(\VFe)$  on $\Vtt$ and $\Vnn$ both have size equal to $8$. It remains to check that these orbits are distinct. For this we compute the orbit:
\begin{align*}
N  \Vnn N^\dagger&=\frac{1}{4}(\one+X\otimes  Y-Y\otimes  Y+Z\otimes  X-X\otimes  Z-Y\otimes Z)=\VTz\\
(N^2)  \Vnn (N^2)^\dagger &=\frac{1}{4}(\one-Y\otimes Y+Y\otimes X-X\otimes Z+Z\otimes Y+Z\otimes X)=\Vfn\\
(N^3)  \Vnn (N^3)^\dagger&=\frac{1}{4}(\one+Y\otimes X+X\otimes X+Z\otimes Y+Y\otimes Z-X\otimes Z)=\VSn\\
(N^4)  \Vnn (N^4)^\dagger&=\frac{1}{4}(\one+X\otimes X+X\otimes Y+Y\otimes Z-Z\otimes X+Z\otimes Y)=\Vee\\
(N^5)  \Vnn (N^5)^\dagger&=\frac{1}{4}(\one+X\otimes Y-Y\otimes Y-Z\otimes X+X\otimes Z+Y\otimes Z)=\VfT\\
(N^6)  \Vnn (N^6)^\dagger&=\frac{1}{4}(\one-Y\otimes Y+Y\otimes X+X\otimes Z-Z\otimes Y-Z\otimes X)=\VtT\\
(N^7)  \Vnn (N^7)^\dagger&=\frac{1}{4}(\one+Y\otimes X+X\otimes X-Z\otimes Y-Y\otimes Z+X\otimes Z)=\Vne
\end{align*}
where $N=(X\otimes YS)(\SWAP)$. 
Observe that $\Vtt$ does not belong to the orbit of $\Vnn$. Thus, the orbits of $\Vtt$ and $\Vnn$ are distinct. 
\end{proof}



We apply the stabilizer computation to show that 
the  type $1$ vertices in 
Fig.~(\ref{fig:T1NoNofT2})
are not neighbors.

\begin{figure}
\centering
\includegraphics[width=0.6\linewidth]{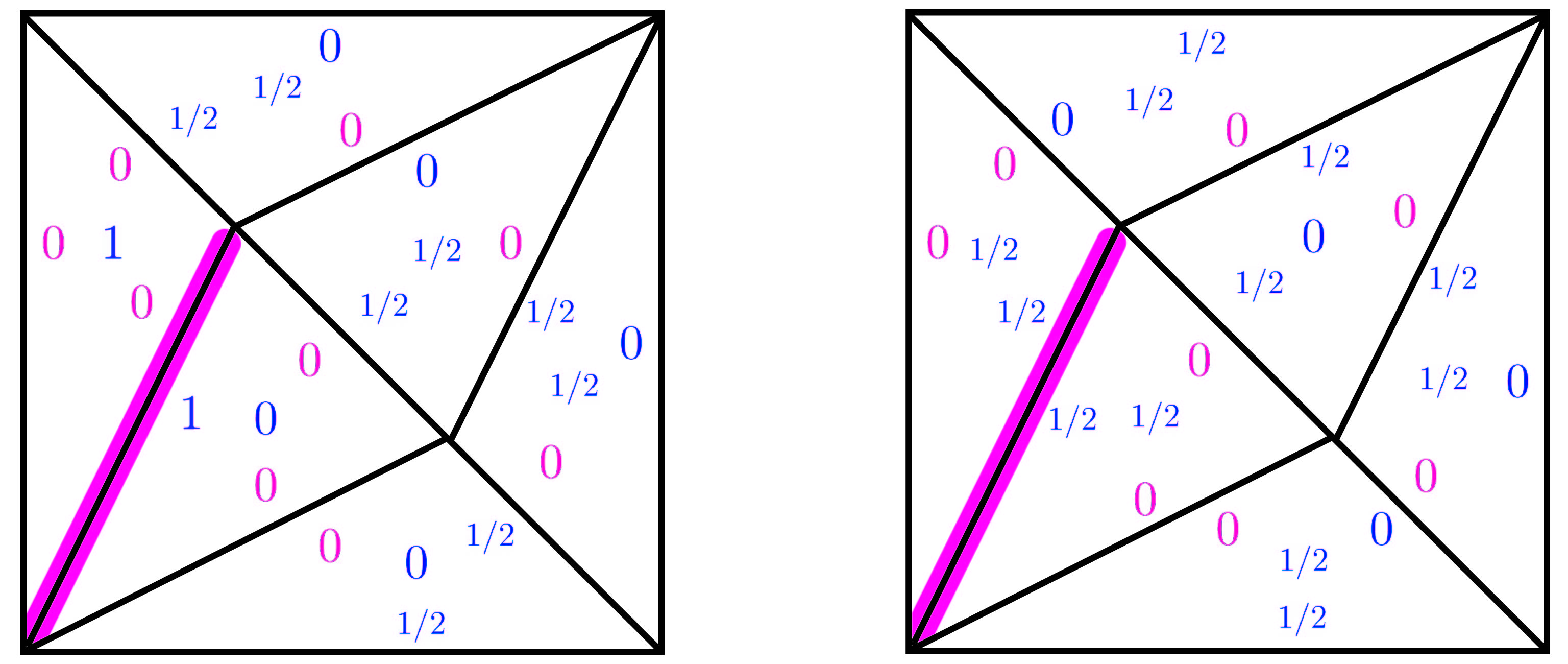}
\caption{We compare the canonical vertex {$p_{0}$} and {\Vte}. The $8$ overlapping zeros are colored in pink. The $4$ zeros on either side of the pink edge cannot all be linearly independent by Lemma~\ref{lem:rank-two-triangles}; see Eq.~(\ref{eq:rank-one-det}).
}
\label{fig:58-overlap-28}
\end{figure}

\Lem{\label{lem:not-neigh}
The vertices in Fig~(\ref{fig:T1NoNofT2}) are not neighbors of $\VFe$.
}
\Proof{
Consider the vertex $\Vte$ given in Eq.~(\ref{eq:V28}) from the list of vertices in Fig~(\ref{fig:T1NoNofT2}). 
By Lemma \ref{intersection} part (2), we have $Stab_{G_1}(\VFe)\cap  Stab_{G_1}(\Vte)\cong C_2$.  
Then the orbit of $\Vte$ under the $\Stab_{G_1}(\VFe)$ action 
has $|D_{16}|/|C_2|=8$ elements since $\Stab_{G_1}(\VFe)\cong D_{16}$ by Lemma \ref{lem:stab58}. This covers the whole set of vertices in Fig~(\ref{fig:T1NoNofT2}).

As discussed in Section \ref{sec:GraphMP1}, 
for two distributions $q_{1}$ and $q_{2}$ to be neighbors they must share $8$ linearly independent tight inequalities.
Let us consider $q_1=\VFe$ and the type $1$ vertex $q_2=\Vte$, and compare the number of overlapping zeros; see Fig.~(\ref{fig:58-overlap-28}).
There are precisely $8$ such zeros. However, by Lemma~\ref{lem:rank-two-triangles}, the two adjacent triangles on either side of the shaded edge cannot have rank $4$, thus the overlapping zeros  have rank $<8$.
Therefore $\VFe$ and $\Vte$ cannot be neighbors. Transitive action of $\Stab_{G_{1}}(\VFe)$ on the set of Fig.~(\ref{fig:T1NoNofT2}) implies that this holds when $q_2$ is one of the other vertices listed in Fig~(\ref{fig:T1NoNofT2}) as well.
}

 \begin{table}[h]
	\begin{subtable}[h]{0.45\textwidth}
        		\centering
		\begin{tabular}{|l|l|l|l|l|l|l|l|l|l|}
			\hline
			$q \in N_1(\VFe)$ & $U\in \Stab_{G_1}(\VFe)$ \\ \hline
			$\VFS$ & $\one,N^7\SWAP$ \\ \hline
			$\VtS$ & $N^4,N^3\SWAP$ \\ \hline
			$\Vooe$ & $N^7,N^6\SWAP$ \\ \hline
			$\Vnf$ & $N^3,N^2\SWAP$ \\ \hline
			$\Voot$ & $N^6,N^5\SWAP$ \\ \hline
			$\VSt$ & $N^2,N\SWAP$ \\ \hline
			$\VFt$ & $N,\SWAP$ \\ \hline
			$\VTo$ & $N^5,N^4\SWAP$ \\ \hline
		\end{tabular}
		\caption{ The action of $\Stab_{G_1}(\VFe)$ on the type $1$ neighbors of $\VFe$. See Lemma \ref{intersection}. The left column are $q\in N_1(\VFe)$, type 1 neighbours of $p_0$. The right column are elements $U\in\Stab_{G_1}(p_0)$ such that $U \VFS U^\dagger=q$.}
		\label{tab:V58-type1-neighbors}
   	 \end{subtable}
\hfill
	\begin{subtable}[h]{0.45\textwidth}
		\centering
		\begin{tabular}{|l|l|l|l|l|l|l|l|l|l|}
			\hline
			$q$ in Fig~(\ref{fig:T1NoNofT2}) & $U\in \Stab_{G_1}(\VFe)$ \\ \hline
			$\Vte$ & $\one,N(\SWAP)$ \\ \hline
			$\Vfs$& $N^2,N^3\SWAP$ \\ \hline
			$\Vtf$ & $N^6,N^7\SWAP$ \\ \hline
			$\VfS$ & $N^4,N^5\SWAP$ \\ \hline
			$\VTf$ & $N^7,\SWAP$ \\ \hline
			$\Vff$ & $N^5,N^6\SWAP$ \\ \hline
			$\VTF$ & $N,N^2\SWAP$ \\ \hline
			$\VFs$ & $N^3,N^4\SWAP$ \\ \hline
		\end{tabular}
	\caption{ The action of $\Stab_{G_1}(\VFe)$ on the vertices  in Fig~(\ref{fig:T1NoNofT2}). See Lemma \ref{intersection}. The left column are $q$, which are vertices  in Fig~(\ref{fig:T1NoNofT2}). The right column are elements $U\in\Stab_{G_1}(p_0)$ such that $U \Vte U^\dagger=q$}
	\label{tab:V58-type1-non-neighbors}
	\end{subtable}
\newline
\newline
\newline
\newline
	\begin{subtable}[h]{0.45\textwidth}
	\centering
	\scalebox{0.8}{
	\begin{tabular}{|l|l|l|l|l|l|l|l|l|l|}
		\hline
		$p \in N_2(\VFe)$  & $U\in \Stab_{G_1}(\VFe)$ & $U\in \Stab_{G_1}(\VFe)$\\ \hline
		$\Vnn$ & $\one,N\SWAP$ & \text{None} \\ \hline
		$\Vee$ & $N^4,N^5\SWAP$ & \text{None}  \\ \hline
		$\VfT$ & $N^5,N^6\SWAP$ & \text{None}  \\ \hline
		$\VTz$ & $N,N^2\SWAP$ & \text{None}  \\ \hline
		$\Vfn$ & $N^2,N^3\SWAP$ & \text{None}  \\ \hline
		$\VtT$ & $N^6,N^7\SWAP$ & \text{None}  \\ \hline
		$\Vne$ & $N^7,\SWAP$ & \text{None}  \\ \hline
		$\VSn$ & $N^3,N^4\SWAP$ & \text{None}  \\ \hline
		$\Vtt$ & \text{None}  & $\one,\SWAP$ \\ \hline
		$\Vfe$ & \text{None}  & $N^4,N^4\SWAP$ \\ \hline
		$\VFT$ & \text{None}  & $N^2,N^2\SWAP$ \\ \hline
		$\VfF$ & \text{None}  & $N^6,N^6\SWAP$ \\ \hline
		$\Vts$ & \text{None}  & $N,N\SWAP$ \\ \hline
		$\VFz$ & \text{None}  & $N^5,N^5\SWAP$ \\ \hline
		$\VTt$ & \text{None}  & $N^7,N^7\SWAP$ \\ \hline
		$\VFo$ & \text{None}  & $N^3,N^3\SWAP$ \\ \hline
	\end{tabular}
	}
	\caption{ The action of $\Stab_{G_1}(\VFe)$ on the type $2$ neighbors of $\VFe$. See Lemma \ref{intersection}. The left column are $p\in N_2(\VFe)$, vertices of type 2 neighbours of $p_0$. The middle column are elements $U\in \Stab_{G_1}(\VFe)$ such that $U \Vnn U^\dagger=p$. The right column are elements $U\in \Stab_{G_1}(\VFe)$ such that $U \Vtt U^\dagger=p$}
	\label{tab:V58-type2-neighbors}
	\end{subtable}
\hfill
	\begin{subtable}[h]{0.45\textwidth}
		\centering
		\begin{tabular}{|l|l|l|l|l|l|l|l|l|l|}
			\hline
			$p \in N(\VFS)$   & $U\in \Stab_{G_1}(\VFS)$ \\ \hline
			$\VFe$ & $\one,AR$ \\ \hline
			$\sVtF$ & $A^6,A^7R$ \\ \hline
			$\sVfn$ & $A^2,A^3R$ \\ \hline
			$\sVFn$ & $A^8,A^9R$ \\ \hline
			$\sVo$ & $A^{11},R$ \\ \hline
			$\sVfe$ & $A^5,A^6R$ \\ \hline
			$\sVee$ & $A^{10},A^{11}R$ \\ \hline
			$\sVs$ & $A^4,A^5R$ \\ \hline
			$\sVfT$ & $A^3,A^4R$ \\ \hline
			$\sVTz$ & $A^9,A^{10}R$ \\ \hline
			$\sVFo$ & $A^7,A^8R$ \\ \hline
			$\sVozt$ & $A,A^2R$ \\ \hline
		\end{tabular}
		\caption{The action of $\Stab_{G_1}(\VFS)$ on the neighbors of $\VFS$. See Lemma \ref{lem:stab57}. The left column are vertex $p\in N(\VFS)$, neighbors of $\VFS$. The right column are elements of $U\in \Stab_{G_1}(\VFS)$ such that $U \VFe U^\dagger=p$.}
		\label{tab:V57-neighbors}
	\end{subtable}
	\caption{The action of $Stab_{G_1}(q_0)$ on the neighbours of $q_0$ and the action of $\Stab_{G_1}(p_0)$ on type $1$ and type $2$ neighbours of $p_0$ and vertices in Fig~(\ref{fig:T1NoNofT2}); respectively.}
\end{table}

\end{document}